\documentclass[aps,prb,notitlepage,showkeys,showpacs,superscriptaddress,nofootinbib,longbibliography,preprintnumbers]{revtex4-2}

\usepackage[utf8]{inputenc}

\usepackage{here}

\usepackage[top=30truemm,bottom=30truemm,left=25truemm,right=25truemm]{geometry}

\usepackage{spectralsequences}
\usepackage{tikz}

\usepackage[all]{xy}

\usepackage[colorlinks=true,linktoc=page,citecolor=red,linkcolor=blue]{hyperref}	
\usepackage[whole]{bxcjkjatype} 

\usepackage{slashed}

\usepackage{amsthm}
\usepackage{thmtools}
\usepackage{blindtext}
\usepackage{braket}

\declaretheoremstyle[
       shaded={bgcolor=\color{rgb}{0.9,0.9,0.9}}  
]{theorem}
\declaretheorem[style=theorem]{theorem}

\declaretheoremstyle[
       shaded={bgcolor=\color{rgb}{0.9,0.9,0.9}}
]{question}

\declaretheoremstyle[
       shaded={bgcolor=\color{rgb}{0.9,0.9,0.9}}  
]{remark}

\declaretheoremstyle[
       shaded={bgcolor=\color{rgb}{0.9,0.9,0.9}}  
]{proposition}
\declaretheorem[style=theorem]{proposition}

\declaretheoremstyle[
       shaded={bgcolor=\color{rgb}{0.9,0.9,0.9}}  
]{definition}
\declaretheorem[style=theorem]{definition}

\declaretheoremstyle[
       shaded={bgcolor=\color{rgb}{0.9,0.9,0.9}}  
]{assumption}

\declaretheoremstyle[
       shaded={bgcolor=\color{rgb}{0.9,0.9,0.9}}  
]{conjecture}

\declaretheoremstyle[
       shaded={bgcolor=\color{rgb}{0.9,0.9,0.9}}  
]{corrorary}

\declaretheoremstyle[
       shaded={bgcolor=\color{rgb}{0.9,0.9,0.9}}  
]{axiom}

\declaretheoremstyle[
       shaded={bgcolor=\color{rgb}{0.9,0.9,0.9}}  
]{lemma}
\declaretheorem[style=theorem]{lemma}

\usepackage{mathrsfs}
\usepackage{amsmath,amssymb}

\usepackage{esvect}

\newcommand{\MPS}{\mathrm{MPS}}

\newcommand{\zmodtwo }{ \mathbb{Z}/2\mathbb{Z}}
\newcommand{\tr}{\mathrm{tr}}
\newcommand{\sumodd}{\sum_{\{i_{k}\},{\rm odd}}}

\newcommand{\hyphen}{\mathchar`-}

\newcommand{\rp}{\mathbb{R}\mathrm{P}}
\newcommand{\cohoZ}[2]{\mathrm{H}^{#1}\left(#2;\mathbb{Z}\right)}
\newcommand{\cohoZmodtwo}[2]{\mathrm{H}^{#1}\left(#2;\zmodtwo\right)}

\newcommand{\abs}[1]{\left|#1\right|}

\newcommand{\coho}[3]{\mathrm{H}^{#1}\left(#2;#3\right)}

\newcommand{\majosite}[2]{\underset{#1}{\hspace{.4em}\bullet}\hspace{-.9em}-\hspace{-.9em}-\hspace{-.8em}-\hspace{-.5em}-\hspace{-.8em}>\hspace{-.7em}-\hspace{-.5em}-\hspace{-.75em}-\hspace{-.5em}\underset{#2}{\bullet}}
\newcommand{\majositewide}[2]{\underset{#1}{\bullet}\hspace{-.5em}-\hspace{-.5em}-\hspace{-.5em}-\hspace{-.8em}>\hspace{-.7em}-\hspace{-.5em}-\hspace{-.75em}-\hspace{-.5em}\underset{#2}{\bullet}}
\newcommand{\majositenarrow}[2]{\underset{#1}{\bullet}\hspace{-1em}-\hspace{-.8em}-\hspace{-.5em}-\hspace{-.8em}>\hspace{-.7em}-\hspace{-.5em}-\hspace{-1em}\underset{#2}{\bullet}}

\newcommand{\contract}[3]{\hspace{-1.3em}\left|\right.\hspace{-1.5mm}\raisebox{1.3ex}[0ex][0mm]{$\overline{\hspace{#1}}$}\hspace{-1mm}\left|\right.\raisebox{1.1ex}[0ex][0mm]{$\hspace{#2}<\hspace{#3}$}}

\newcommand{\zmod}[1]{\mathbb{Z}/#1\mathbb{Z}}

\newcommand{\grayrect}{\colorbox[gray]{0.8}{\textcolor[gray]{0.8}{j}}}

\usepackage{titletoc}

\usepackage{xcolor}
\usepackage[normalem]{ulem}

\newcommand{\Z}{\mathbb{Z}}
\usepackage{bm}

\begin{document}

\title{Generalized Thouless Pumps \\in $1+1$-dimensional Interacting Fermionic Systems}
\author{Shuhei Ohyama}
\email{shuhei.oyama@yukawa.kyoto-u.ac.jp}
 \affiliation{Yukawa Institute for Theoretical Physics, Kyoto University, Kyoto 606-8502, Japan}
\author{Ken Shiozaki}
 \email{ken.shiozaki@yukawa.kyoto-u.ac.jp}
\affiliation{Yukawa Institute for Theoretical Physics, Kyoto University, Kyoto 606-8502, Japan}
\author{Masatoshi Sato}
 \email{msato@yukawa.kyoto-u.ac.jp}
\affiliation{Yukawa Institute for Theoretical Physics, Kyoto University, Kyoto 606-8502, Japan}

\date{\today} 
\preprint{YITP-22-43}

\begin{abstract}
The Thouless pump is a phenomenon in which $\mathrm{U}(1)$ charges are pumped from an edge of a fermionic system to another edge. The Thouless pump has been generalized in various dimensions and for various charges. In this paper, we investigate the generalized Thouless pumps of fermion parity in both trivial and non-trivial phases of $1+1$-dimensional interacting fermionic short-range entangled (SRE) states. For this purpose, we use matrix product states (MPSs). MPSs describe many-body systems in $1+1$-dimensions and can characterize SRE states algebraically. We prove fundamental theorems for fermionic MPSs (fMPSs) and use them to investigate the generalized Thouless pumps. We construct non-trivial pumps in both the trivial and non-trivial phases and we show the stability of the pumps against interactions. Furthermore, we define topological invariants for the generalized Thouless pumps in terms of fMPSs and establish consistency with existing results. These are invariants of the family of SRE states that are not captured by the higher dimensional Berry curvature. We also argue a relation between the topological invariants of the generalized Thouless pump and the twist of the $K$-theory in the Donovan-Karoubi formulation.
\end{abstract}

\maketitle

\setcounter{tocdepth}{3}
\tableofcontents

\section{Introduction}\label{intro}

\subsection{Kitaev's Argument and the Kitaev Pump}\label{intro:Kitaevcanno}
A short-range entangled (SRE) state is a unique gapped ground state for a system without a  boundary.
An integer quantum Hall state is a representative example of $2+1$-dimensional SRE states. 
To date, SRE states with various symmetries in various dimensions have been discovered, which include physically important systems such as topological insulators and topological superconductors~\cite{GW09,RevModPhys.82.3045,RevModPhys.83.1057}.

A remarkable property of SRE states is invertibility: Any SRE state $\ket{\chi}$ in $d+1$ space-time dimensions has an anti-SRE state $\ket{\bar{\chi}}$ that satisfies
\begin{eqnarray}\label{deform}
\ket{\chi}_{d+1}\otimes\ket{\bar{\chi}}_{d+1}\sim\ket{0}_{d+1}\otimes\ket{0}_{d+1}\sim\ket{\bar{\chi}}_{d+1}\otimes\ket{\chi}_{d+1}.
\end{eqnarray}
Here $\sim$ represents a continuous deformation keeping a gap and symmetry, and $\ket{0}_{d+1}$ is the trivial state in $d+1$-dimensions \footnote{We often omit the space-time dimension when it is clear from the context.}. Because of the invertibility, SRE states are often called as invertible states.

A fundamental question for SRE states is what a kind of quantum phases and the related phenomena they deliver for fixed space-time dimensions and symmetry. Let $\mathcal{M}_{d+1}^{G}$ be the set of all $d+1$-dimensional SRE states with symmetry $G$. 
The first step to answer the question is the identification of connected components $\pi_{0}\left(\mathcal{M}_{d+1}^{G}\right)$ 
of $\mathcal{M}_{d+1}^{G}$: Each connected component in $\pi_{0}\left(\mathcal{M}_{d+1}^{G}\right)$ specifies a possible symmetry protected topological (SPT) phase (See Fig.~\ref{fig:variouspump}).

Importantly, we can also consider a more complicated topology in $\mathcal{M}_{d+1}^{G}$, 
such as the fundamental group $\pi_{1}\left(\mathcal{M}_{d+1}^{G}\right)$.
In associated with this, 
Kitaev considered a loop in $\mathcal{M}_{d+1}^{G}$ that gives a non-trivial topological phenomenon \cite{Kitaev13}. 
Following his argument, let us start with a $d+1$-dimensional trivial state $\ket{0}_{d+1}$ obtained by arranging the $(d-1)+1$-dimensional trivial states $\ket{0}_{(d-1)+1}$ in a line.
\begin{eqnarray*}
\ket{0}\hspace{0.4mm}\ket{0}\hspace{0.4mm}\ket{0}\hspace{0.4mm}\ket{0}\hspace{0.4mm}&\cdots&\ket{0}\hspace{0.4mm}\ket{0}\hspace{0.4mm}\ket{0}\hspace{0.4mm}\ket{0}\hspace{0.4mm}. 
\end{eqnarray*}
Then, choosing an arbitrary $(d-1)+1$-dimensional SRE state $\ket{\chi}$, we perform the deformation $\ket{0}\ket{0}\sim\ket{\chi}\ket{\bar{\chi}}$ 
on neighboring trivial states:
\begin{eqnarray*}
\ket{\chi}\ket{\bar{\chi}}\ket{\chi}\ket{\bar{\chi}}\hspace{1mm}&\cdots&\hspace{0.5mm}\ket{\chi}\ket{\bar{\chi}}\ket{\chi}\ket{\bar{\chi}}\\
\abs{\grayrect}\quad\abs{\grayrect}\hspace{3mm}&\cdots&\hspace{2.3mm}\abs{\grayrect}\quad\abs{\grayrect}\\
\ket{0}\hspace{0.4mm}\ket{0}\hspace{0.4mm}\ket{0}\hspace{0.4mm}\ket{0}\hspace{1mm}&\cdots&\hspace{1mm}\ket{0}\hspace{0.4mm}\ket{0}\hspace{0.4mm}\ket{0}\hspace{0.4mm}\ket{0}\hspace{0.4mm},
\end{eqnarray*}
where $\abs{\grayrect}$ is the continuous deformation in Eq.(\ref{deform}). Finally, by accomplishing the reverse transformation $\ket{\bar{\chi}}\ket{\chi}\sim\ket{0}\ket{0}
$ 
for neighboring states shifted by one site, we obtain again the $d+1$-dimensional trivial state. 
This process defines a loop in $\mathcal{M}_{d+1}^{G}$ that starts from the $d+1$-dimensional trivial state and returns to itself, and interestingly, if the system has a boundary,  the same process pumps $(d-1)+1$-dimensional SRE states at the boundary, as shown below.
\begin{eqnarray*}
\ket{\chi}\hspace{0.4mm}\ket{0}\hspace{0.4mm}\ket{0}\hspace{0.4mm}\ket{0}\hspace{2.5mm}&\cdots&\hspace{3mm}\ket{0}\hspace{0.4mm}\ket{0}\hspace{0.4mm}\ket{0}\hspace{0.4mm}\ket{\bar{\chi}}\hspace{0.4mm}\\
\abs{\grayrect}\quad\abs{\grayrect}\hspace{-0mm}&\cdots&\hspace{-0mm}\abs{\grayrect}\quad\abs{\grayrect}\\
\ket{\chi}\ket{\bar{\chi}}\ket{\chi}\ket{\bar{\chi}}\hspace{2.5mm}&\cdots&\hspace{2.5mm}\ket{\chi}\ket{\bar{\chi}}\ket{\chi}\ket{\bar{\chi}}\\
\abs{\grayrect}\hspace{3.5mm}\abs{\grayrect}\hspace{5mm}&\cdots&\hspace{4.3mm}\abs{\grayrect}\hspace{3.5mm}\abs{\grayrect}\\
\ket{0}\hspace{0.4mm}\ket{0}\hspace{0.4mm}\ket{0}\hspace{0.4mm}\ket{0}\hspace{3.3mm}&\cdots&\hspace{2.8mm}\ket{0}\hspace{0.4mm}\ket{0}\hspace{0.4mm}\ket{0}\hspace{0.4mm}\ket{0}\hspace{0.4mm}. 
\end{eqnarray*}
This is a generalization of the Thouless pump that pumps the $\mathrm{U}(1)$ charge by a periodic change of a potential~\cite{PhysRevB.27.6083}. 
This Kitaev's  pump \footnote{In the following, we call this pump as Kitaev's canonical pump.} applies to any SRE states in arbitrary dimensions with any symmetry. 

Whereas the above procedure only provides an injective map from a $(d-1)+1$-dimensional SRE state to a loop of $d+1$-dimensional SRE states, Kitaev conjectured that this correspondence is one-to-one (up to homotopy), 
\begin{eqnarray}
\mathcal{M}_{(d-1)+1}^{G}\sim\Omega\mathcal{M}_{d+1}^{G},
\end{eqnarray}
where $\Omega\mathcal{M}_{d+1}^{G}$ is the based loop space of $\mathcal{M}_{d+1}^{G}$ :
\begin{eqnarray}
\Omega\mathcal{M}_{d+1}^{G}:=\{\gamma:\left[0,1\right]\to\mathcal{M}_{d+1}^{G}\left|\right.\gamma(0)=\gamma(1)=\ket{0}_{d+1}\}.
\end{eqnarray}
Mathematically, this means that $\{\mathcal{M}_{d}^{G}\}_{d\in\mathbb{Z}}$ is an $\Omega$-spectrum with the base point $\{\ket{0}_{d}\}_{d\in \mathbb{Z}}$. This conjecture is important because this implies that a generalized cohomology theory works for the classification of SRE states \cite{GJF19}.

\subsection{Summery of This Paper}
As discussed above, the space of SRE states $\mathcal{M}_{d+1}^{G}$ determines both SPT phases and generalized Thouless pumps.  
On the basis of $K$-theory, previous researches
specify $\mathcal{M}^G_{d+1}$ for free fermionic SRE states, i.e. fermionic SRE states with quadratic Hamiltonians, with $G$ on-site symmetry \cite{Kitaev09, RSFL10,TK10}.
The classification of the free fermionic SRE states has been done both from field theory and lattice Hamiltonian perspectives. 
On the other hand, for interacting fermionic SRE states, it is difficult to determine $\mathcal{M}^G_{d+1}$, in particular, in lattice Hamiltonian formalism.

In this paper, using fermionic matrix product states (fMPSs), we analyze the space $\mathcal{M}_{1+1}$ with fermion parity symmetry, in the presence of interactions. We establish the existence of non-trivial pumps both in the trivial (orange loop in Fig. \ref{fig:variouspump}) and the non-trivial SPT phases (blue loop in Fig. \ref{fig:variouspump}). 

\begin{figure}[H]
 \begin{center}
  \includegraphics[width=80mm]{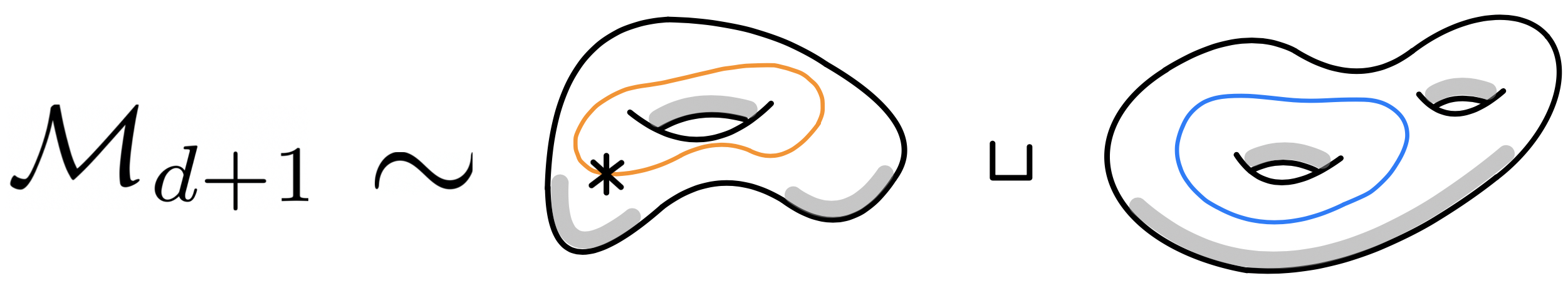}
 \end{center}
  \caption[]{
  A schematic picture of the space of $d+1$-dimensional SRE states, where $\ast$ represents the trivial state $\ket{0}_{d+1}$. Each connected component of $\mathcal{M}_{d+1}$ gives a different SPT phase. The component including the trivial state belongs to the trivial SPT phase, and the others are in non-trivial SPT phases. While the Kitaev pump is a loop in the trivial phase, we also consider pumps (i.e. loops) in non-trivial phases. We collectively call these pumps as generalized Thouless pumps.
}
\label{fig:variouspump}
\end{figure}

While the Kitaev's (and the original Thouless) pump is a pump in the trivial SPT phase as explained in the above, pumps in the non-trivial SPT phase also have been studied, especially in free fermionic systems \cite{Kitaev01, FK09, TK10}.\footnote{In the following, we collectively refer all of the above pumps as generalized Thouless pump and, in particular, we call the pump in the trivial phase the Kitaev pump.} Pumps in our analysis are consistent with these previous studies. We also present the topological invariants that characterize pumps both in trivial and non-trivial SPT phases in terms of fMPSs, and check the validity of them for several interacting models. 
We also give a geometric interpretation of the topological invariants.

\subsection{Outlook of This Paper}\;\\

The rest of the paper is organized as follows. 

In Sec.\ref{sre}, we give a quick review of the free Kitaev chain as the simplest example of fermionic SRE states in 1+1 dimensions (Sec.\ref{kitaev}). This model hosts two SPT phases: the trivial phase and the non-trivial phase, and  
shows a pump of the fermion parity both in the trivial and non-trivial phases.
We explain the fermion parity pump of the Kitaev chain in the trivial (Sec.\ref{nume2}) and non-trivial (Sec.\ref{nume3}) phase from several perspectives.

In Sec.\ref{MPS}, we introduce MPSs of bosonic (Sec.\ref{bmps}) and fermionic systems (Sec.\ref{fmps}) and identify several MPSs of bosonic and fermionic models. In particular, we characterize SRE states by an algebraic property of MPSs called an injective MPS, where the ($\zmod{2}$-graded) central simplicity of the algebra generated by matrices of MPSs plays a crucial role \cite{BWHV17}. We also illustrate this by using concrete examples. (We give a review of ($\zmod{2}$-graded) central simple algebra in Appendix \ref{CSA}.) In addition, we provide the necessary and sufficient condition for two injective fMPSs to give the same SRE state and summarize the condition in the form of Theorems \ref{fthmplus} and \ref{fthmminus}. 

In Sec.\ref{PumpInNonTriv}, we specify the space of the fMPS with the small matrix sizes and we reveal the existence of a non-contractible loop giving a pump in the non-trivial phase.

In Sec.\ref{inv}, we present a general theory to construct topological invariants for pumps in 1+1 dimensional fermionic SRE states in the formulation of fMPSs. 
Our construction is based on the Wall's structure theorem and works both in trivial and non-trivial phases.
For the trivial phase, fMPSs are similar to bosonic MPSs, and our construction is consistent with that for bosonic MPSs proposed in Ref.[\onlinecite{Shiozaki21}] (Sec.\ref{PumpInTriv}). On the other hand, our topological invariant in the non-trivial phase is essentially new since the non-trivial phase appears only in the fermionic case (Sec.\ref{subsubsec:topinv}). We also give geometric interpretations of the topological invariants (Sec.\ref{geometricint}).

In Sec.\ref{exampleGTP}, we apply our general theory of the pump topological invariants in Sec.\ref{inv} to several interacting fermionic models. We evaluate the topological invariants of pumps in trivial (Sec.\ref{Exampleofpumps}) and non-trivial (Sec.\ref{GTPKitaev}) phases and clarify the robustness of pumps in the presence of interactions.

Prior works are listed here.
Adiabatic pumps in SRE states have been discussed in the context of the Floquet SPT phase~\cite{EN16,KS16,PMV16,RH17}, where the periodic unitary time evolution which can be stroboscopic is studied. 
Studies focusing more on adiabatic pumps in SRE states/Hamiltonians themselves include bosonic systems with onsite symmetry~\cite{Shiozaki21,2204.03763}, multiple adiabatic parameters~\cite{KS20-1,KS20-2,YOF21,Xueda21}, and topological ordered states~\cite{AWH22}.

\section{Generalized Thouless pump in short-range entangle states in Kitaev chain}\label{sre}

SRE states provide the simplest class of topological phases. They are characterized by (a) the existence of global symmetry, (b) the uniqueness of the ground state, and (c) the existence of a finite energy gap. Despite their simplicity, SRE states describe many physically important systems, such as topological insulators and superconductors, and the Haldane chain, and so on.

In this section, we examine pump phenomena via the free 
Kitaev chain. In Sec.\ref{kitaev}, we first review the Kitaev chain as an example of $1+1$-dimensional SRE states. In Sec.\ref{nume2} and Sec.\ref{nume3}, we investigate pumps in the the free Kitaev chain for the trivial and non-trivial phases, respectively. In each phase, we investigate pumps in two different methods: The first one is through the action of symmetry on the boundary of the open chain (Sec.\ref{sec:edge_state_nontriv} and Sec.\ref{sec:edge_state_triv}), and the second is via a Hamiltonian with a texture mimicking a loop for pump in the closed chain (in Sec.\ref{nTexture} and \ref{nume}).

\subsection{Ground states in the Kitaev Chain}\label{kitaev}

The Kitaev chain is a model of a $1+1$-dimensional superconductor \cite{Kitaev01} with the Hamiltonian,
\begin{eqnarray}\label{KitaHam}
H=\sum_{j=1}^{L}\left(-\omega a^{\dagger}_{j+1}a_{j}-\omega a^{\dagger}_{j}a_{j+1}-\mu\left(a^{\dagger}_{j}a_{j}-\frac{1}{2}\right)+\Delta a_{j}a_{j+1}+\Delta^{\ast} a^{\dagger}_{j}a^{\dagger}_{j+1}\right),
\label{eq:Kitaev_chain}
\end{eqnarray}
where $L\in\mathbb{Z}$ is the system size, $a_{j}$ and $a_{j}^{\dagger}$ are the annihilation and creation operators with  the anti-commutation
relation
\begin{eqnarray}
\{a_{i},a_{j}\}=0,\hspace{3mm}\{a^{\dagger}_{i},a^{\dagger}_{j}\}=0,\hspace{3mm}\{a_{i},a^{\dagger}_{j}\}=\delta_{ij},\hspace{3mm}
\end{eqnarray}
$\omega\in\mathbb{R}$ is the hopping amplitude of the neighboring sites, $\mu\in\mathbb{R}$ is the chemical potential, and $\Delta=e^{i\theta}\abs{\Delta}\in\mathbb{C}$ is the gap function of the superconductivity. This Hamiltonian has fermion parity symmetry
\begin{eqnarray}
\left[H,P\right]=0,
\end{eqnarray}
with the fermion parity operator $P:=(-1)^{\sum_{j}a^{\dagger}_{j}a_{j}}$. 
In the periodic boundary condition, the Hamiltonian reads
  \begin{eqnarray}
H=\sum_{k}\frac{1}{2}(a^{\dagger}_{k},a_{-k})
{\cal H}_{\rm BdG}(k)
\left(
    \begin{array}{c}
      a_{k} \\
      a^{\dagger}_{-k} \\
    \end{array}
  \right),
 \end{eqnarray}
with the Bogoliubov-de Gennes (BdG) Hamiltonian
\begin{eqnarray}
\mathcal{H}_{\rm BdG}(k)=\left(
    \begin{array}{cc}
      -2\omega\cos(k)-\mu & i\Delta\sin(k) \\
      -i\Delta\sin(k) & 2\omega\cos(k)+\mu \\
    \end{array}
  \right),
\end{eqnarray}
where $k$ is the momentum $k$ along the chain. 
Diagonalizing the BdG Hamiltonian, we have the quasi-particle spectrum
\begin{eqnarray}
\epsilon(k)=\pm\sqrt{\left(2\omega\cos(k)+\mu\right)^{2}+4\left|\Delta\right|^{2}\sin^{2}(k)},
\end{eqnarray}
which is nonzero except for $2|\omega|=|\mu|$. 
Thus, except for $2|\omega|=|\mu|$,
the ground state is gapped and unique, and thus an SRE state.

The Kitaev chain has two different phases, trivial ($2|\omega|<|\mu|$) 
and non-trivial ($2|\omega|>|\mu|$) phases, which are 
separated by the gap closing point at $2|\omega|=|\mu|$. 
For the description of these phases, it is convenient to introduce the Majorana fermion 
\begin{eqnarray}
c_{2j-1}=e^{i\frac{\theta}{2}}a_{j}+e^{-i\frac{\theta}{2}}a^{\dagger}_{j},
\quad
c_{2j}=-i\left(e^{i\frac{\theta}{2}}a_{j}-e^{-i\frac{\theta}{2}}a^{\dagger}_{j}\right),
\label{eq:Majorana}
\end{eqnarray}
with the anti-commutation relation
\begin{eqnarray}
\{c_{i},c_{j}\}=2\delta_{i,j},\hspace{3mm}c^{\dagger}_{j}=c_{j}.
\end{eqnarray}
In the Majorana reprentation, the Hamiltonain in Eq.(\ref{KitaHam}) is recast into
\begin{eqnarray}
H=\frac{i}{2}\sum_{j}\left(-\mu c_{2j-1}c_{2j}+\left(\omega+\left|\Delta\right|\right)c_{2j}c_{2j+1}+\left(\omega-\left|\Delta\right|\right)c_{2j-1}c_{2j+2}\right),
\end{eqnarray}
with $P=\prod_{j}(-ic_{2j-1}c_{2j})$.
The analysis of the phases is particularly simple for (i) $\left|\Delta\right|=\omega=0$, $\mu<0$ (trivial phase), and (ii) $\left|\Delta\right|=\omega\neq0$, $\mu=0$ (non-trivial phase), as shown below.

\vspace{2ex}
\noindent
\underline{(i) $\left|\Delta\right|=\omega=0,\hspace{2mm}\mu<0$.}
\vspace{1ex}

  In this case, the Hamiltonian reads
  \begin{eqnarray}
  H=\frac{\mu}{2}\sum_{j}\left(-ic_{2j-1}c_{2j}\right).
  \label{eq:trivial}
  \end{eqnarray}
  Because any terms in the Hamiltonian commute with each other, and the eigenvalue of $-ic_{2j-1}c_{2j}$ is $\pm1$, the ground state $\ket{\mathrm{GS}}$ obeys
  \begin{eqnarray}\label{gscondi}
  -ic_{2j-1}c_{2j}\ket{\mathrm{GS}}=\ket{\mathrm{GS}}\hspace{2mm}\Leftrightarrow\hspace{2mm} a_{j}\ket{\mathrm{GS}}=0
  \end{eqnarray}
  for any sites $j=1,...,L$. Thus, the ground state does not have a fermion consisting of $c_{2j-1}$ and $c_{2j}$, which we represent by the diagram $\majosite{2j-1}{2j}$. In terms of the diagram, the ground state is given as
  \begin{eqnarray}
  \ket{\mathrm{GS}}=
  \majositewide{1}{2}\hspace{5mm}
  \majositewide{3}{4}
 \hspace{2mm}\cdots
 \majositenarrow{2L-3}{2L-2}
   \majosite{2L-1}{2L}.
   \end{eqnarray}
  As mentioned in the above, the ground state is unique as Eq.(\ref{gscondi}) imposes $L$ conditions on the Hilbert space with the dimension $2^L$. Putting a fermion, say at site $1$, we have the first excited state $a^{\dagger}_1\ket{\rm GS}$, which we represent as
    \begin{eqnarray}
  \hspace{3mm}\underset{1}{\bullet}\hspace{-.3em}-\hspace{-.2em}-\hspace{-.8em}>\hspace{-.7em}-\hspace{-.2em}-\hspace{-.3em}\underset{2}{\bullet}\hspace{5mm}
  \majositewide{3}{4}
 \hspace{2mm}\cdots
 \majositenarrow{2L-3}{2L-2}
   \majosite{2L-1}{2L}
   \end{eqnarray}
   The first excitation energy is $-\mu >0$. Since the ground state has a finite energy gap independent of the size of the system, it is an SRE state.
   
  Note that the above analysis works both for closed and open chains. 
 For both cases, we can impose the same condition in Eq.(\ref{gscondi}) on any site of the chain. In particular, no gapless boundary state appears in the trivial phase.
   
\vspace{2ex}   
\noindent
 \underline{(ii) $\left|\Delta\right|=\omega>0,\hspace{2mm}\mu=0$.}
\vspace{1ex}
 
  In this case, the Hamiltonian in the periodic boundary condition reads 
    \begin{eqnarray}\label{nontrivHam}
  H=-\omega\sum_{j}\left(-ic_{2j}c_{2j+1}\right),
  \end{eqnarray}
  with $c_{2L+1}=c_{1}$. Introducing a virtual complex fermion $\tilde{a}_j$ as
  \begin{eqnarray}
  \tilde{a}_{j}=\frac{1}{2}\left(c_{2j}+ic_{2j+1}\right),
  \end{eqnarray}
  we have
   \begin{eqnarray}
  H=2\omega\sum_{j=1}^{L}
  \left(\tilde{a}^{\dagger}_{j}\tilde{a}_{j}-\frac{1}{2}\right), 
  \end{eqnarray}
  and thus the ground state satisfies
  \begin{eqnarray}\label{nontrivgs}
  \tilde{a}_{j}\ket{\mathrm{GS}}=0,
  \end{eqnarray}
  for any $j=1,...,L$. 
  From Eq.(\ref{nontrivgs}), the ground state does not have a fermion consisting of $c_{2j}$ and $c_{2j+1}$, so we can represented it by the following diagram.
  \begin{eqnarray}
  \ket{\mathrm{GS}}
  =\hspace{3mm}\underset{1}{\bullet}\hspace{3mm}
  \majositewide{2}{3}\hspace{5mm}\majositewide{4}{5}\hspace{2mm}\cdots
  \majositenarrow{2L-4}{2L-3}\majositenarrow{2L-2}{2L-1}\hspace{3mm}
  \underset{2L}{\bullet}\hspace{-81.94mm}
  \raisebox{1.5mm}[0ex][0ex]
  {\text{$\contract{83.5mm}{-50mm}{50mm}$}}
\end{eqnarray}
  The ground state is unique and has a finite gap $2\omega$, and thus an SRE state again.
  
  In contrast to the case (i), the present case shows zero energy boundary modes in the open chain: In the open boundary condition, no bond between site L and site $1$ exists, and thus the summation in Eq.(\ref{nontrivHam}) excludes $j=L$. As a result, the ground state does not satisfy Eq.(\ref{nontrivgs}) at $j=L$. Therefore, in addition to the original ground state obeying  $\tilde{a}_L|{\rm GS}\rangle=0$, $\tilde{a}_{L}^{\dagger}\ket{\rm GS}$ also gives the ground state. Thus, the ground state in the open boundary condition has $2$-fold degeneracy. Physically, the $2$-fold degeneracy originates from the  Majorana fermions $c_{1}$ and $c_{2L}$ at the boundary of the system. 
  Since they do not participate in the Hamiltonian in Eq.(\ref{nontrivHam}), they 
  becomes gapless. 
  
  Note that fermion parity distinguishes the degenerate ground states:
  $\tilde{a}_{L}\ket{\rm GS}$ has an odd fermion parity relative to $\ket{\rm GS}$. 
    The doubly degenerate ground states due to Majorana boundary modes is a hallmark of the non-trivial phase in the Kitaev chain, which remain in the entire parameter region with $2|\omega|>|\mu|$.

\subsection{Adiabatic pump in the non-trivial phase}\label{nume2}

To investigate the fermion parity pump in fermionic SRE states, 
we consider a one-parameter family of unique gapped Hamiltonians $\{H(\theta)\}_{\theta\in\left[0,2\pi\right]}$ with
\begin{eqnarray}
H(0)=H(2\pi).
\end{eqnarray}
Below, we employ two different methods in the analysis of such a family of Hamiltonians. The first one is through the action of fermion parity symmetry on the boundary, and the second one is via a Hamiltonian of a closed system with spatially modulated $\theta$ mimicking a loop of a pump.

In this section, we examine the fermion parity pump in the Kitaev chain in the non-trivial phase. 
We introduce a phase of the gap function as the parameter of a pump (Sec.\ref{nModel}), and  perform both the boundary (Sec.\ref{sec:edge_state_nontriv}) and texture (Sec.\ref{nTexture}) analyses for the fermion parity pump.

\subsubsection{Model}\label{nModel}
As explained in the previous section,
in the non-trivial phase, the open chain hosts doubly degenerate ground states with opposite fermion parity, caused by Majorana boundary modes.  
For a finite chain, the degeneracy is slightly lifted, and the true ground state has either an even or odd fermion parity. 
As shown by Kitaev, the $2\pi$ phase rotation of the gap function flips the fermion parity of the ground state \cite{Kitaev01}.
Inspired by this observation, we regard the Hamiltonian in Eq.(\ref{KitaHam}) with 
$|\Delta|=\omega=1$ and $\mu=0$ as 
a one-parameter family of Hamiltonians in the non-trivial phase,
\begin{align}
    H(\theta)
    &=-\sum_{j=1}^{L} (a^\dag_{j+1}a_j+e^{i\theta} a_{j+1}a_j + h.c.),
    \label{eq:fixed_pt_Kitaev_open}
\end{align}
with fermion parity symmetry $P=\prod_{j=1}^L(-1)^{a_j^\dagger a_j}$.
As already shown in Sec.\ref{kitaev}, 
the Hamiltonian and the fermion parity operator read
\begin{align}
    H(\theta)
    =-\sum_{j=1}^{L}(-ic^{\frac{\theta}{2}}_{2j}c^{\frac{\theta}{2}}_{2j+1}),
    \quad
    P=\prod_{j=1}^L (-i c^{\frac{\theta}{2}}_{2j-1}c^{\frac{\theta}{2}}_{2j}) 
\label{eq:fixed_2}
\end{align}
in terms of the Majorana fermion in Eq.(\ref{eq:Majorana}), 
where we make explicit the $\theta$-dependence of the Majorana fermion.
Note that the Majorana fermion $c^\frac{\theta}{2}_j$ is $4\pi$-periodic in $\theta$, while the Hamiltonian is $2\pi$-periodic.

\subsubsection{Open chain}
\label{sec:edge_state_nontriv}

In the open chain,  $c_{2L+1}^{\frac{\theta}{2}}$ identically vanishes, and thus the ground state condition $-ic^{\frac{\theta}{2}}_{2j}c^{\frac{\theta}{2}}_{2j+1}=1$ excludes  $j=L$. 
The Majorana fermions $c^{\frac{\theta}{2}}_1$ and $c^{\frac{\theta}{2}}_{2L}$ do not participate in the Hamiltonian, so they are gapless in the whole region of $\theta$.

We investigate here the action of fermion parity symmetry on the boundary Majorana fermions and extract a topological invariant of the adiabatic process given by $\theta$. 
On the ground state, the fermion parity $P$ is written as 
\begin{align}
    P|_{\rm G.S.} = -i c^\frac{\theta}{2}_1 c^\frac{\theta}{2}_{2L},
    \label{eq:edge_fermion_parity}
\end{align}
from the ground state condition $-ic^{\frac{\theta}{2}}_{2j}c^{\frac{\theta}{2}}_{2j+1}=1$ ($j=1,\dots,L-1$).
Therefore, the fermion parity is ``fractionalized" on the ground state: it splits into two well-separated Majorana fermions $c^\frac{\theta}{2}_1$ and $c^\frac{\theta}{2}_{2L}$.
Note that the fractionalized fermion parity is not compatible with the $2\pi$-periodicity in $\theta$.
For instance, let us consider the left contribution  $\gamma^{\rm L}_\theta = c^\frac{\theta}{2}_1$ of $P|_{\rm G.S.}$. 
Using the $\mathrm{U}(1)$ phase ambiguity in the definition of $\gamma^{\rm L}_\theta$, we can recover the $2\pi$ periodicity by multiplying a suitable $\mathrm{U}(1)$ phase like $\gamma^{\rm L}_\theta = e^{\frac{i\theta}{2}} c^\frac{\theta}{2}_1$. 
However, the choice $\gamma^{\rm L}_\theta = e^{\frac{i\theta}{2}} c^\frac{\theta}{2}_1$ does not provide a proper definition of the Majorana fermion since it obeys $(\gamma^{\rm L}_{\theta})^2=e^{i\theta}\neq 1$. 

The incompatibility observed in the above is general and originates from a topological obstruction.
Since Majorana boundary modes are only excitations between the (nearly) degenerate ground states of the nontrivial phase, the fermion parity always shows the fractionalization in the above.
The left contribution $\gamma^{\rm L}_\theta$ consists of the Majorana mode on the left boundary, and can be chosen to be $2\pi$-periodic in $\theta$, $\gamma^{\rm L}_{\theta+2\pi}=\gamma^{\rm L}_\theta$, by using the $\mathrm{U}(1)$ phase ambiguity.  
However, the square of the $2\pi$-periodic $\gamma^{\rm L}_\theta$ gives a non-zero $\mathrm{U}(1)$ phase $(\gamma^{\rm L}_\theta)^2=e^{i\Phi_\theta} \in \mathrm{U(1)}$ in general, from which we can define the $\zmod{2}$ invariant 
\begin{align}
    \nu = \frac{1}{2\pi i} \oint d \Phi_\theta .
    \label{eq:bdy_Z2_inv}
\end{align}
Although the above integral takes an integer, $\nu$ defines a $\zmod{2}$ number because the $2\pi$-periodic $\gamma^{\rm L}_\theta$ has the ambiguity $\gamma^{\rm L}_\theta \mapsto e^{i\alpha_\theta} \gamma^{\rm L}_\theta$ with a smooth 2$\pi$-periodic function $e^{i\alpha_\theta}$, and $e^{i\alpha_\theta}$ changes $\nu$ by an even integer. 
Then, an odd $\nu$ obstructs the $2\pi$-periodic $\gamma^{\rm L}_\theta$ to obey
the proper parity relation $(\gamma^{\rm L}_\theta)^2=1$ at the same time. 
In the above case, we have $\nu=1$, and thus the incompatibility remains for any deformation keeping the gap.

\subsubsection{Textured Hamiltonian}\label{nTexture}

In the previous subsection, we investigate a family of Hamiltonians $H(\theta)$ ($\theta\in\left[0,2\pi\right]$) in the open chain. Here we examine the closed chain 
using a textured Hamiltonian similar to $H(\theta)$.
The textured Hamiltonian is a Hamiltonian with a spatially modulated parameter: Let $h_{j}(\theta)$ be the local term at site $j$ in Eq.(\ref{eq:Kitaev_chain}) with $\Delta=|\Delta|e^{i\theta}$
(i.e. $H(\theta)=\sum_{j}h_{j}(\theta)$), then we define the textured Hamiltonian $H_{\rm text}^{l}$ as follows,
\begin{eqnarray}
H_{\rm text}^{l}=\sum_{j=1}^{l} h_{j}(\theta=\frac{2\pi j}{l})+\sum_{j=l+1}^{L}h_{j}(\theta=2\pi=0).
\end{eqnarray}
where $l$ is the size of the texture. 
In the nontrivial phase, the spatial texture in the gap function results in a fermion parity pump in the spatial direction, and thus the ground state of $H_{\rm text}^{l}$ is expected to host an odd fermion parity relative to that of an untextured Hamiltonian.

Figure \ref{fig:numfppump} shows our numerical result for the fermion parity of the ground state.
This result confirms that the texture in the gap function actually induces a flip of the fermion parity in the non-trivial phase.

\begin{figure}[h]
 \begin{center}
  \includegraphics[width=70mm]{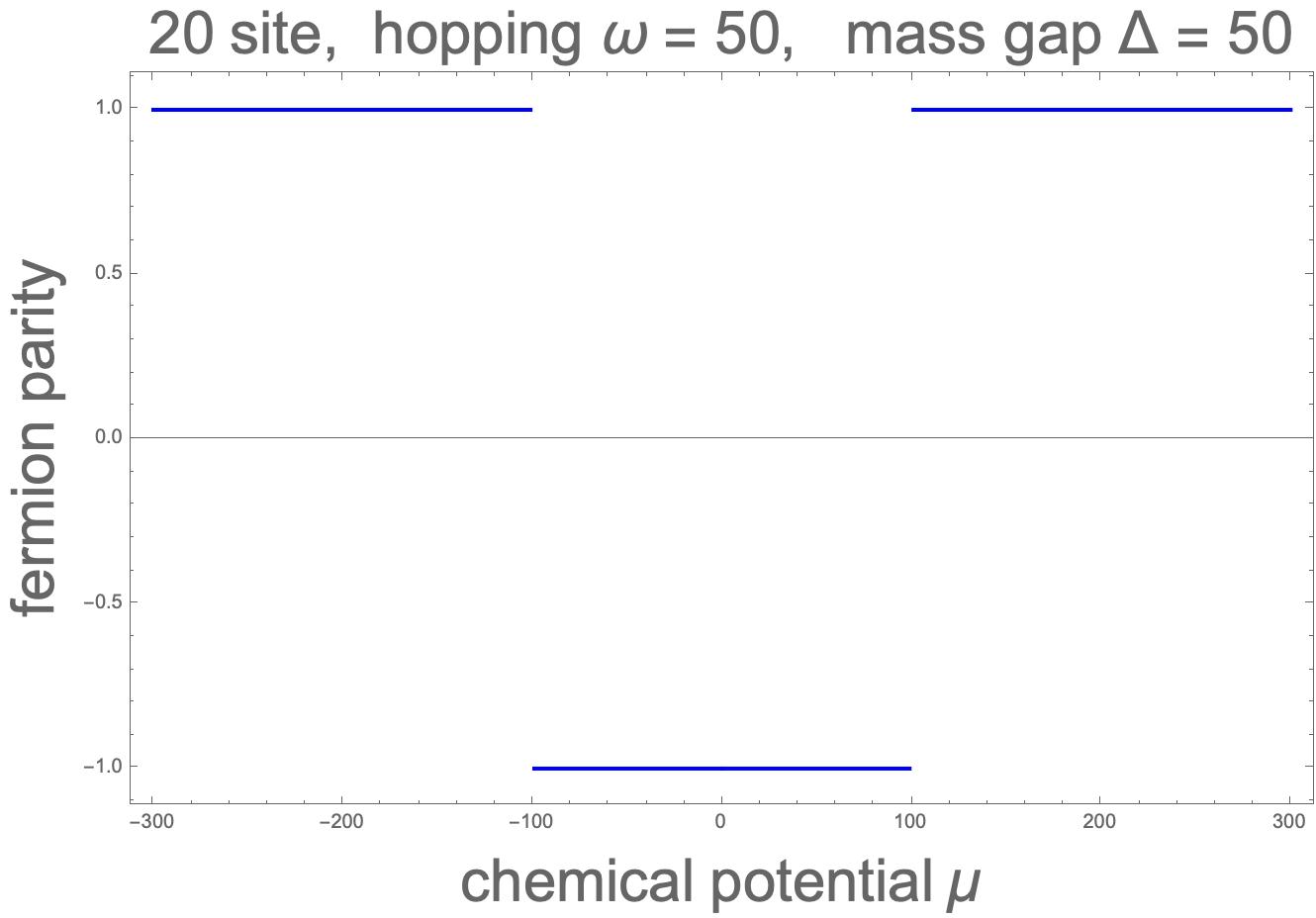}
 \end{center}
 \caption[]{
 The ratio of the fermion parity between the ground states of the textured and the untextured Hamiltonian with $L=l=20$, and $\omega=\abs{\Delta}=50$. 
 The ratio is $-1$ in the nontrivial phase ($|\mu|<2|\omega|$). 
}
 \label{fig:numfppump}
\end{figure}

We can also analytically demonstrate the flip of the fermion parity. For simplicity, we consider $h_j(\theta)$ with $|\Delta|=\omega=1$, $\mu=0$, and $L=l$. 
When the system size $L$ is sufficiently large, the resultant textured Hamiltonian is almost approximated by the unitary transformation 
\begin{eqnarray}
U_{\rm text}=\prod_{j}e^{i\frac{\theta_{j}}{2}n_{j}}, 
\quad (\theta_{j}=2\pi j/L, 
\,n_{j}=a^{\dagger}_{j}a_{j}),
\end{eqnarray}
on the Hamiltonian
\begin{align}\label{untextured}
    H
    &=-\sum_{j=1}^{L} (a^\dag_{j+1}a_j
    +a^\dag_{j+1}a_j
    +a_{j}^{\dagger}a_{j+1}^{\dagger}
    +a_{j+1}a_j ).
\end{align}
Actually, we have
\begin{eqnarray}
    U_{\rm text}HU_{\rm text}^{\dagger}
    &=&-\sum_{j=1}^{L} (e^{i\frac{\theta_{j+1}-\theta_{j}}{2}}a^\dag_{j+1}a_j
    +e^{i\frac{\theta_{j+1}+\theta_{j}}{2}}a_{j}^{\dagger}a_{j+1}^{\dagger} +{\rm h.c.})\\
    &=&-\sum_{j=1}^{L-1} (e^{i\frac{\pi}{L}}a^\dag_{j+1}a_j
    +e^{i\frac{2\pi j}{L}+i\frac{\pi}{L}}a_{j}^{\dagger}a_{j+1}^{\dagger} +{\rm h.c.})-(e^{i\frac{\pi}{L}-i\pi}a_{1}^{\dagger}e^{\frac{2\pi L}{L}+i\pi}a_{L}+a_{L}^{\dagger}a_{1}^{\dagger}+{\rm h.c.}),\label{eq:texturedlike}
\end{eqnarray}
which is almost the textured Hamiltonian if we ignore terms $O(\frac{\pi}{L})\ll 1$.
However, because of the unnecessary factor $e^{-i\pi}$ between sites $L$ and $1$ in the second term of Eq.(\ref{eq:texturedlike}), this Hamiltonian does not satisfy the periodic boundary condition even for $L\gg 1$. To avoid this problem, we modify the unitary transformation $U_{\rm text}$ as
\begin{eqnarray}
\tilde{U}_{\rm text}=c_{2L}U_{\rm text},
\end{eqnarray}
where $c_{2L}$ is the Majorana fermion at site $L$, then $\tilde{H}_{\rm text}:=\tilde{U}_{\rm text}H\tilde{U}_{\rm text}^{\dagger}$ approximates the textured Hamiltonian:
\begin{eqnarray}
H_{\rm text}^{l=L}=\tilde{H}_{\rm text}+O(\frac{\pi}{L}).
\end{eqnarray} 

The ground state of $H_{\rm text}$ is given by $\ket{{\rm GS'}}=\tilde{U}_{\rm text}\ket{\rm GS}$ with $\ket{\rm GS}$ the ground state of the untextured Hamiltonian in Eq.(\ref{untextured}).
Since $\tilde{U}_{\rm text}$ anti-commutes with the fermion parity $P=(-1)^{\sum_{j}n_{j}}$, the ratio
\begin{eqnarray}
\frac{\bra{{\rm GS'}}P\ket{{\rm GS'}}}{\bra{\rm GS}P\ket{\rm GS}}
\end{eqnarray}
is equal to $-1$.

\subsection{Adiabatic pump in the trivial phase}\label{nume3}

In this section, we consider the fermion parity pump in the trivial phase (Sec.\ref{nModel}). We perform both the boundary (Sec.\ref{sec:edge_state_nontriv}) and texture (Sec.\ref{nTexture}) analyses.

\subsubsection{Model}\label{tModel}
First, we give a solvable model of a pump in the trivial phase, which is constructed from the pump Hamiltonian in Eq.(\ref{eq:fixed_2}) in the non-trivial phase.
For this purpose, it is useful to rewrite the local term in  (\ref{eq:fixed_2}) in the form of the unitary transformation 
\begin{align}
    -ic^\frac{\theta}{2}_{2j}c^\frac{\theta}{2}_{2j+1}
    =
    U'_\theta (-i c_{2j}c_{2j+1}) [U'_\theta]^{-1},
    \label{eq:local_ham_nontriv}
\end{align}
where $c_j$ is the Majorana fermion in Eq.(\ref{eq:Majorana}) with $\theta=0$, and 
\begin{align}
    U'_\theta 
    = \prod_{j \in \Z} e^{-\frac{i\theta}{2} a^\dag_ja_j}
    = \prod_{j\in \Z} e^{-\frac{i\theta}{2} \frac{1+ic_{2j-1}c_{2j}}{2}}
\end{align}
is the $(\theta/2)$-phase rotation of the complex fermion $a_j$. 
Noting that Eq.(\ref{eq:trivial}) of the trivial phase is related to Eq.(\ref{nontrivHam}) of the non-trivial phase by the transformation $c_j\to c_{j-1}$, 
we consider the local term 
\begin{align}
    B_j(\theta)
    := U_\theta (-i c_{2j-1}c_{2j}) U_\theta^{-1}
\label{eq:Bj}
\end{align}
with 
\begin{align}
    U_\theta 
    = \prod_{j\in \Z} e^{-\frac{i\theta}{2} \frac{1+ic_{2j}c_{2j+1}}{2}}.
\end{align}
In terms of the original complex fermion, $B_j(\theta)$ is given as
\begin{align}\label{eq:DWCMHam}
    B_j(\theta)
    &=
    \frac{1+\cos \theta}{2}(1-2a^\dag_ja_j)-\frac{1-\cos \theta}{2}(a_{j-1}+a^\dag_{j-1})(a_{j+1}-a^\dag_{j+1})
    +i \sin \theta (a_ja_{j+1} + a^\dag_ja^\dag_{j+1}).
\end{align}
The resultant Hamiltonian $H(\theta)=-\sum_j B_j(\theta)$ has the $2\pi$-periodicity in $\theta$, and is unitary equivalent to Eq.(\ref{eq:trivial}) with $\mu=-1$. Therefore, it defines a pump in the trivial phase.
Note that $H(\theta)$ is solvable
since $B_j(\theta)$s commute with each other and have eigenvalues $\pm 1$.

\subsubsection{Open chain}\label{sec:edge_state_triv}
Similar to the analysis in Sec.~\ref{sec:edge_state_nontriv}, we investigate the fermion pump in an open chain of the solvable model through fermion parity on the boundaries.
For the open chain with $L$ sites, we consider the Hamiltonian, 
\begin{align}
H(\theta)
= H_{\rm bulk}(\theta) + H_{\rm bdy}(\theta),
\end{align}
where $H_{\rm bulk}(\theta)$ is the bulk Hamiltonian
\begin{align}
H_{\rm bulk}(\theta)
= -\sum_{j=2}^{L-1} B^\theta_j
=-\sum_{j=2}^{L-1} U_\theta (-ic_{2j-1}c_{2j}) U_\theta^{-1}, 
\end{align}
with $U_\theta$ in the open chain
\begin{align}
U^{\rm open}_\theta 
    = \prod_{j=1}^{L-1} e^{-\frac{i\theta}{2} \frac{1+ic_{2j}c_{2j+1}}{2}}  
=\prod_{j=1}^{L-1}e^{-i\frac{\theta}{4}}
\left[
\cos\left(\theta/4\right)
+c_{2j}c_{2j+1}\sin\left(\theta/4\right)
\right].
\end{align}
On the boundaries, we consider a local Hamiltonian $H_{\rm bdy}(\theta)$ instead of $B_{j=1, L}^\theta$, which are defined by $B_{j=1}^\theta = U^{\rm open}_\theta (\frac{i}{2}c_1c_2)[U^{\rm open}_\theta]^{-1}$ and $B_{j=L}^\theta = U^{\rm open}_\theta (\frac{i}{2}c_{2L-1}c_{2L})[U^{\rm open}_\theta]^{-1}$, respectively, because the latter terms are not $2\pi$-periodic in $\theta$.
We assume that $H_{\rm bdy}(\theta)$ is  2$\pi$-periodic in $\theta$ and small compared to the bulk gap.

The system supports four-fold ground state degeneracy: 
Since $H_{\rm bulk}(\theta=0)=2\sum_{j=2}^{L-1}(a_j^\dagger a_j-1/2)$, 
the ground states of $H_{\rm bulk}(\theta=0)$ are annihilated by $a_j$ with $j=2,\dots,L-1$, which consist of the four states 
\begin{align}
\ket{\Psi^1_0}=\ket{\rm vac},\quad 
\ket{\Psi^2_0}=a_1^\dag \ket{\rm vac},\quad 
\ket{\Psi^3_0}=a_L^\dag \ket{\rm vac},\quad 
\ket{\Psi^4_0}=a_1^\dag a_L^\dag \ket{\rm vac}, 
\label{eq:four}    
\end{align}
with the Fock vacuum $|{\rm vac}\rangle$.
Thus, from $H_{\rm bulk}(\theta)
=U_\theta H_{\rm bulk}(\theta=0)U_\theta^{-1}$, we have four-fold degenerate ground states of $H_{\rm bulk}(\theta)$
$\ket{\Psi_\theta^i}=U_\theta\ket{\Psi_0^i}$ ($i=1,2,3,4$).
Note that even in the presence of $H_{\rm bdy}(\theta)$,
the ground states are nearly degenerate as long as $H_{\rm bdy}(\theta)$ is small enough.

To study the fermion parity pump, we rewrite the four states in Eq.(\ref{eq:four}) as
\begin{align}
&\ket{\Psi_0^1}=\frac{1}{2^L}
\sum_{\sigma_1,\dots,\sigma_L}
|\sigma_1,\dots, \sigma_L\rangle,
\quad
\ket{\Psi_0^2}=\frac{1}{2^L}
\sum_{\sigma_1,\dots,\sigma_L}
\sigma_1|\sigma_1,\dots, \sigma_L\rangle,
\nonumber\\
&\ket{\Psi_0^3}=\frac{1}{2^L}
\sum_{\sigma_1,\dots,\sigma_L}
\sigma_L|\sigma_1,\dots, \sigma_L\rangle,
\quad
\ket{\Psi_0^4}=\frac{1}{2^L}
\sum_{\sigma_1,\dots,\sigma_L}
\sigma_1\sigma_L|\sigma_1,\dots, \sigma_L\rangle,
\label{eq:vac}
\end{align}
with
\begin{align}
|\sigma_1,\dots, \sigma_L\rangle=
(1+\sigma_1 a_1^{\dagger})(1+\sigma_2 a_2^{\dagger})\cdots
(1+\sigma_L a_L^{\dagger})|{\rm vac}\rangle,    
\end{align}
where the summation in Eq.(\ref{eq:vac})  
runs over all possible $\sigma_j=\pm 1$. 
Then, using the relation
\begin{align}
c_{2j}c_{2j+1}|\sigma_1,\dots, \sigma_L\rangle
&=-i(a_j-a_j^\dagger)(a_{j+1}+a_{j+1}^\dagger)
|\sigma_1,\dots,\sigma_L\rangle
\nonumber\\
&=i\sigma_j\sigma_{j+1}    
|\sigma_1,\dots,\sigma_L\rangle, 
\end{align}
we obtain 
\begin{align}
\ket{\Psi_\theta^1}
&=\frac{1}{2^L}
\sum_{\sigma_1,\dots,\sigma_L}
e^{-\frac{i\theta}{2} N_{\rm dw}}
\ket{\sigma_1,\dots,\sigma_L}, 
\quad
\ket{\Psi_\theta^2}=\frac{1}{2^L}
\sum_{\sigma_1,\dots,\sigma_L}
e^{-\frac{i\theta}{2} N_{\rm dw}}
\sigma_1\ket{\sigma_1,\dots,\sigma_L}, 
\nonumber\\
\ket{\Psi_\theta^3}
&=\frac{1}{2^L}
\sum_{\sigma_1,\dots,\sigma_L}
e^{-\frac{i\theta}{2} N_{\rm dw}}
\sigma_L\ket{\sigma_1,\dots,\sigma_L}, 
\quad
\ket{\Psi_\theta^4}=\frac{1}{2^L}
\sum_{\sigma_1,\dots,\sigma_L}
e^{-\frac{i\theta}{2} N_{\rm dw}}
\sigma_1\sigma_L\ket{\sigma_1,\dots,\sigma_L}.
\end{align}
Here $N_{\rm dw} = \sum_{j=1}^{L-1} \frac{1-\sigma_j\sigma_{j+1}}{2}$ counts domain walls in the array $\sigma_1,\dots,\sigma_L$ where adjacent $\sigma_j$ and $\sigma_{j+1}$ have an opposite sign.

The above ground states $\ket{\Psi_\theta^i}$ exhibit a fermion (or anti-fermion) pump.
For instance, let us consider $\ket{\Psi_\theta^1}$, which is the Fock vacuum $\ket{{\rm vac}}$ at $\theta=0$.
Because $N_{\rm dw}$ is even (odd) when $\sigma_1=\sigma_L$ ($\sigma_1=-\sigma_L$), we have
\begin{align}
|\Psi^1_{\theta=2\pi}\rangle
&=\frac{1}{2^L}\sum_{\sigma_1,\dots, \sigma_{L-1}} 
|\sigma_1,\dots,\sigma_{L-1}, \sigma_1\rangle
-
\frac{1}{2^L}\sum_{\sigma_1,\dots, \sigma_{L-1}}
|\sigma_1,\dots,\sigma_{L-1}, -\sigma_1\rangle
\nonumber\\
&=a_1^{\dagger}a_L^{\dagger}|{\rm vac}\rangle
(=|\Psi^4_{0}\rangle).
\label{eq:pump_solvable}
\end{align}
Thus, the one-cycle evolution 
pumps boundary fermions $a_1^\dagger$ and $a_L^\dagger$ on $\ket{\Psi^1_0}$. As a result, $\ket{\Psi^1_0}$ goes to $\ket{\Psi^4_0}$ and vice versa.
In a similar manner, we can show that $\ket{\Psi^2_0}$ and $\ket{\Psi^3_0}$ are interchanged after the one-cycle evolution. Note that
the pumped fermions vanish after the next one-cycle evolution, which suggests that a $\mathbb{Z}/2\mathbb{Z}$ number characterizes the pump. 
Actually, we can construct the $\mathbb{Z}/2\mathbb{Z}$ number from the fermion parity operator.

To construct the $\mathbb{Z}/2\mathbb{Z}$ number, we take the basis in which the  ground states of $H_{\rm bulk}(\theta)$ are $2\pi$-periodic in $\theta$:
Taking linear combinations of $\Psi_\theta^i$, we have  
\begin{align}
\Ket{\Psi_\theta(\sigma_1,\sigma_L)} =\sum_{\sigma_2,\dots,\sigma_{L-1}} 
e^{-\frac{i\theta}{2}(N_{\rm dw}+1-\frac{\sigma_1+\sigma_L}{2})}
\ket{\sigma_1,\sigma_2,\cdots,\sigma_{L-1},\sigma_L},
\end{align}
where $\sigma_1=\pm 1$ and $\sigma_L=\pm 1$ are now the indices specifying the four-fold degenerate ground states.
The ground states in the new basis are $2\pi$-periodic in $\theta$ since $N_{\rm dw}$ is even (odd) when $\sigma_1=\sigma_L$ ($\sigma_1=-\sigma_L$).

The fermion parity operator $P=(-1)^{\sum_j a^\dagger_j a_j}$ acts on the ground states as
\begin{align}
    P \ket{\Psi_\theta(\sigma_1,\sigma_L)}
    &=
    e^{i\theta \frac{\sigma_1+\sigma_L}{2}} \ket{\Psi_\theta(-\sigma_1,-\sigma_L)}. 
    \nonumber\\
    &=\sum_{\sigma_1',\sigma_L'}
    [e^{\frac{i\theta}{2}\bar{\sigma}_1^z}\bar{\sigma}_1^x]_{\sigma_1\sigma_1'}
    [e^{\frac{i\theta}{2}\bar{\sigma}_L^z}\bar{\sigma}_L^x]_{\sigma_L\sigma_L'}
    \ket{\Psi_\theta(\sigma'_1,\sigma'_L)}, 
    \label{eq:fermion_cluster_basis}
\end{align}
where $\bar \sigma^\mu_i$ are the Pauli matrices acting on the index $\sigma_i$ ($i=1,L$).
Therefore, we have the matrix representation of the fermion parity in a fractionalized form
\begin{align}
P|_{\rm G.S.}= p_1^\theta p^\theta_L 
\label{eq:P_pp}
\end{align}
with
\begin{align}
p_j^\theta = e^{\frac{i\theta}{2} \bar \sigma^z_j} \bar \sigma^x_j, \quad j=1,L. 
\label{eq:bdy_parity_left}
\end{align} 
The fractionalized parity operator obeys $(p_j^\theta)^2=1$ like an ordinary parity operator, 
but it is not $2\pi$-periodic in $\theta$, {\it i.e.} $p_j^{\theta+2\pi}=-p_j^\theta$.
We note that $p_1^\theta$ has a $\mathrm{U}(1)$ phase ambiguity: a simultaneous redefinition $p_1^\theta \mapsto e^{i\alpha} p_1^{\theta}$ and $p_L^\theta \mapsto e^{-i\alpha} p_L^\theta$ does not change the equality \eqref{eq:P_pp}. 
Whereas the $2\pi$-periodicity of $p_j^\theta$ can be recovered by using the phase ambiguity, 
it is not compatible with $(p_j^\theta)^2=1$: Once we choose the $\mathrm{U}(1)$ phase of $p_1^{\theta}$ such that $p_1^{\theta}$ is $2\pi$-periodic in $\theta$, we have a non-trivial $\mathrm{U}(1)$ phase in $(p_1^\theta)^2$, $(p_1^\theta)^2=e^{i\Phi_\theta}$, and thus $p_1^\theta$ is now a projective representation of the parity.

As discussed in Sec.\ref{sec:edge_state_nontriv}, 
the incompatibility originates from a topological obstruction:
In a manner similar to Sec.\ref{sec:edge_state_nontriv}, 
the phase $\Phi_\theta$ defines the topological number $\nu$ in Eq. (\ref{eq:bdy_Z2_inv}). 
Note that only the $\zmod{2}$ part of $\nu$
is relevant, since the $2\pi$-periodic $p_1^\theta$ still has a phase ambiguity $p_1^\theta \mapsto e^{i\alpha_\theta} p_1^\theta$ with a periodic function $e^{i\alpha_\theta}$,  which changes $\nu$ by an even integer. 
In the present case, we obtain the $2\pi$-periodic $p_1^\theta$ by multiplying $p_j^\theta$ in Eq. (\ref{eq:bdy_parity_left}) by $e^{\frac{i\theta}{2}}$. Thus, we have $e^{i\Phi_\theta}=e^{i\theta}$ and $\nu=1$, 
which means that  $\Phi_\theta$ cannot be identically zero.

\subsubsection{Stacked Kitaev chain with texture}\label{nume}
In a manner similar to Sec.\ref{nTexture}, we can investigate the fermion parity pump in the closed chain of the solvable model in the above by introducing a texture in the Hamiltonian. We expect that the fermion parity of the ground state changes by $-1$ by introducing the texture, 
but instead of repeating the straightforward analysis, we here consider another model of the closed chain in the trivial phase.

A stack of two Kitaev chains is topologically trivial since the Kitaev chain belongs to a $\mathbb{Z}/2\mathbb{Z}$ phase. 
In this section, we consider a pump in the following $4\times4$ Hamiltonian describing the stack of Kitaev chains, 
\begin{eqnarray}
H=\frac{1}{2}\sum_{k,\sigma,\sigma'} ( a^{\dagger}_{k,\sigma}, a_{-k,\sigma})
\mathcal{H}_{\mathrm{BdG}}(k)_{\sigma,\sigma'}
  \left(
    \begin{array}{c}
      a_{k,\sigma'} \\
      a^{\dagger}_{-k,\sigma'}
  \end{array} \right),\\
\vspace{4mm}\hspace{-2mm}\mathcal{H}_{\mathrm{BdG}}(k)=\sin(k)\tau_{1}\otimes\sigma_{0}+(m+\cos(k))\tau_{3}\otimes\sigma_{3},
\end{eqnarray}
where $\tau_{i}$ 
are Pauli matrices in the Nambu space, $\sigma_{i}$ are Pauli matrices labeling the two Kitaev chains, and $m$ is a real parameter. 
This model has particle-hole symmetry 
$\left[H,\Xi\right]=0$ with $\Xi=K\tau_{1}\otimes\sigma_{0}$ and $K$ the complex conjugate operator.

To investigate the fermion parity pump of this model,  
we add a term with a spatial texture. The additional texture term should keep particle-hole symmetry and commute with the first term of the above Hamiltonian to maintain a gap of the system. 
Based on this argument, we consider the following one-parameter family of Hamiltonians
\begin{eqnarray}\label{eq:paramBdG}
\mathcal{H}_{\mathrm{BdG}}(k,\theta)=\sin(k)\tau_{1}\otimes\sigma_{0}+\sin(\theta)\tau_{3}\otimes\sigma_{1}+(m+\cos(\theta)+\cos(k))\tau_{3}\otimes\sigma_{3}.
\end{eqnarray}
Performing the Fourier transformation, 
\begin{eqnarray}
a_{k,\sigma}=\sum_{n}e^{-ink}a_{n,\sigma},\quad
a^{\dagger}_{k,\sigma}=\sum_{n}e^{ink}a^{\dagger}_{n,\sigma}, 
\end{eqnarray}
we have the Hamiltonian in the real space
\begin{align}
H(\theta)&=\sum_{j,\sigma,\sigma'}
\left[
\frac{1}{2}a^{\dagger}_{j,\sigma}(\sigma_3)_{\sigma,\sigma'}a_{j-1,\sigma'}
+\frac{1}{2}a^{\dagger}_{j-1,\sigma}(\sigma_3)_{\sigma,\sigma'}a_{j,\sigma'}
+\sin(\theta)a_{j,\sigma}^{\dagger}
(\sigma_1)_{\sigma,\sigma'}a_{j,\sigma'}
\right.
\nonumber\\
&
\left.
+(m+\cos(\theta))a^{\dagger}_{j,\sigma}(\sigma_3)_{\sigma,\sigma'}a_{j,\sigma'}
\right]
\nonumber\\
&+\sum_{j,\sigma}
\left[
+\frac{1}{2i}a^{\dagger}_{j,\sigma}a^{\dagger}_{j-1,\sigma}
-\frac{1}{2i}a^{\dagger}_{j-1,\sigma}a^{\dagger}_{j,\sigma}
+\frac{1}{2i}a_{j,\sigma}a_{j-1,\sigma}
-\frac{1}{2i}a_{j-1,\sigma}a_{j,\sigma}
\right].
\label{paraHam}
\end{align}
When $\theta=0,\pi$, the system reduces to two decoupled Kitaev chains. 
The decoupled Kitaev chains belong to the same phase, but the common phase can be different between $\theta=0$ and $\theta=\pi$.
Actually, this happens for $|m|<2$, which suggests that $H(\theta)$ with $|m|<2$ gives a non-trivial loop.



Similar to Sec.\ref{nTexture}, we numerically examine the fermion parity pump of the stacked Kitaev chain under the periodic boundary condition by introducing the textured Hamiltonian $H^{l=L}_{\mathrm{text}}=\sum_j h_j(\theta=\frac{2\pi j}{L})$, where $h_j(\theta)$ is the local term of $H(\theta)$ in Eq.(\ref{paraHam}), {\it i.e.} $H(\theta)=\sum_j h_j(\theta)$.
%
Figure \ref{4by4texture} shows the numerical result for the ratio of the fermion parity of the ground state for $H^{l=L}_{\mathrm{text}}$ and that for $H(\theta=0)$.
This result suggests the presence of the fermion parity pump for $|m|<2$.

\begin{figure}[H]
 \begin{center}
  \includegraphics[width=60mm]{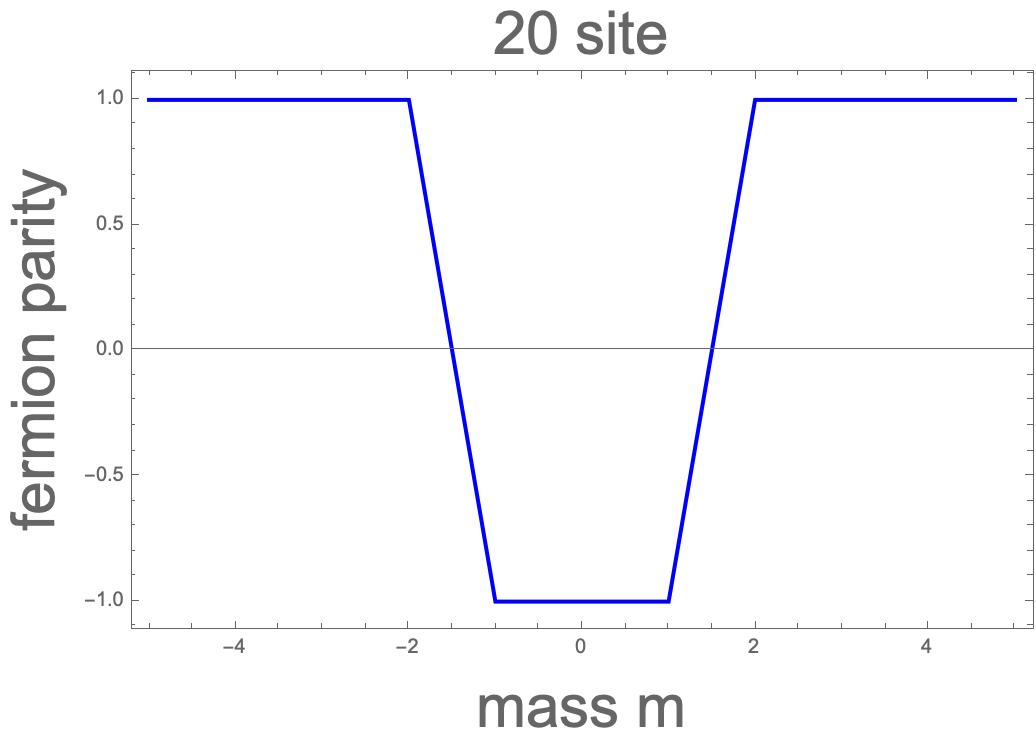}
 \end{center}
 \vspace{0mm}
 \caption[]{
 The fermion parity of the ground state of the textured Hamiltonian with the system $L=20$. 
 The fermion parity changes by a factor of $-1$ at $m=2$
}
 \label{4by4texture}
 \end{figure}

\section{Matrix Product State}\label{MPS}

So far, we have considered pumps in a particular model, i.e. the free Kitaev chain. 
Now we develop a theory of pumps in $1+1$-dimensional translation invariant SRE states including interactions, based on matrix product states (MPSs) representations~\cite{P-GVWC07}~\footnote{Most of these properties are stated in Ref.\onlinecite{FNW92}, but in a mathematical style.}.
MPSs provide a systematic way to describe $1+1$-dimensional many-body quantum states by using a set of matrices. 
MPSs can approximate any non-degenerate gapped ground states with arbitrary precision by increasing the bond dimension as a polynomial function of system size
\cite{ALVV17}.
For bosonic states, a class of MPSs called injective MPSs played an important role in studying topological natures of the non-degenerate gapped ground states.
An injective MPS is a translation invariant MPS with a fixed finite bond dimension, irrespective of system size, and has algebraic properties described below. 
Injective MPSs can describe any non-degenerate gapped ground states of short-ranged  Hamiltonians~\cite{P-GVWC07}.
In this section, we introduce fermionic injective MPSs (fMPSs), $\Z_2$-graded generalization of injective MPSs, along the lines of Ref.\onlinecite{BWHV17}.

In contrast to general MPSs,  injective MPSs have a limitation to describe non-degenerate gapped ground states: 
Whereas a generic non-degenerate gapped ground state
may allow power-law corrections in exponentially decaying correlation functions~\cite{Tasaki20},
MPSs with a fixed bond dimension do not have such corrections. 
We leave the topological classification of pumps for general fermionic SRE states in future work. 

Below we assume the translational invariance of states $\ket{\psi}$: i.e., $T \ket{\psi} = \ket{\psi}$ with $T$ the translation operator by a lattice constant. 


\subsection{Bosonic MPS}\label{bmps}

\subsubsection{Injective MPS}
Consider a 1-dimensional lattice with $L$ sites with local Hilbert space spanned by the orthonormal basis $\{ \ket{i_k}\}_{i_k=1}^N$ at site $k$. 
The lattice-translation operator $T$ is defined by $T \ket{i_1 i_2 \cdots i_L} = \ket{i_L i_1 \cdots i_{L-1}}$. 
We call states invariant under the lattice translation $T \ket{\psi} = \ket{\psi}$ as translation invariant states.
For the wave function $\psi(i_1,\dots,i_L)$ defined by $\ket{\psi}=\sum_{i_1,\dots,i_L=1}^N \psi(i_1,\dots,i_L)\ket{i_1\cdots i_L}$, the state $\ket{\psi}$ is translation invariant if
\begin{align}
    \psi(i_1,i_2,\dots,i_L) = \psi(i_2,\dots,i_L,i_1)
    \label{eq:psi_T_boson}
\end{align}
holds for any $i_1,\dots,i_L$. 
It is known that any translation invariant state $\ket{\psi}$ is represented in the form of a translation invariant MPS~\cite{P-GVWC07}
\begin{align}
    \ket{\psi}
    =
    \ket{\{A^i\}_i}_L
    :=
    \sum_{i_1,\dots,i_L=1}^N \tr [A^{i_1} \cdots A^{i_L}] \ket{i_1\cdots i_L}, 
\end{align}
where $A^i$s are $n \times n$ matrices.
($n$ is called the bond dimension.)
For $A^{i}_{\alpha\beta}$, we call $i$  the physical leg and $\alpha$ and $\beta$ the virtual legs. 

Apparently,the MPS representation of $\ket{\psi}$ is not unique. 
For example, two MPSs related by $A^i = e^{i\beta} X^{-1} B^i X$ with $e^{i\beta}$ a $\mathrm{U}(1)$ phase and $X$ an invertible matrix give the same physical state with the same norm for any system size $L \in \mathbb{N}$. 
\begin{definition}[Gauge equivalence condition of MPS]
We call two MPS representations by $\{A^i\}_i$ and $\{B^i\}_i$ are gauge equivalent $\{A^i\}_i \sim \{B^i\}_i$ if there exists a $\mathrm{U}(1)$ phase $e^{i\alpha_L}$ for any $L\in \mathbb{N}$ such that 
\begin{align}
    \ket{\{A^i\}_i}_L
    = 
    e^{i\alpha_L} \ket{\{B^i\}_i}_L.
\label{eq:ge}
\end{align}
The condition in Eq.(\ref{eq:ge}) is equivalent to 
\begin{align}
\tr[A^{i_1} \cdots A^{i_L}] = e^{i\alpha_L} \tr[B^{i_1} \cdots B^{i_L}]
\end{align}
for any $L \in \mathbb{N}$ and $i_1,\dots,i_L$. 
\end{definition}



To rephrase the gauge equivalence condition in terms of the set of matrices $\{A^i\}_i$, we introduce injective MPSs~\cite{P-GVWC07} described below.
A set of matrices $\{A^i\}_i$ is said to be irreducible if any $A^i$ does not have a proper left invariant subspace, i.e., there is no projector $P$ such that $A^i P = P A^i P$ for any $i$~\cite{P-GVWC07}.
The irreducible condition is equivalent to that the algebra generated by $\{A^i\}_i$, which is spanned by all possible products of matrices $A^{i_1}\cdots A^{i_k}$ with all $k \in \mathbb{N}$, coincides with the set of $n\times n$ matrices ${\rm M}_n(\mathbb{C})$, or in other words, the algebra generated by $A^i$s is central simple.
(See Appendix \ref{CSA} for the definition of central simple.)
Then,
a set of matrices $\{A^i\}_i$ is said to be injective if all possible products of  matrices ${A^{i_{1}}\cdots A^{i_{l}}}$ with {\it a fixed $l$} spans ${\rm M}_n(\mathbb{C})$ \cite{P-GVWC07}.
Obviously, if $\{A^i\}_i$ is injective, $\{A^i\}_i$ is irreducible.\footnote{The converse is not true. 
For example, if we take $A^{\uparrow}=\begin{pmatrix}
&1\\
0&
\end{pmatrix}$ and $A^{\downarrow}=\begin{pmatrix}
&0\\
1&
\end{pmatrix}$, the algebra generated by them is isomorphic to ${\rm M}_2(\mathbb{C})$, and thus they are irreducible.
However, they are not injective because products of odd (even) numbers of them only span $\mathbb{C}\begin{pmatrix}
&1\\
0&
\end{pmatrix}\oplus\mathbb{C}\begin{pmatrix}
&0\\
1&
\end{pmatrix}$ (
$\mathbb{C}\begin{pmatrix}
1&\\
&0
\end{pmatrix}\oplus\mathbb{C}\begin{pmatrix}
0&\\
&1
\end{pmatrix}$ ). }
The injective condition is equivalently expressed 
as the following conditions
for the transfer matrix ${\cal E}_A: {\rm M}_n(\mathbb{C}) \to {\rm M}_n(\mathbb{C})$ defined by ${\cal E}_A(X) := \sum_{i=1}^N A^i X A^{i\dag}$ \cite{SP-GWC10}: 
Let $\rho_A$ be the the maximum of the absolute values of eigenvalues of the transfer matrix ${\cal E}_A$. 
A set of matrices $\{A^i\}_i$ is injective if and only if 
(i) ${\cal E}_A$ has a unique eigenvalue, $\lambda$, with $|\lambda|=\rho_A$, 
and 
(ii) the corresponding eigenvector is a positive definite matrix. 
An injective MPS does not give a superposition of macroscopically different states, and shows exponentially decaying correlation functions. 
It is known that if $\{A^i\}_i$ is injective, the smallest integer $l_{\rm min}$ for which the set of products $\{A^{i_1}\cdots A^{i_l}\}$ spans ${\rm M}_n(\mathbb{C})$ is bounded from above as $l_{\rm min} < (n^2-N+1)n^2$~\cite{SP-GWC10}.

\subsubsection{Fundamental theorem of injective MPS}
\label{ref:bMPS_FT}

The necessary and sufficient conditions for two injective MPSs to give the same physical state are known as the fundamental theorem of MPS~\cite{CP-GSV21}.
Before stating the theorem, it is useful to introduce the canonical form of MPS~\cite{P-GVWC07}. 
When the set of matrices $\{A^i\}_i$ is irreducible, one can normalize $A^i$ so that $\sum_i A^i A^{i\dag}=1_n$ while keeping the physical state unchanged~\footnote{
Let $Y$ be the eigenvector of the transfer matrix ${\cal E}_A$ with the eigenvalue $\lambda=\rho_A$.
The eigenvector $Y$ is positive definite. 
Then, $\tilde A^i = \rho_A^{-1/2} Y^{-1/2}A^i Y^{1/2}$ gives the canonical form.}.
The set of irreducible MPS with $\sum_i  A^i A^{i\dag}=1_n$ is said to be in the canonical form.
Note that the spectral radius of ${\cal E}_A$ 
is 1 when we take the canonical from. 
We start with Theorem 7 in Ref.~\onlinecite{P-GVWC07}~\footnote{
Theorem 7 in Ref.~\onlinecite{P-GVWC07} states the equivalence condition for two MPSs $\{A^i\}_i$ and $\{\tilde A^i\}_i$ as $\ket{\{\tilde A^i\}_i}_L = \ket{\{A^i\}_i}_L$. 
Namely, the equivalence as a vector in the Hilbert space. 
In Theorem~\ref{fthm} of this paper, the equivalence condition is set as the same physical state with the same norm.
}.
\begin{theorem}[\cite{P-GVWC07,P-GWSV08}] \label{fthm_original}
Let a set of $n\times n$ matrices $\{A^{i}\}_i$ be injective and in the canonical form, 
and let another set of $n\times n$ matrices $\{\tilde A^{i}\}_i$ be irreducible and in the canonical form.
Then, the following two statements are equivalent. 
\begin{itemize}
\item[(i)]
Two sets $\{A^i\}_i$ and $\{\tilde A^i\}_i$ represent the same physical state for some length $L>2l+n^4$ in the sense that $\ket{\{\tilde A^i\}_i}_L = e^{i\alpha} \ket{\{A^i\}_i}_L$ holds with a $\mathrm{U}(1)$ phase $e^{i\alpha}$. 
\item[(ii)]
There exist a unitary matrix $V\in\mathrm{U}(n)$ and a $\mathrm{U}(1)$ phase $e^{i\beta} \in\mathrm{U}(1)$ satisfying
\begin{eqnarray}
\tilde A^{i}=e^{i\beta}V^\dag A^{i}V.
\end{eqnarray}
\end{itemize}
Here $l$ is a positive integer for which the set of products $\{A^{i_1} \cdots A^{i_l}\}$ spans ${\rm M}_n(\mathbb{C})$, $V$ is unique up to a $\mathrm{U}(1)$ phase, and $e^{i\beta}$ is unique. 
\end{theorem}
See \cite{P-GVWC07} for the existence of such $V$ and $e^{i\beta}$. 
The uniqueness of $V$ and $e^{i\beta}$ follows from the property that $1_n$ is the nondegenerate eigenvalue of the transfer matrix ${\cal E}_A$, and there is no eigenvalues of magnitude 1~\cite{P-GWSV08}.
As a corollary, we have the following, which we refer to the fundamental theorem for bosonic MPS in this paper. 
\begin{theorem}[Fundamental theorem for bosonic MPS~\cite{P-GVWC07,P-GWSV08}] \label{fthm}
Let $\{A^{i}\}_i$ and $\{B^{i}\}_i$ be injective MPSs in the canonical form. 
They are gauge equivalent $\{A^i\}_i \sim \{B^i\}_i$ if and only if there exist a unitary matrix $V\in\mathrm{U}(n)$ and a $\mathrm{U}(1)$ phase $e^{i\beta} \in\mathrm{U}(1)$ satisfying
\begin{eqnarray}
\tilde A^{i}=e^{i\beta}V^\dag A^{i}V.
\end{eqnarray}
$V$ is unique up to a $\mathrm{U}(1)$ phase, and  $e^{i\beta}$ is unique. 
\end{theorem}

This theorem means that a family of injective MPSs for the same physical state over a parameter space $X$ is a ${\rm U}(1) \times {\rm PU}(n)$ bundle over $X$, where $n$ is the bond dimension and ${\rm PU}(n) ={\rm U}(n)/{\rm U}(1)$ is the projective unitary group of U($n$)~\cite{HMOV02}.
For adiabatic pumps, where $X=S^1$, ${\rm U}(1)\times{\rm PU}(n)$ bundle over $S^1$ is always trivial,
and thus no nontrivial adiabatic pumps exist. 
However, in the presence of onsite symmetry, nontrivial adiabatic pumps are possible, as described in Sec.~\ref{sec:bMPS_G_sym}.

In the rest of this section, we give examples of an injective MPS and a non-injective one, respectively.

Example 1 : The Cluster Model (as an injective case)---
Consider the Hamiltonian of the cluster model on a periodic chain:
 \begin{eqnarray}
 H_{\rm cluster}=\sum_{j=1}^{L}\sigma_{j}^{z}\sigma_{j+1}^{x}\sigma_{j+2}^{z}.
 \end{eqnarray}
This model has a $\zmodtwo\times\zmodtwo$ symmetry generated by $U_{e}=\prod_{j}\sigma^{x}_{2j}$ and $U_{o}=\prod_{j}\sigma^{x}_{2j+1}$. The ground state of this model is unique and gapped, which is given by
 \begin{eqnarray}
 \ket{\mathrm{GS}}&=&\sum_{\sigma_{1},...,\sigma_{L}}\{\prod_{j=1}^{L}s(\sigma_{j},\sigma_{j+1})\}\ket{\sigma_{1},...,\sigma_{L}}\\
 &=&\sum_{\sigma_{1},...,\sigma_{L}}(-1)^{N_{\rm DW}(\sigma_{1},...,\sigma_{L})}\ket{\sigma_{1},...,\sigma_{L}}
 \end{eqnarray}
 where $\sigma_{j}=\uparrow,\downarrow$ and $s(\sigma,\sigma')$ is given by
 \begin{eqnarray}
 s(\sigma,\sigma')=
 \begin{cases}
    -1 & (\sigma,\sigma')=(\uparrow,\downarrow), (\downarrow,\uparrow)\\
    1 & (\sigma,\sigma')=(\uparrow,\uparrow), (\downarrow,\downarrow)
  \end{cases}
 \end{eqnarray}
 and $N_{\rm DW}(\sigma_{1},\cdots,\sigma_{L})$ is the number of $j$ with $|\cdots, \sigma_j, \sigma_{j+1},\cdots\rangle=|\cdots, \downarrow, \uparrow, \cdots\rangle$ or $|\cdots, \uparrow, \downarrow, \cdots\rangle$. 
An MPS of this model is given by 
\begin{eqnarray}
A^{\uparrow}=
\begin{pmatrix}
1&1\\
0&0\\
\end{pmatrix}
,\hspace{5mm}
A^{\downarrow}=
\begin{pmatrix}
0&0\\
-1&1\\
\end{pmatrix}.
\end{eqnarray}
The set of matrices $\{A^i A^j\}_{i,j \in \{\uparrow,\downarrow\}}$ spans ${\rm M}_2(\mathbb{C})$, and thus the MPS is injective. 


 
Example 2 : The Ising model (as a non-injective case)---
 Consider the Hamiltonian of the Ising model:
  \begin{eqnarray}
 H_{\rm Ising}=-\sum_{j\in\mathbb{Z}}\sigma_{j}^{z}\sigma_{j+1}^{z}.
 \end{eqnarray}
 An MPS of this model is given by
 \begin{eqnarray}
A^{\uparrow}=
 \begin{pmatrix}
1&0\\
0&0\\
 \end{pmatrix}
,\hspace{5mm}
A^{\downarrow}=
\begin{pmatrix}
0&0\\
0&1\\
\end{pmatrix}.
\end{eqnarray}
The algebra generated by $A^\uparrow$ and $A^\downarrow$ is not ${\rm M}_2(\mathbb{C})$. 


\subsubsection{Adiabatic cycle of injective MPS with onsite symmetry}
\label{sec:bMPS_G_sym}
Suppose that the total Hilbert space is equipped with a group action of a symmetry group $G$ such that it acts on the local Hilbert space as 
\begin{align}
    \hat U_g \ket{j} = \ket{i} \sum_j [u_g]_{ij},\quad 
    g \in G, 
\end{align}
where $u_g$ are $\mathrm{U}(n)$ matrices. 
Some elements of $G$ can be antiunitary and are specified by a homomorphism $s: G \to \{\pm 1\}$ so that $\hat U_g i \hat U_g^{-1} = s_g i$. 
We also require $\hat U_g$ to be a linear representation, namely, $u_g u_h^{s_g}=u_{gh}$ holds. 
Here we have introduced the notation 
\begin{align}
    X^{s_g}
    =\left\{\begin{array}{ll}
    X     & s_g=1, \\
    X^*     & s_g=-1, 
    \end{array}\right.
\end{align}
for matrices $X$. 

Now suppose that an injective MPS $\{A^i\}_i$ in the canonical form preserves the $G$ symmetry in the sense that for any system size $L \in \mathbb{N}$ there exists $e^{i\alpha_L} \in {\rm U}(1)$ such that
\begin{align}
    \hat U_g \ket{\{A^i\}}_L
    = e^{i\alpha_L} \ket{\{A^i\}}_L. 
\end{align}
This is equivalent to say $\{ \sum_j [u_g]_{ij} (A^j)^{s_g}\}_i \sim \{A^i\}_i$ for any $g \in G$ as a bosonic MPS. 
From the fundamental theorem for bosonic MPS, there exists a unique $\mathrm{U}(1)$ phase $e^{i\beta_g}$ and a $\mathrm{U}(n)$ matrix $V_g$ such that 
\begin{align}
\sum_j [u_g]_{ij}(A^j)^{s_g} = e^{i\beta_g}V_g^\dag A^i V_g
\end{align}
for $g \in G$, where $V_g$ is unique up to $\mathrm{U}(1)$ phases. 
The linearity of $u_g$ and the uniqueness of $e^{i\beta_g}$ and $V_g$ implies that $e^{i\beta_g}e^{is_g \beta_h} = e^{i\beta_{gh}}$ and there exists $z_{g,h} \in {\rm U}(1)$ such that  
\begin{align}
    V_gV_h^{s_g} = z_{g,h} V_{gh}
\end{align}
for $g,h\in G$. 
The equality $V_g (V_h V_k^{s_h})^{s_g} = (V_gV_h^{s_g})V_k^{s_{gh}}$ implies that $z_{g,h}$ is a two cocycle $z \in Z^2(G,{\rm U}(1)_s)$ as it satisfies the two-cocycle condition 
\begin{align}
    z_{h,k}^{s_g} z_{gh,k}^{-1}z_{g,hk} z_{g,h}^{-1}=1 
\end{align}
for $g,h,k \in G$, 
where ${\rm U}(1)_s$ means the left module defined as $g.z = z^{s_g}$ for $g \in G$ on $z \in {\rm U}(1)$. 

Under these preparations, we consider a loop, parameterized by $\theta \in [0,2\pi]$, of injective MPS $\{A^i(\theta)\}_i$ in the canonical form and with $G$ symmetry which starts and ends at the same physical state in the sense that $\{A^i(2\pi)\}_i \sim \{A^i(0)\}_i$. 
Along the loop, the $G$ action on the Hilbert space is assumed to be in common. 
As mentioned in Sec.~\ref{ref:bMPS_FT}, there always exists a global gauge so that $A^i(2\pi)=A^i(0)$ holds, however, the following discussion does not change irrespective of whether $A^i(2\pi) = A^i(0)$ holds or not.
For $\theta \in [0,2\pi]$, we have $\mathrm{U}(1)$ matrices $e^{i\beta_g(\theta)}$ and U($n$) phases $V_g(\theta)$ from the relations 
\begin{align}
    \sum_j [u_g]_{ij}[A^j(\theta)]^{s_g} = e^{i\beta_g(\theta)} V_g(\theta)^\dag A^i(\theta) V_g(\theta)
\end{align}
for $g \in G$. 
From $V_g(\theta)$s, we have a parameter family of two-cocycle $z_{g,h}(\theta) \in Z^2(G,{\rm U}(1)_s)$ which may not be $2\pi$-periodic $z_{g,h}(2\pi) \neq z_{g,h}(0)$ but relates with each other with a one-coboundary. 
To see this, applying the fundamental theorem to $\{A^i(2\pi)\}$ and $\{A^i(0)\}$, we get 
\begin{align}
    A^i(2\pi) = e^{i\gamma} W^\dag A^i(0) W 
\end{align}
with $e^{i\gamma} \in {\rm U}(1)$ and $W \in {\rm U}(n)$. 
The $G$ action on both the sides leads to the equality 
\begin{align}
    A^i(2\pi)
    =
    e^{-i\beta_g(2\pi)}e^{is_g\gamma}e^{i\beta_g(0)}V_g(2\pi)[W^\dag]^{s_g} V_g(0)^\dag
    A^i(0)
    V_g(0)W^{s_g} V_g(2\pi)^\dag
\end{align}
for $g \in G$.
Then the uniqueness of $e^{i\beta}$ and $W$ gives us 
\begin{align}
    e^{i\beta_g(2\pi)}=e^{i(s_g\gamma-\gamma)} e^{\beta_g(0)},
\end{align}
and 
\begin{align}
    V_g(2\pi)=e^{i\phi_g} W^\dag V_g(0) W^{s_g}
\end{align}
with $e^{i\phi_g} \in {\rm U}(1)$. 
Therefore, we have 
\begin{align}
    z_{g,h}(2\pi)
    =
    e^{is_g\phi_h}e^{-i\phi_{gh}}e^{i\phi_g} z_{g,h}(0).
    \label{eq:mps_G_z_bc}
\end{align}
Introducing a lift $C^1(G,\mathbb{R}_s/2\pi \mathbb{Z}_s) \ni \phi_{g} \mapsto \tilde \phi_{g} \in C^1(G,\mathbb{R}_s)$, we define the following quantity 
\begin{align}
    n_{g,h}
    =
    \frac{1}{2\pi} (\delta \tilde \phi)_{g,h}
    -\frac{1}{2\pi i} \int_0^{2\pi} d \log z_{g,h}(\theta). 
\end{align}
The relation (\ref{eq:mps_G_z_bc}) implies that $n_{g,h}$ is a two-cocycle of $Z^2(G,\mathbb{Z}_s)$. 
The change of lift $\tilde \phi_{g} \to \tilde \phi_{g} + 2\pi b_{g}$ with one-cochains $b_{g} \in C^1(G,\mathbb{Z}_s)$ gives the shift $n_{g,h} \mapsto n_{g,h} + (\delta b)_{g,h}$.
Therefore, $n_{g,h}$ is well-defined only as a cohomology group $H^2(G,\mathbb{Z}_s)$. 
It is also shown that $n$ is invariant under changes of the $\mathrm{U}(1)$ phases of $V_g(\theta)$.
Therefore, $[n] \in H^2(G,\mathbb{Z}_s) \cong H^1(G,{\rm U}(1)_s)$ is a topological invariant of adiabatic pumps.

We comment on the simplification of the topological invariant $[n]$ for the two cases. 
When the two-cocycle $z_{g,h}(\theta)$ is $2\pi$-periodic as $e^{i\phi_g} \equiv 1$, $n_{g,h}$ is recast as the phase winding $n_{g,h} = - \frac{1}{2\pi i} \oint d \log z_{g,h}(\theta)$~\cite{Shiozaki21}.
When the two-cocycle $z_{g,h}(\theta)$ is constant for $\theta$, then (\ref{eq:mps_G_z_bc}) means that $e^{i\phi_g}$ is a one-dimensional representation of $G$, and this is nothing but the topological invariant $[e^{i\phi}] \in H^1(G,{\rm U}(1)_s)$.

\subsection{Fermionic MPS}\label{fmps}

Fermionic MPSs (fMPSs) were first introduced in [\onlinecite{FK11}] and developed in [\onlinecite{BWHV17}]. 
In this paper, we adapt the formulation in \cite{BWHV17,PTSC21}. 
See also \cite{PhysRevB.104.075146} on an application of fMPS to Lieb-Schultz-Mattis type theorems for Majorana fermion systems. 
For the classification of $1+1$-dimensional fermionic SRE states without restricting SRE states to the class of the fMPS, see \cite{bourne_ogata_2021}.

\subsubsection{Preliminary}
\label{sec:f_pre}

Let us consider a 1-dimensional fermionic system with $L$ sites. 
We denote the creation/annihilation operator of complex fermions at cite $k$ by $a^\dag_{k,f}/a_{k,f}$ for $f=1,\dots,N_F$, where $N_F$ is the number of flavors.~\footnote{Note that if spin degrees of freedom coexist, they can be regarded as internal degrees of freedom of complex fermions.}
We introduce the following shorthand notations: 
\begin{align}
&i_k = (i_{k,1},\dots,i_{k,N_F}) \in \{0,1\}^{\times N_F},\\
&|i_k| = \sum_{f=1}^{N_F} |i_{k,f}| \ \ ({\rm mod\ }2) \in \{0,1\},\\
&(a^\dag_k)^{i_k}=(a_{k_1}^\dag)^{i_{k,1}} \cdots (a_{k_F}^\dag)^{i_{k,N_F}}, \\
&\ket{i_1\cdots i_L} = (a_1^\dag)^{i_1} \cdots (a_L^\dag)^{i_L} \ket{0}, 
\end{align}
where $\ket{0}$ is the Fock vacuum defined by $a_{k,f}\ket{0}=0$ for all possible $k$s and $f$s.
The index $i$ per site runs over $N=2^{N_F}$ possible combinations of $(i_{1},\dots,i_{N_F})$.
Similarly to the bosonic MPS, we also denote it by $i \in \{1,\dots,N\}$ equipped with the definite parity $|i| \in \{0,1\}$. 
A fermionic state is written as $\ket{\psi} = \sum_{\{i_k\}} \psi(i_1,\dots,i_L) \ket{i_1\cdots i_L}$ with $\psi(i_1,\dots,i_L)$ the wave function in the occupation basis.
We define the fermion parity operator $(-1)^F$ by $(-1)^F = \prod_{k,f} (-1)^{a^\dag_{k,f} a_{k,f}}$, and
assume that $\ket{\psi}$ has a definite fermion parity $(-1)^F\ket{\psi}=(-1)^{|\psi|}\ket{\psi}$, $(-1)^{|\psi|}\in \{\pm 1\}$,  
which implies that only wave functions $\psi(i_1,\dots,i_L)$ satisfying the constraint $(-1)^{\sum_{j=1}^L |i_j|} = (-1)^{|\psi|}$ can be nonzero.
We also define the translation operators for the periodic boundary condition (PBC) and the anti-periodic boundary condition (APBC) by 
\begin{align}
&T_{\rm P} a^\dag_{k,f} T_{\rm P}^{-1} = a^\dag_{k+1,f} \mbox{ for } k=1,\dots,L-1,\quad 
T_{\rm P} a^\dag_{L,f} T_{\rm P}^{-1} = a^\dag_{1,f},\quad 
T_{\rm P}\ket{0} = \ket{0}, 
\end{align}
and 
\begin{align}
&T_{\rm AP} a^\dag_{k,f} T_{\rm AP}^{-1} = a^\dag_{k+1,f} \mbox{ for } k=1,\dots,L-1,\quad 
T_{\rm AP} a^\dag_{L,f} T_{\rm AP}^{-1} = -a^\dag_{1,f},\quad 
T_{\rm AP}\ket{0} = \ket{0}, 
\end{align}
respectively, and 
assume that $\ket{\psi}$ is invariant under translation both under PBC and APBC,
\begin{align}
T_{\rm P} \ket{\psi}_{\rm P} = \ket{\psi}_{\rm P}
\quad \Leftrightarrow \quad
\psi(i_1,i_2,\dots,i_L)
=(-1)^{|i_1|(|\psi|+1)}\psi(i_2,\dots,i_L,i_1), 
\label{eq:psi_Tp}
\end{align}
and 
\begin{align}
T_{\rm AP} \ket{\psi}_{\rm AP} = \ket{\psi}_{\rm AP}
\quad \Leftrightarrow \quad
\psi(i_1,\dots,i_L)
=(-1)^{|i_1||\psi|}\psi(i_2,\dots,i_L,i_1)
\label{eq:psi_Tap}
\end{align}
for any system size $L$
\footnote{For translation invariant fermionic states with fermion parity per site, irrespective of the $\zmod{2}$-class, the Bloch momentum depends on the number $L$ of sites as in $T_{\rm P} \ket{\psi}_{P} = (-1)^{L-1} \ket{\psi}_{\rm P}$ and $T_{\rm AP} \ket{\psi}_{AP} = (-1)^L \ket{\psi}_{\rm AP}$.
Even when this is the case, by regarding two sites as a unit cell, the fermion parity per site can always be removed. 
Therefore, the assumption in the main text should not lose the generality for the purpose of studying pumps of fermionic states.}.
Fermionic unique gapped ground states in 1+1 dimension are classified by $\zmod{2}$~\cite{Kitaev01}, and when they are translational invariant, $T_{\rm P/AP} \ket{\psi}_{\rm P/AP}=\ket{\psi}_{\rm P/AP}$, the $\zmod{2}$ class is detected by the fermion parity under PBC~\cite{Kitaev01}: 
\begin{align}
(-1)^F \ket{\psi}_{\rm P} = \pm \ket{\psi}_{\rm P}.
\label{eq:psi_p_fp}
\end{align} 
For APBC, irrespective of the $\zmod{2}$-class, the fermion parity of the ground state is always even, i.e., 
\begin{align}
(-1)^F \ket{\psi}_{\rm AP} = \ket{\psi}_{\rm AP}.    
\label{eq:psi_ap_fp}
\end{align}

\subsubsection{Fermionic MPS}
\label{sec:fMPS_def}

We introduce translation invariant fMPSs in such a way that (\ref{eq:psi_Tp}), (\ref{eq:psi_Tap}), (\ref{eq:psi_p_fp}), and (\ref{eq:psi_ap_fp}) are satisfied. 

In the case of fermionic systems, the local Hilbert space is $\zmod{2}$-graded by the fermion parity
\begin{eqnarray}
\mathfrak{h}_{k}=\mathfrak{h}^{(0)}_{k}\oplus\mathfrak{h}^{(1)}_{k},
\end{eqnarray}
where the superscripts (0) and (1) indicate the even and odd fermion parities, respectively. 
The femrion anti-commutation relation implies that this space has the $\zmod{2}$-graded tensor product which is a tensor product with the non-trivial braiding rule
\begin{eqnarray}
v\otimes_{gr} w=(-1)^{\abs{v}\abs{w}}w\otimes_{gr}v,
\end{eqnarray}
where $v\in V^{(\abs{v})}$ and $w\in W^{(\abs{w})}$ are elements of $\zmod{2}$-graded vector spaces $V$ and $W$ with the fermion parities $\abs{v}$ and $\abs{w}$. We call the basis that diagonalizes $\zmodtwo$-grading the standard basis. 
The total Hilbert space is the $\zmod{2}$-graded tensor product space ${\cal H} = \bigotimes_{k}\mathfrak{h}_{k}$. 

For fMPSs, not only the physical Hilbert space, but also the bond Hilbert space $V$ is $\zmod{2}$-graded $V = V^{(0)} \oplus V^{(1)}$. 
$V^{(0)}$ and $V^{(1)}$ need not be of the same dimension, but to represent a nontrivial adiabatic pump, we assume $\dim V^{(0)} = \dim V^{(1)}=n$.
We introduce the grading matrix $Z$ such that $Z v = (-1)^{i}v$ for $v \in V^{(i)}\; (i=0,1)$. 
It holds that $Z^2=1_{2n}$ and $Z^\dag = Z$. 
To implement the fermion parity in fMPSs, we impose the following constraint on matrices $\{A^i\}_i$,  
\begin{align}
(-1)^{|i|} A^i = Z A^i Z 
\label{eq:Ai_even}
\end{align}
in the standard basis of the physical Hilbert space. 
In the basis of the bond Hilbert space $V$ such that $Z =\sigma_z \otimes 1_n$, the matrices $A^i$ are in the forms as 
\begin{eqnarray}
A^{i}=
\begin{pmatrix}
B^{i}&\\
&C^{i}\\
\end{pmatrix}\;\mathrm{for}\;\abs{i}=0,\\
A^{i}=
\begin{pmatrix}
&B^{i}\\
C^{i}&\\
\end{pmatrix}\;\mathrm{for}\;\abs{i}=1.
\end{eqnarray}
Then, an fMPS for $\{A^i\}_i$ is introduced as 
\begin{eqnarray}
\ket{\{A^{i}\}_i,\Omega}_{L}
:=\sum_{\{i_{k}\}}\tr[\Omega A^{i_{1}}\cdots A^{i_{L}}]\ket{i_{1},...,i_{L}}\in\mathcal{H},
\end{eqnarray}
where $\Omega$ is a matrix only with virtual legs, called the boundary matrix~\cite{BWHV17} to be determined depending on the boundary condition and the $\zmod{2}$ class of the state.
For the fMPS to have a definite fermion parity, $\Omega$ is taken as $Z \Omega = \pm \Omega Z$, i.e. $Z\Omega=\Omega Z$ for even parity, $Z\Omega=-\Omega Z$ for odd parity.
The translation invariance further constrains $\Omega$. 
Under APBC, from (\ref{eq:psi_ap_fp}), the fermion parity of fMPSs should be even, so the translation invariance in (\ref{eq:psi_Tap}) yields ${\rm tr}[\Omega A^{i_1}\cdots A^{i_L}]
={\rm tr}[\Omega A^{i_2}\cdots A^{i_L}A^{i_1}]$. This condition is satisfied for the trivial boundary matrix $\Omega=1_{2n}$.
However, under PBC, from (\ref{eq:psi_p_fp}), the fermion parity of fMPSs depends on the $\zmod{2}$ class, so does $\Omega$, which is determined so to obey the translation invariance in (\ref{eq:psi_Tp}). 
We give the explicit construction of $\Omega$ for PBC in Sec.\ref{sec:irr_fMPS}.

We here settle what kind of space we regard as the set of fMPSs. 
We abbreviate the unitary equivalence of two matrices $A$ and $B$ to $A \sim_{\rm u} B$. 
Namely, $A \sim_{\rm u} B$ means there exists a unitary matrix $U$ such that $A = U^\dag B U$. 

\begin{definition}[Space of fMPS]
For a given local Hilbert space including fermions spanned by a standard basis $\ket{i}$, we define the space $\tilde {\cal M}^{\rm f}_{n}$ of fMPSs with bond dimension $2n$ as sets of $2n \times 2n$ matrices $\{A^i\}_i$ equipped with a grading matrix defined by (\ref{eq:Ai_even}). 
Explicitly, 
\begin{align}
    \tilde {\cal M}^{\rm f}_n:= 
    \{\{A^i\}_i| 
    {}^\exists Z \sim_{\rm u} \sigma_z \otimes 1_n, {\rm\  s.t.\ }
    (-1)^{|i|} A^i = Z A^i Z\}.
\end{align}
\end{definition}
We note that the grading matrix $Z$ is not unique in general. 
We introduce the gauge equivalence condition in the space of fMPSs as the equivalence of physical states in APBC.
\begin{definition}[Gauge equivalence condition of fMPS]
We define that $\{A^i\}_i$, $\{B^i\}_i \in \tilde {\cal M}^{\rm f}_n$ are gauge equivalent $\{A^i\}_i \sim \{B^i\}_i$ if  there exists a $\mathrm{U}(1)$ phase $e^{i\alpha_L}$ for any $L\in \mathbb{N}$ such that 
\begin{align}
    \ket{\{A^i\}_i,1_{2n}}_L
    = 
    e^{i\alpha_L} \ket{\{B^i\}_i,1_{2n}}_L
\end{align}
holds.
This is equivalent to the wave function equalities
\begin{align}
\tr[A^{i_1} \cdots A^{i_L}] = e^{i\alpha_L} \tr[B^{i_1} \cdots B^{i_L}]
\end{align}
for any $L \in \mathbb{N}$ and $i_1,\dots,i_L$. 
\end{definition}
Note that this definition does not depend on the $\zmod{2}$ class of states, although as we will see later, the boundary matrix $\Omega$ in PBC depends on the $\zmod{2}$ class.

\subsubsection{Irreducible fMPS}
\label{sec:irr_fMPS}
In the next section, we introduce the injectivity of fMPS to represent fermionic unique gapped ground states with finite range correlation. 
Before doing so, it would be helpful to introduce the irreducibility of fMPSs~\cite{BWHV17,PTSC21} as a necessary condition for the injectivity, in view of its relation to the graded algebra.

A set of matrices $\{A^i\}_i \in \tilde {\cal M}^{\rm f}_n$ generates an algebra ${\cal A}$ that is spanned by linear summations of products $c_{i_1,\dots,i_l} A^{i_1} \cdots A^{i_l}$ with $c_{i_1,\dots,i_l} \in \mathbb{C}$ for $l \in \mathbb{N}$. 
The algebra ${\cal A}$ is a graded algebra ${\cal A} = {\cal A}^{(0)} \oplus {\cal A}^{(1)}$ of which the grading is defined by $\sum_{k=1}^l |i_k| \,(\mbox{mod.2}) \in \{0,1\}$ for elements $A^{i_1} \cdots A^{i_l}$~\footnote{Here, we impose that $\mathcal{A}$ has the unit element and $\mathcal{A}^{(0)}$ and $\mathcal{A}^{(1)}$ are non-zero.}. 
Thus, ${\cal A}^{(0)} = {\rm Span}(\{A^{i_1}\cdots A^{i_l} | \sum_{k=1}^l |i_k| \equiv 0 \,(\mbox{mod.2}), l \in \mathbb{N}\})$, and ${\cal A}^{(1)} = {\rm Span}(\{A^{i_1}\cdots A^{i_l} | \sum_{k=1}^l |i_k| \equiv 1 \,(\mbox{mod.2}),l \in \mathbb{N}\})$. 
We note that grading matrices $Z$ can be used to detect even and odd elements of ${\cal A}$.
If $Z a = a Z$ ($Z a = - a Z$) for $a \in {\cal A}$, then $a \in {\cal A}^{(0)}$ ($a \in {\cal A}^{(1)}$). 

We define a set of matrices $\{A^i\}_i \in \tilde {\cal M}^{\rm f}_n$ to be graded irreducible if the graded algebra ${\cal A}$ generated by $\{A^i\}_i$ is graded central simple~\cite{BWHV17}.
(We give a brief review on the graded central simple algebra in Appendix \ref{GCSA}.)
Note that the above definition does not depend on choices of grading matrices $Z$. 
It is known that graded central simple algebras are classified into two types:  $(+)$-algebra and $(-)$-algebra~\cite{Wall64}.
For each type of algebra, there is a characteristic matrix $u$ which is essentially unique up to a sign. 
We call the type of algebra the Wall invariant, and $u$ the Wall matrix. 

If $\mathcal{A}$ is a $(+)$-algebra, then $\mathcal{A}$ is central simple as an ungraded algebra. 
In addition, there is a unique element $u\in Z(\mathcal{A}^{(0)})$ up to a phase factor so that $u^{2}$ is proportional to $1$ and $u$ itself is not proportional to $1$ \cite{Wall64}. 
Here, $Z(\mathcal{A}^{(0)})$ is the center of $\mathcal{A}^{(0)}$.
In terms of the matrices $\{A^i\}_i \in \tilde {\cal M}^{\rm f}_n$, the ungraded central simplicity of ${\cal A}$ is rephrased as that the set of all possible products of matrices $A^{i_1} \cdots A^{i_l}$ span the vector space ${\rm M}_{2n}(\mathbb{C})$.
This is equivalent to the absence of left invariant subspace of $A_i$s, i.e., there is no proper projectors $P$ such that $A^iP=PA^iP$ holds for all $i$. 
The matrix $u$ should be proportional to the grading matrix $Z$ since $Z \in Z({\cal A}^{(0)})$. 
We note that the grading matrix $Z$ is unique up to a sign. 
In fact, if $Z'$ is also a grading matrix satisfying (\ref{eq:Ai_even}), then $ZZ' a = a ZZ'$ holds true for any $a \in {\cal A} \cong {\rm M}_{2n}(\mathbb{C})$. 
From the Schur's lemma we have $ZZ' \propto 1$, thus, $Z' = \pm Z$. 
Along the line of thought above, we define the graded irreducibility with the Wall invariant (+) as follows. 
\begin{definition}[Irreducible fMPS with Wall invariant (+)]
A set of matrices $\{A^i\}_i \in \tilde {\cal M}^{\rm f}_n$ is graded irreducible with the Wall invariant (+) if the set of products $A^{i_1} \cdots A^{i_l}$ with all possible $l \in \mathbb{N}$ and $i_1,\cdots,i_l$ span the vector space ${\rm M}_{2n}(\mathbb{C})$.
The grading matrix $Z$ is unique up to a sign, and the Wall matrix $u$ is given by $u = \pm Z$.
\end{definition}

If $\mathcal{A}$ is a $(-)$-algebra, then ${\cal A}$ is not ungraded central simple but $\mathcal{A}^{(0)}$ is ungraded central simple.
In addition, there is a unique element $u\in Z(\mathcal{A})\cap\mathcal{A}^{(1)}$ up to a phase factor so that $u^{2}$ is proportional to $1$ and $\mathcal{A}^{(0)}u=\mathcal{A}^{(1)}$ \cite{Wall64}. 
In terms of matrices $\{A^i\}_i$, ${\cal A}$ being the $(-)$-algebra implies that there is a proper left invariant subspace of $A^i$s. 
Let $S_1$ be the left invariant subspace that contains no smaller left invariant subspaces of $A^i$s, and let $P$ be the orthogonal projector onto $S_1$. 
$P$ satisfies $P^2=P, P^\dag = P$, and $A^i P = P A^i P$ for all $i$. 
Let $Z$ be a grading matrix satisfying (\ref{eq:Ai_even}). 
We find that $[P,Z] \neq 0$, otherwise, ${\cal A}$ is not graded central simple.
From (\ref{eq:Ai_even}), $Q=ZPZ$ is also an orthogonal projector onto a different ungraded left invariant subspace $S_2=ZS_1$. 
It is found that $PQ=0$, $P+Q = 1_{2n}$, $P-Q$ is unitary, the matrix $u$ is given explicitly by $u = P-Q = 2P-1$, and $[u,A^i]=0$ for all $i$~\cite{BWHV17,PTSC21}. 
In the basis where $Z = \sigma_z \otimes 1_n$, solving $u^\dag = u, u^2=1, \{Z,u\}=0$, $u$ can be written as 
\begin{align}
    u = \begin{pmatrix}
    &U^\dag\\
    U\\
    \end{pmatrix}
\end{align}
with $U \in {\rm U}(n)$ a unitary matrix. 
Solving $A^i P = PA^i P$, we have 
\begin{eqnarray}
A^i 
= 
\begin{pmatrix}
1\\
&U\\
\end{pmatrix}
(\sigma_x^{|i|} \otimes B^i )
\begin{pmatrix}
1\\
&U^\dag\\
\end{pmatrix}
\end{eqnarray}
where $B^{i}$s are $n\times n$ matrices. 
The condition that ${\cal A}^{(0)}$ is central simple and ${\cal A}^{(1)} \cong {\cal A}^{(0)} u$ can be expressed for the matrices $\{B^i\}_i$: 
Both the even subalgebra $\{B^{i_1}\cdots B^{i_l} | \sum_{k=1}^l |i_k| \equiv 0 \,(\mbox{mod.2}), l \in \mathbb{N}\}$ and the odd subalgebra $\{B^{i_1}\cdots B^{i_l} | \sum_{k=1}^l |i_k| \equiv 1 \,(\mbox{mod.2}), l \in \mathbb{N}\}$ span the vector space ${\rm M}_n(\mathbb{C})$.
Along the line of thought above, we define the graded irreducibility with the Wall invariant $(-)$ as follows. 
\begin{definition}[Irreducible fMPS with Wall invariant $(-)$]
\label{def:fMPS_irr_-}
A set of matrices $\{A^i\}_i \in \tilde {\cal M}^{\rm f}_n$ is graded irreducible with the Wall invariant $(-)$ if there exists a unitary matrix $u\in {\rm U}(2n)$ such that the following two conditions are fulfilled: 
(i) $u$ is unitary equivalent to $\sigma_x \otimes 1_n$. 
(ii) Let $W \in {\rm U}(n)$ be a unitary matrix that diagonalizes $u$ as $u = W (\sigma_x \otimes 1_n) W^\dag$. 
Then, the matrices $A^i$ are written as 
\begin{align}
    A^i = W (\sigma_x^{|i|} \otimes B^i) W^\dag,
\end{align}
and both the even subalgebra $\{B^{i_1}\cdots B^{i_l} | \sum_{k=1}^l |i_k| \equiv 0 \,(\mbox{mod.2}), l \in \mathbb{N}\}$ and the odd subalgebra $\{B^{i_1}\cdots B^{i_l} | \sum_{k=1}^l |i_k| \equiv 1 \, (\mbox{mod.2}), l \in \mathbb{N}\}$ span the vector space ${\rm M}_n(\mathbb{C})$.
\end{definition}
The matrix $u$ is unique up to a sign, as shown below. 
Suppose that $\tilde u = \tilde W (\sigma_x \otimes 1_n) \tilde W^\dag$ is another Wall matrix in Definition \ref{def:fMPS_irr_-}. 
The matrices $A^i$ can also be written as $A^i = \tilde W (\sigma_x^{|i|} \otimes \tilde B^i) \tilde W^\dag$, and the even and odd subalgebras generated by $\{\tilde B^i\}_i$ span ${\rm M}_n(\mathbb{C})$. Set $X = W^\dag \tilde W$.
We have 
\begin{align}
    &X (1_2 \otimes \tilde b)= (1_2 \otimes b) X, \label{eq:X_even}\\
    &X (\sigma_x \otimes \tilde b)= (\sigma_x \otimes b) X \label{eq:X_odd}
\end{align}
for any $b \in {\rm M}_n(\mathbb{C})$.
Let us write $X$ in a block form $X= \begin{pmatrix}
    x&y\\
    z&w\\
    \end{pmatrix}$.
If $x$ is invertible, (\ref{eq:X_even}) leads to $\tilde b = x^{-1} b x$, and from the Schur's lemma we have $y = \lambda_2 x, z = \lambda_3 x, z = \lambda_4 x$ with $\lambda_2,\lambda_3,\lambda_4 \in \mathbb{C}$. 
Substituting them into (\ref{eq:X_odd}), we get $\lambda_4 = \eta$ and $\lambda_3 = \eta \lambda_2$ with $\eta$ a sign $\eta \in \{\pm 1\}$. 
When $x$ is noninvertible, one can show $X$ is in the form $X = \begin{pmatrix}
0&y \\
\eta y&0
\end{pmatrix}$ with $\eta \in \{\pm 1\}$. 
Therefore, $X$ can eventually be written as $X = \begin{pmatrix}
\lambda_1 x & \lambda_2 x \\
\eta \lambda_2 x &\eta \lambda_1 x& 
\end{pmatrix}$ with $\lambda_1,\lambda_2 \in \mathbb{C}$ and $\eta \in \{\pm 1\}$. 
We conclude that $X \sigma_x= \eta \sigma_x X$, and thus, $\tilde u = \eta u$. 

The sign ambiguity of $u$ is the origin of the $\zmod{2}$-nontrivial pump in the $(-)$-algebra. 
To see this, let us consider a periodic one-parameter family of the set of graded irreducible matrices $\{A^i(\theta)\}$ in the basis such that $Z = \sigma_z \otimes 1_n$. 
We have two orthogonal projectors $P(\theta)$ and $Q(\theta)$ which also depend on $\theta$. 
By a one cycle $\theta = 2\pi$, the projector $P(2\pi)$ is equal to either $P(0)$ or $Q(0)$, and the latter case indicates a nontrivial pump. 

Now we can represent a translation invariant fMPS with PBC by using the Wall matrix $u$ as the boundary operator regardless of the type of the algebra $(+)$ and $(-)$\cite{BWHV17}:
\begin{eqnarray}
\ket{\{A^i\}_i,u}_{L}=\sum_{\{i_{k}\}}\tr(u A^{i_{1}}\cdots A^{i_{L}})\ket{i_{1},...,i_{L}}\in\mathcal{H}=\bigotimes_{k}\mathfrak{h}_{k}. 
\end{eqnarray}
It is easy to show $T_P \ket{\{A^i\}_i,u}_L=\ket{\{A^i\}_i,u}_L$. 
Remark that an MPS given by a $(+)$-algebra is fermion parity even and an MPS given by a $(-)$-algebra is fermion parity odd. 
Note also that although the Wall matrix $u$ has sign ambiguity in general, it only affects the MPS by an overall sign, so the physical state is uniquely determined. 
We also introduce another equivalence condition in the space of fMPSs as the equivalence of physical states in PBC.
\begin{definition}[Gauge equivalence condition of irreducible fMPS in PBC]
We define that two irreducible fMPSs $\{A^i\}_i,\{\tilde A^i\}_i \in \tilde {\cal M}^{\rm f}_n$ are gauge equivalent in PBC $\{A^i\}_i \sim_{\rm PBC} \{\tilde A^i\}_i$ if
there exists a $\mathrm{U}(1)$ phase $e^{i\alpha_L}$ for any $L \in \mathbb{N}$ 
such that 
\begin{align}
    \ket{\{A^i\}_i,u}_L
    = 
    e^{i\alpha_L} \ket{\{\tilde A^i\}_i,\tilde{u}}_L
\end{align}
holds. 
Here, $u$ and $\tilde u$ are Wall matrices for $\{A^i\}_i$ and $\{\tilde A^i\}_i$, respectively. 
This is equivalent to the wave function equalities
\begin{align}
\tr[u A^{i_1} \cdots A^{i_L}] = e^{i\alpha_L} \tr[\tilde u \tilde A^{i_1} \cdots \tilde A^{i_L}]
\end{align}
for any $l \in \mathbb{N}$ and $i_1,\dots,i_L$. 
\end{definition}

\subsubsection{Injective fMPS and fundamental theorem}
\label{sec:inj_fmps}

For each type of the Wall invariant, we further impose the following graded injectivity on graded irreducible fMPSs as follows, which we call the injective fMPS. 
One can show that an injective fMPS is essentially unique up to conjugate transformations. First, we discuss the case of $(+)$-algebra.
\begin{definition}[Injective fMPS with Wall invariant (+)]
A set of matrices $\{A^i\}_i \in \tilde {\cal M}^{\rm f}_n$ is graded injective with the Wall invariant (+) if the set of all possible products of matrices $A^{i_1} \cdots A^{i_l}$ with a fixed $l\in \mathbb{N}$ spans the vector space ${\rm M}_{2n}(\mathbb{C})$.
The grading matrix $Z$ is unique up to a sign, and the Wall matrix $u$ is given by $u = \pm Z$.
\end{definition}
Before moving on to the fundamental theorem of injective fMPS, it is useful to introduce the canonical form of fMPS. 
Since $\{A^i\}_i$ is also ungraded irreducible, one can normalize $\{A^i\}$ to be the canonical form~\cite{P-GVWC07}
\begin{align}
    \sum_i A^i A^{i\dag} = 1_{2n}.
    \label{eq:+_c_form}
\end{align}
\begin{theorem}(fundamental theorem for fMPS with Wall invariant $(+)$)\label{fthmplus}\\
Let $\{A^{i}\}_i$ and $\{\tilde{A}^{i}\}_i$ be injective fMPSs with the Wall invariant $(+)$ in the canonical form (\ref{eq:+_c_form}).
They give the same physical state in APBC, in other words, $\{A^i\}_i \sim \{\tilde A^i\}_i$ holds if and only if there exist a unitary matrix $V\in\mathrm{U}(2n)$ and a $\mathrm{U}(1)$ phase $e^{i\beta}\in\mathrm{U}(1)$ obeying that
\begin{eqnarray}\label{fmpseqplusA}
\tilde A^{i}=e^{i\beta}V^\dag {A}^{i}V. 
\label{eq:fthm_fmps_+}
\end{eqnarray}
The unitary matrix $V$ is unique up to $\mathrm{U}(1)$ phase, and $e^{i\beta}$ is unique. 

Furthermore, if we take the Wall matrices for $\{A^i\}_i$ and $\{\tilde A^i\}_i$ to be $u$ and $\tilde u$, respectively, they are connected by
\begin{align}
    \tilde u = \eta V^\dag u V
    \label{eq:fthm_fmps_+_Z}
\end{align}
with $\eta \in \{\pm 1\}$ a sign. 

This theorem holds even if the assumption $\{A^i\}_i\sim\{\tilde{A}^i\}_i$ is changed to $\{A^i\}_i\sim_{\rm PBC}\{\tilde{A}^i\}_i$.
\end{theorem}
The former part is the same as the fundamental theorem of injective MPS in Theorem~\ref{fthm}, since $\{A^i\}_i$ and $\{\tilde A^i\}_i$ are also ungraded injective. 
It is easy to show Eq.(\ref{eq:fthm_fmps_+_Z}) as follows. 
Substituting (\ref{eq:fthm_fmps_+}) into the relation $(-1)^{|i|} \tilde A^i = \tilde u \tilde A \tilde u$, we have $\tilde A^i = e^{i\beta} \tilde u V^\dag u A^i u V \tilde u$. 
The uniqueness of $V$ and $\tilde u^2 = 1_{2n}$ gives us (\ref{eq:fthm_fmps_+_Z}). 
For a proof of PBC, see Appendix \ref{sec:pf_thm_2}.

We note that the sign $\eta$ depends on the choice of signs of $u$ and $\tilde u$. 
Replacing the signs $u$ and $\tilde u$ with $u \mapsto \mu u$ and $\tilde u \mapsto \tilde \mu \tilde u$, $\eta$ changes to $\eta \mapsto \eta \mu \tilde \mu$.
The sign $\eta$ plays the central role in the pump invariant. See Sec.~\ref{PumpInTriv}.

Next, we will discuss the case of $(-)$-algebra.

\begin{definition}[Injective fMPS with Wall invariant $(-)$]
\label{def:fMPS_inj_-}
A set of matrices $\{A^i\}_i \in \tilde {\cal M}^{\rm f}_n$ is graded injective with the Wall invariant $(-)$ if there exists a unitary matrix $u\in {\rm U}(2n)$ such that the following two conditions are fulfilled: 
(i) $u$ is unitary equivalent to $\sigma_x \otimes 1_n$. 
(ii) Let $W \in {\rm U}(2n)$ be a unitary matrix that diagonalizes $u$ as $u = W (\sigma_x \otimes 1_n) W^\dag$. 
Then, the matrices $A^i$ are written as 
\begin{align}
    A^i = W (\sigma_x^{|i|} \otimes B^i) W^\dag, 
\end{align}
and 
both the even subalgebra $\{B^{i_1}\cdots B^{i_l} | \sum_{k=1}^l |i_k| \equiv 0\,(\mbox{mod.2})\}$ and the odd subalgebra $\{B^{i_1}\cdots B^{i_l} | \sum_{k=1}^l |i_k| \equiv 1\,(\mbox{mod.2})\}$ with a fixed $l\in \mathbb{N}$ span the vector space ${\rm M}_n(\mathbb{C})$.
\end{definition}
The matrix $u$ is unique up to a sign.
Since the matrices $\{B^i\}$ is ungraded irreducible, one can normalize $\{B^i\}$ to be the canonical form $\sum_i B^i B^{i\dag} = 1_n$, that is, we have the canonical form (\ref{eq:+_c_form}) for $\{A^i\}_i$. 
\begin{theorem}(Fundamental theorem for fMPS with Wall invariant $(-)$)\label{fthmminus}\\
Let $\{A^{i}\}_i$ and $\{\tilde{A}^{i}\}_i$ be injective fMPSs with the Wall invariant $(-)$ in the canonical form (\ref{eq:+_c_form}).
They give the same physical state in APBC, in other words, $\{A^i\}_i \sim \{\tilde A^i\}_i$ holds if and only if there exist a unitary matrix $V\in\mathrm{U}(2n)$ and a $\mathrm{U}(1)$ phase $e^{i\beta}\in\mathrm{U}(1)$ obeying that
\begin{eqnarray}\label{fmpseqplusA}
\tilde A^{i}=e^{i\beta}V^\dag A^{i}V. 
\end{eqnarray}
The $\mathrm{U}(1)$ phase $e^{i\beta}$ is unique.
The unitary matrix $V$ is unique up to multiplications $V \mapsto e^{i\theta u} V e^{i\phi \tilde u}$, where $e^{i\theta u}$ and $e^{i\phi \tilde u}$ are any unitary matrices in the centers of the algebras generated by $\{A^i\}_i$ and $\{\tilde A^i\}_i$, respectively. 

Furthermore, if we take the Wall matrices for $\{A^i\}_i$ and $\{\tilde A^i\}_i$ to be $u$ and $\tilde u$, respectively, they are connected by
\begin{align}
    \tilde u = \eta V^\dag u V
    \label{eq:fthm_fmps_-_u}
\end{align}
with $\eta \in \{\pm 1\}$ a sign. 

This theorem holds even if the assumption $\{A^i\}_i\sim\{\tilde{A}^i\}_i$ is changed to $\{A^i\}_i\sim_{\rm PBC}\{\tilde{A}^i\}_i$.
\end{theorem}
We sketch the proof. 
Set $A^i = W (\sigma_x^{|i|} \otimes B^i) W^\dag$ and $\tilde A^i= \tilde W (\sigma_x^{|i|} \otimes \tilde B^i) \tilde W^\dag$.
Correspondingly, the Wall matrices for $\{A^i\}_i$ and $\{\tilde A^i\}_i$ are $u = W (\sigma_x \otimes 1_n) W^\dag$ and $\tilde u = \tilde W (\sigma_x \otimes 1_n) \tilde W^\dag$, respectively. 
Then $\{A^i\}_i \sim \{\tilde A^i\}_i$ ($\{A^i\}_i \sim_{\rm PBC} \{\tilde A^i\}_i$) implies that $\{\sigma_x^{|i|} \otimes B^i\}_i \sim \{\sigma_x^{|i|} \otimes \tilde B^i\}_i$ ($\{\sigma_x^{|i|} \otimes B^i\}_i \sim_{\rm PBC} \{\sigma_x^{|i|} \otimes \tilde B^i\}_i$).
In Appendix \ref{prfoffundamental}, we prove the following lemma.
\begin{lemma}
\label{lem:fMPS_-}
Let $\{\sigma^{|i|} \otimes B^i\}_i$ and $\{\sigma_x^{|i|} \otimes \tilde{B}^{i}\}_i$ be injective fMPSs with the Wall invariant $(-)$ in the canonical form (\ref{eq:+_c_form}). 
If either $\{\sigma_x^{|i|} \otimes B^i\}_i \sim \{\sigma_x^{|i|} \otimes \tilde B^i\}_i$ or $\{\sigma_x^{|i|} \otimes B^i\}_i \sim_{\rm PBC} \{\sigma_x^{|i|} \otimes \tilde B^i\}_i$ holds true, then there exist a $\mathrm{U}(1)$ phase $e^{i\beta}$, a unitary matrix $v \in {\rm U}(n)$, and a sign $\eta \in \{\pm 1\}$ such that 
\begin{align}
    \tilde B^i = e^{i\beta} \eta^{|i|} v^\dag B^i v. 
\end{align}
Moreover, the $\mathrm{U}(1)$ phase $e^{i\beta}$ and the sign $\eta$ are unique, and the unitary matrix $v$ is unique up to $\mathrm{U}(1)$ phases. 
\end{lemma}
Setting $V = W (\sigma_z^{\frac{1-\eta}{2}} \otimes v) \tilde W^\dag$ yields the desired relations (\ref{fmpseqplusA}) and (\ref{eq:fthm_fmps_-_u}). 
The matrix $V$ is not unique. 
To see this, suppose that a unitary matrix $V' \in {\rm U}(2n)$ also satisfies the relation (\ref{fmpseqplusA}) with the same $\mathrm{U}(1)$ phase $e^{i\beta}$. 
Then we have the equality $(\sigma_x \otimes B^i) W^\dag V V'^\dag W=W^\dag V V'^\dag W (\sigma_x \otimes B^i)$. 
Since both the even and odd subalgebras generated by $\{B^i\}_i$ produces the matrix algebra ${\rm M}_n(\mathbb{C})$, the matrix $W^\dag V V'^\dag W$ can be written in the form $W^\dag V V'^\dag W = e^{-i\theta \sigma_x} \otimes 1_n$ with $\theta \in [0,2\pi]$. 
Therefore, $V' = W e^{i \theta \sigma_x} W^\dag V
= e^{i\theta u} V$. 
This ambiguity of $V$ does not affect the relation (\ref{eq:fthm_fmps_-_u}). 
Using (\ref{eq:fthm_fmps_-_u}), the ambiguity can also be written as $e^{i\theta u}V =Ve^{i \eta \theta \tilde u}$.

We note that the sign $\eta$ depends on the choice of signs of $u$ and $\tilde u$. 
Replacing the signs $u$ and $\tilde u$ with $u \mapsto \mu u$ and $\tilde u \mapsto \tilde \mu \tilde u$, $\eta$ changes to $\eta \mapsto \eta \mu \tilde \mu$. The sign $\eta$ plays the central role in the pump invariant. See Sec.~\ref{defofpumpinvinnontriv}.


We would like to point out that these theorems naturally includes the fermion parity symmetry
\begin{eqnarray}
A^{i}\mapsto (-1)^{\abs{i}}A^{i}
\end{eqnarray}
when we take $V=Z$, a grading matrix and $e^{i\beta}=1$.

\medskip

In the rest of this section, we give three examples of a injective fMPSs.

\subsubsection{Example 1 : The Kitaev chain in the non-trivial phase}\label{Kitaevpump}

Let's compute the MPS that represents the Kitaev chain~\cite{BWHV17}. In the case of $\theta=0$, the expression that characterizes the ground state of the Kitaev chain is eq.(\ref{nontrivgs}), which can be rewritten in terms of complex fermions as \begin{eqnarray}
\left(a_{j}a_{j+1}+a_{j}^{\dagger}a_{j+1}-a_{j}a_{j+1}^{\dagger}-a_{j}^{\dagger}a_{j+1}^{\dagger}\right)\ket{GS}=\ket{GS}.
\end{eqnarray}
This is a local condition. Therefore, if we denote the basis of the single particle Hilbert space by $\ket{0}$ and $\ket{1}$, and expand the state with respect to site $j$ and $j+1$ as
\begin{eqnarray}
x_{1}\ket{0}\ket{0}+x_{2}\ket{0}\ket{1}+x_{3}\ket{1}\ket{0}+x_{4}\ket{1}\ket{1},
\end{eqnarray}
the conditions on the coefficients are 
\begin{eqnarray}\label{coeff}
x_{1}=-x_{4},\quad x_{2}=x_{3}.
\end{eqnarray}
Therefore, an fMPS of the Kitaev chain is given by
\begin{eqnarray}
A^{0}=\begin{pmatrix}1&\\&1\end{pmatrix},\quad A^{1}=\begin{pmatrix}&-1\\1&\end{pmatrix}.
\end{eqnarray}
In fact, we can verify that $A^{0}A^{0}=-A^{1}A^{1}$ and $A^{0}A^{1}=A^{1}A^{0}$ hold at the level of matrices, so we can see that the fMPS constructed from them satisfies the condition eq.(\ref{coeff}). The algebra generated by these matrices is $\mathcal{A}\simeq\mathbb{C}A^{0}\oplus\mathbb{C}A^{1}$ which is $\zmod{2}$-graded central simple with the Wall invariant $(-)$. In this case, the Wall matrix $u$ is given by
\begin{eqnarray}
u=\begin{pmatrix}&-i\\i&\end{pmatrix},
\end{eqnarray}
up to a phase factor.  

Since it will be used in a later analysis, we will examine how the MPS changes with respect to the phase shift of the gap function. The phase shift of the gap function $\Delta\mapsto e^{i\theta}\Delta$ can be regarded as that of the complex fermion
\begin{eqnarray}
a_{j}\mapsto e^{i\frac{\theta}{2}}a_{j}\quad a_{j}^{\dagger}\mapsto e^{-i\frac{\theta}{2}}a_{j}^{\dagger},
\end{eqnarray}
so the conditions on coefficients are modified as 
\begin{eqnarray}
x_{1}=-e^{-i\theta}x_{4},\quad x_{2}=x_{3}.
\end{eqnarray}
Therefore, for example, the MPS of the Kitaev chain for general $\theta\in\left[0,2\pi\right]$ is given by
\begin{eqnarray}
A^{0}=\begin{pmatrix}1&\\&1\end{pmatrix},\quad A^{1}=e^{\frac{i\theta}{2}}\begin{pmatrix}&-1\\1&\end{pmatrix},\quad u=\begin{pmatrix}&-1\\1&\end{pmatrix}.
\end{eqnarray}

\subsubsection{Example 2 : A Domain Wall Counting Model}\label{DWCM}
Let's compute an fMPS matrices of a domain wall counting model which is introduced in Sec.\ref{tModel}. By introducing the Majorana fermion $c_{i}$, the Hamiltonian Eq.(\ref{eq:DWCMHam}) recast into
\begin{eqnarray}
H(\theta)&=&\sum_{j}-\frac{1+\cos (\theta)}{2}ic_{2j-1}c_{2j}-\frac{1-\cos (\theta)}{2}ic_{2j-2}c_{2j+1}-i \frac{\sin (\theta)}{2} (c_{2j-2}c_{2j}-c_{2j-1}c_{2j+1}),
\label{eq:DW_count_model_H}
\end{eqnarray}
for $\theta\in\left[0,2\pi\right]$. This Hamiltonian can be obtained from the Kitaev chain in the trivial phase by unitary transformation
\begin{eqnarray}
H(\theta)=\sum_{j}U_\theta (-i c_{2j-1}c_{2j}) U_\theta^{\dagger},
\end{eqnarray}
with a unitary operator $U_\theta 
    = \prod_{j\in \Z} e^{-i\frac{\theta}{2} \frac{1+ic_{2j}c_{2j+1}}{2}}$.

An fMPS of this Hamiltonian is given by
\begin{eqnarray}\label{DDWphase}
A^{0}(\theta)=\begin{pmatrix}
1&e^{i\frac{\theta}{2}}\\
e^{i\frac{\theta}{2}}&1
\end{pmatrix},\;\;\;
A^{1}(\theta)=\begin{pmatrix}
1&e^{i\frac{\theta}{2}}\\
-e^{i\frac{\theta}{2}}&-1
\end{pmatrix}.
\end{eqnarray}
In order to obtain this matrices, we recall that the ground state in the open chain of this model was a state with a phase factor on the domain wall:
\begin{align}
     \sum_{\sigma_2,\dots,\sigma_{L-1}}e^{\frac{i\theta}{2}N_{\rm DW}}(1+\sigma_1 a^\dag_1)(1+\sigma_2 a^\dag_2)\cdots(1+\sigma_L a^\dag_L)\ket{0}, \quad 
    \sigma_j \in \{\pm 1\}, 
    \label{eq:fermion_cluster}
\end{align}
where $N_{\rm DW}=\sum_{j}\frac{1-\sigma_{j}\sigma_{j+1}}{2}$.
This is a variant of the cluster model, and fMPS matrices for the state (\ref{eq:fermion_cluster}) is given by 
\begin{eqnarray}
B^{\theta,+}=\begin{pmatrix}
1&e^{i\frac{\theta}{2}}\\
0&0\\
\end{pmatrix},\;\;\;
B^{\theta,-}=\begin{pmatrix}
0&0\\
e^{i\frac{\theta}{2}}&1
\end{pmatrix}, 
\end{eqnarray}
in the basis of $\ket{\sigma_1\cdots \sigma_L}=(1+\sigma_1 a_1^\dag)\cdots(1+\sigma_La^\dag_L)\ket{0}$. 
When $\theta\neq0$, the algebra generated by these matrices is $\mathcal{A}\simeq {\rm M}_2(\mathbb{C})$ which is $\zmod{2}$-graded central simple with the Wall invariant $(+)$.  In this case, the Wall matrix $u$ is given by
\begin{eqnarray}
u=\begin{pmatrix}
&1\\
1&
\end{pmatrix}
\end{eqnarray}
up to a sign.\footnote{Since $\zmod{2}$-grading is given by whether the fermion parity is even or odd, it gives the action of swapping $+$ and $-$ in the $\sigma=\pm$ basis. This is confirmed by the fact that the MPS in the occupation basis is given by $A^{0}(\theta)=B^{\theta,+}+B^{\theta,-}$ and $A^{1}(\theta)=B^{\theta,+}-B^{\theta,-}$, as we will see later.} 
When $\theta=0$, the algebra $\mathcal{A}\simeq \mathbb{C}$. This is  a central simple algebra, but the odd part is zero. In this case, $u$ which is not proportional to $1$ does not exist, so we will take $1$ as the boundary operator. Therefore the ground states with anti-periodic boundary condition  is 
\begin{eqnarray}
\sum_{\{\sigma_{k}\}}\tr(uB^{\theta,\sigma_{1}}\cdots B^{\theta,\sigma_{L}})\ket{\sigma_{1},...,\sigma_{L}}.
\end{eqnarray}
Rewriting this into a basis of fermion occupation basis, we get 
\begin{eqnarray}
\sum_{\{\sigma_{k}\}}\tr(uB^{\theta,\sigma_{1}}\cdots B^{\theta,\sigma_{L}})\ket{\sigma_{1},...,\sigma_{L}}&=&\sum_{\{\sigma_{k}\}}\tr(uB^{\theta,+}B^{\theta,\sigma_{2}}\cdots B^{\theta,\sigma_{L}})\ket{+,\sigma_{2},...,\sigma_{L}}\nonumber\\
&\;&+\sum_{\{\sigma_{k}\}}\tr(uB^{\theta,-}B^{\theta,\sigma_{2}}\cdots B^{\theta,\sigma_{L}})\ket{-,\sigma_{2},...,\sigma_{L}}\\
&=&\sum_{\{\sigma_{k}\}}\tr(u(B^{\theta,+}+B^{\theta,-})B^{\theta,\sigma_{2}}\cdots B^{\theta,\sigma_{L}})\ket{0,\sigma_{2},...,\sigma_{L}}\nonumber\\
&\;&+\sum_{\{\sigma_{k}\}}\tr(u(B^{\theta,+}-B^{\theta,-})B^{\theta,\sigma_{2}}\cdots B^{\theta,\sigma_{L}})\ket{1,\sigma_{2},...,\sigma_{L}}\\
&=&\sum_{\{\sigma_{k},i_{1}\}}\tr(uA^{i_{1}}(\theta)B^{\theta,\sigma_{2}}\cdots B^{\theta,\sigma_{L}})\ket{i_{1},\sigma_{2},...,\sigma_{L}},
\end{eqnarray}
where $i_{1}=0,1$ and 
\begin{eqnarray}
A^{0}(\theta)=B^{\theta,+}+B^{\theta,-}=\begin{pmatrix}
1&e^{i\frac{\theta}{2}}\\
e^{i\frac{\theta}{2}}&1
\end{pmatrix},\quad
A^{1}(\theta)=B^{\theta,+}-B^{\theta,-}=\begin{pmatrix}
1&e^{i\frac{\theta}{2}}\\
-e^{i\frac{\theta}{2}}&-1
\end{pmatrix}.
\end{eqnarray}
By applying this operation to all sites, finally we get the fMPS representation with matrices
\begin{eqnarray}
A^{\circ}(\theta)=\begin{pmatrix}
1&e^{i\frac{\theta}{2}}\\
e^{i\frac{\theta}{2}}&1
\end{pmatrix},\;\;\;
A^{\bullet}(\theta)=\begin{pmatrix}
1&e^{i\frac{\theta}{2}}\\
-e^{i\frac{\theta}{2}}&-1
\end{pmatrix},\;\;\;
u=\begin{pmatrix}
&1\\
1&
\end{pmatrix}.
\end{eqnarray}
If we diagonalize $u$, we obtain
\begin{eqnarray}
A^{0}(\theta)=\begin{pmatrix}
1+e^{i\frac{\theta}{2}}&\\
&1-e^{i\frac{\theta}{2}}
\end{pmatrix},\;\;\;
A^{1}(\theta)=\begin{pmatrix}
&1-e^{i\frac{\theta}{2}}\\
1+e^{i\frac{\theta}{2}}&
\end{pmatrix},\;\;\;
u=\begin{pmatrix}
1&\\
&-1
\end{pmatrix}.
\end{eqnarray}


\subsubsection{Example 3 : The Gu-Wen model}

Let's compute an fMPS of the Gu-Wen model \cite{GW14}\cite{WG18} as an example of graded irreducible fMPS. The Hamiltonian of this model is defined by 
 \begin{eqnarray}\label{Gu-WenHam}
 H_{\mathrm{Gu}\hyphen\mathrm{ Wen}}=-\sum_{j}(a^{\dagger}_{j}-a_{j})\tau^{x}_{j+\frac{1}{2}}(a_{j+1}+a_{j+1}^{\dagger})-\sum_{j}\tau^{z}_{j-\frac{1}{2}}(1-2a^{\dagger}_{j}a_{j})\tau^{z}_{j+\frac{1}{2}},
 \end{eqnarray}
where $\tau^{x},\tau^{y}$ and $\tau^{z}$ are the Pauli matrices. This model has a $\zmodtwo\times\zmodtwo$ symmetry generated by $U_{\rm spin}=\prod_{j\in\mathbb{Z}+\frac{1}{2}}\tau^{x}$ and $U_{\rm fermion}=\prod_{j\in \mathbb{Z}}(1-2a^{\dagger}_{j}a_{j})$. This model can be obtained by applying the Jordan-Wigner transformation to one of the $\zmodtwo$ of $\zmodtwo\times\zmodtwo$ symmetry of the cluster model described in Section \ref{bmps}.
 
 First, we investigate the ground state of the Gu-Wen model. Any terms of the Hamiltonian commutes with each others, and the eigenvalue is $\pm1$. We call the first term in Eq.(\ref{Gu-WenHam}) as the fluctuation term and the second term as the configuration term. The configuration term is minimized by placing fermions only at domain walls of spins:
  \begin{eqnarray}\label{Gu-WenConfig}
\cdots\downarrow\bullet\uparrow\circ\uparrow\circ\uparrow\bullet\downarrow\circ\downarrow\bullet\uparrow\cdots,
 \end{eqnarray}
 where $\circ$ (resp. $\bullet$) denote the state without (resp. with) fermion and $\uparrow$ (resp. $\downarrow$) denote the state whose eigenvalue of $\tau^{z}$ is $1$ (resp. $-1$). We call such a state the decorated domain wall (DDW) state \cite{CYV14}. 
 
 The fluctuation term, on the other hand, map a DDW state to another DDW state by the following processes:
\begin{eqnarray}
\cdots\uparrow\bullet\downarrow\bullet\uparrow\cdots&\leftrightarrows&\cdots\uparrow\circ\uparrow\circ\uparrow\cdots,\label{flucA}\\
\cdots\uparrow\circ\uparrow\bullet\downarrow\cdots&\leftrightarrows&\cdots\uparrow\bullet\downarrow\circ\downarrow\cdots\label{flucB}.
 \end{eqnarray}
 Because no additional weight is given to the state by the fluctuation term, the ground state is given by the summation of all DDW states with the equal weights. The ground state is unique and gapped, so it is an SRE state,
 
 The MPS of this ground state is given as follows: Corresponding to the $\uparrow$, $\downarrow$, $\circ$ and $\bullet$ configurations in Eq.(\ref{Gu-WenConfig}), we introduce $A^{\uparrow},A^{\uparrow},B^{\circ}$ and $B^{\bullet}$. Here, the $\zmodtwo$-grading is even for $A^{\uparrow},A^{\downarrow},B^{\circ}$ and odd for $B^{\bullet}$. Since the ground state is invariant under the maps in Eqs. (\ref{flucA}) and (\ref{flucB}), these matrices obey 
 \begin{eqnarray}
\uparrow\bullet\downarrow\bullet\uparrow\;=\;\uparrow\circ\uparrow\circ\uparrow\;&\longleftrightarrow&\; A^{\uparrow}B^{\bullet}A^{\downarrow}B^{\bullet}A^{\uparrow}=A^{\uparrow}B^{\circ}A^{\uparrow}B^{\circ}A^{\uparrow}\\
\uparrow\circ\uparrow\bullet\downarrow\;=\;\uparrow\bullet\downarrow\circ\downarrow\;&\longleftrightarrow&\;A^{\uparrow}B^{\circ}A^{\uparrow}B^{\bullet}A^{\downarrow}=A^{\uparrow}B^{\bullet}A^{\downarrow}B^{\circ}A^{\downarrow}
 \end{eqnarray}
 and all other products are zero. In the standard basis of the entanglement spaces, these matrices are written as
  \begin{eqnarray}
A^{\uparrow}=
\begin{pmatrix}
a^{\uparrow}&\\
&b^{\uparrow}\\
\end{pmatrix}
,\hspace{5mm}
A^{\downarrow}=
\begin{pmatrix}
a^{\downarrow}&\\
&b^{\downarrow}\\
\end{pmatrix}
,\hspace{5mm}
B^{\circ}=
\begin{pmatrix}
a^{\circ}&\\
&b^{\circ}\\
\end{pmatrix}
,\hspace{5mm}
B^{\bullet}=
\begin{pmatrix}
&a^{\bullet}\\
b^{\bullet}&\\
\end{pmatrix}
\end{eqnarray}
so the above conditions read
\begin{eqnarray}
b^{\downarrow}=a^{\uparrow}=0,\hspace{5mm}
a^{\circ}a^{\circ}a^{\uparrow}=\pm b^{\bullet}b^{\downarrow}a^{\bullet},\\
b^{\downarrow}b^{\circ}=\pm a^{\circ}a^{\uparrow},\hspace{5mm}
b^{\circ}b^{\circ}b^{\downarrow}=\pm a^{\bullet}a^{\uparrow}b^{\bullet},
\end{eqnarray}
When the matrix size is $2$,  we can easily solve the above relation as
 \begin{eqnarray}
A^{\uparrow}=
\begin{pmatrix}
1&\\
&0\\
\end{pmatrix}
,\hspace{5mm}
A^{\downarrow}=
\begin{pmatrix}
0&\\
&1\\
\end{pmatrix}
,\hspace{5mm}
B^{\circ}=
\begin{pmatrix}
1&\\
&1\\
\end{pmatrix}
,\hspace{5mm}
B^{\bullet}=
\begin{pmatrix}
&1\\
1&\\
\end{pmatrix}.
\end{eqnarray}
The algebra generated by these matrices is $\mathcal{A}\simeq {\rm M}_2(\mathbb{C})$ which is central simple with the Wall invariant $(+)$. In this case, the Wall matrix $u$ is given by
\begin{eqnarray}
u=\begin{pmatrix}
1&\\
&-1
\end{pmatrix}
\end{eqnarray}
up to a sign.

\section{Computation of the Space of SRE States using fMPS}\label{mps}\label{PumpInNonTriv}

In this section, we compute the topology of the space of injective fMPSs for a few cases. 
Our strategy is as follows. First, let $\tilde{\mathcal{M}}_{n,N}$ be the set of $N$ pairs of $2n\times2n$ matrices $\{A^{i}\}_{i=1}^N$ such that they are graded injective in the canonical form:
\begin{eqnarray}
\tilde{\mathcal{M}}^{\rm inj}_{n,N}=\{\{A^{i}\}_{i=1}^{N}\in\tilde {\cal M}^{\rm f}_n \left|\right.\{A^{i}\}_{i=1}^N \text{ is injective fMPS and in the canonical form}\}.
\end{eqnarray}
For a fermionic system with $N_F$ flavors, $N=2^{N_F}$.
Let ${\cal A}$ be the graded algebra generated by the set of matrices $\{A^i\}_{i=1}^N$. 
By Wall's structure theorem (App.\ref{GCSA} Thm.\ref{structurethm}), a $\zmod{2}$-graded central simple algebra is isomorphic to either ${\rm M}_{2n}(\mathbb{C})$ (called $(+)$-type) or ${\rm M}_n(\mathbb{C})\oplus {\rm M}_n(\mathbb{C})$ (called $(-)$-type), which physically correspond to the trivial and non-trivial fermionic SPT phases, respectively \cite{Wall64,BWHV17}. Thus $\tilde{\mathcal{M}}_{n,N}^{\rm inj}$ consists of two connected components 
\begin{eqnarray}
\tilde{\mathcal{M}}^{\rm inj}_{n,N}=\tilde{\mathcal{M}}_{n,N}^{{\rm triv.}}\sqcup\tilde{\mathcal{M}}_{n,N}^{\text{ non$\hyphen$triv.}},
\end{eqnarray} defined as
\begin{eqnarray}
\tilde{\mathcal{M}}_{n,N}^{{\rm triv.}}=\{\{A^{i}\}_{i=1}^{N}\in\tilde{\mathcal{M}}^{\rm inj}_{n,N}\left|\right.\mathcal{A}:={\rm Span}(\{A^{i}\})\simeq {\rm M}_{2n}(\mathbb{C})\}
\end{eqnarray}
and
\begin{eqnarray}
\tilde{\mathcal{M}}_{n,N}^{\text{ non$\hyphen$triv.}}=\{\{A^{i}\}_{i=1}^{N}\in\tilde{\mathcal{M}}^{\rm inj}_{n,N}\left|\right.\mathcal{A}={\rm Span}(\{A^{i}\})\simeq {\rm M}_n(\mathbb{C})\oplus {\rm M}_n(\mathbb{C})\}.
\end{eqnarray}
As we saw in Sec.\ref{sec:inj_fmps}, injective fMPS has gauge redundancy. 
Thus it is necessary to divide $\tilde{\mathcal{M}}_{n,N}^{\rm inj}$ by the gauge redundancy, and then we can obtain an approximate space $\mathcal{M}_{n,N}$ of the space of SRE states $\mathcal{M}$:
\begin{eqnarray}
\mathcal{M}_{n,N}\simeq\tilde{\mathcal{M}}^{\rm inj}_{n,N}\big/\sim.
\end{eqnarray}
And finally, by taking $n$ and $N$ large enough, one would expect to obtain the space of SRE states $\mathcal{M}$. 
Although determining such a space is in general difficult, it is possible in the non-trivial phase to perform a specific analysis under appropriate assumptions, as we will see later. Thus, in the following sections, we determine the space of injective fMPSs with the Wall invariant $(-)$ for several cases with small matrix sizes $n$ and compute the fundamental group of it.  A more general characterization of the pump in the trivial and non-trivial phases is given in Sec.\ref{inv}.

\subsection{Gauge-fixing condition}

We compute the space ${\cal M}^{\rm non\hyphen triv.}_{n,N} = \tilde {\cal M}^{\rm non\hyphen triv.}_{n,N}/\sim$ for a few cases with small matrix sizes. 
By taking the unitary matrix $V \in {\rm U}(2n)$ in Theorem~\ref{fthmminus} to be $W$ itself, the matrices $A^i$ can be in the form
\begin{align}
    &A^i = \sigma_x^{|i|} \otimes B^i.  \label{eq:gauge_fixing_condition}
\end{align}
Under this gauge-fixing condition, the Wall matrix is given by 
\begin{align}
u=\pm \sigma_x \otimes 1_n \label{eq:gauge_u}, 
\end{align}
and from Lemma~\ref{lem:fMPS_-}, the residual gauge transformation is given by 
\begin{align}
B^i \sim e^{i\beta} \eta^{|i|} v^\dag B^i v, 
\end{align}
where $e^{i\beta}$, $v$, and $\eta$ are ${\rm U}(1)$ phase, $U(n)$ matrix, and a sign, respectively.
We note that the gauge transformation $B^i \sim (-1)^{|i|} B^i$ is nothing but the fermion parity symmetry. 
The condition for the matrices $\{\sigma_x^{|i|} \otimes B^i\}_{i=1}^N$ to be in the canonical form (\ref{eq:+_c_form}) is 
\begin{align}
    \sum_{i=1}^N B^i B^{i\dag} = 1_n.
\end{align}



\subsection{For $2\times 2$ Matrices and $1$-flavor i.e. $n=1, N_F=1$}\label{subsubsec:2by21f}

Consider the above problem for $2\times 2$ matrices and $1$-flavor (i.e. $n=1, N_F=1$). 
Under the gauge-fixing condition (\ref{eq:gauge_fixing_condition}), the matrices $A^{0},A^{1}$ are given by 
\begin{eqnarray}\label{2by2fmps}
A^{0}=\frac{1}{\sqrt{|a^0|^2+|a^1|^2}}\begin{pmatrix}a^{0}&\\&a^{0}\\
\end{pmatrix},\quad 
A^{1}=\frac{1}{\sqrt{|a^0|^2+|a^1|^2}}\begin{pmatrix}&a^{1}\\a^{1}&\\
\end{pmatrix},\quad
a^0,a^1 \in \mathbb{C}.
\end{eqnarray}
The graded injectivity requires $a^0 \neq 0$ and $a^1 \neq 0$. 
Using the residual gauge transformation by $e^{i\beta} \in {\rm U}(1)$, $a^0$ can be a real positive number $a^0>0$. 
Furthermore, since $A^0$ and $A^1$ depend only on the ratios $a^0/|a^0|, a^1/|a^0|$, $a^0$ can be set as $a^0=1$. 
We considered the case where both $a^{0}$ and $a^{1}$ are non-zero because, in fact, the fMPSs with $a^{0}=0$ or $a^{1}=0$ are $\bigotimes\ket{0}$ and $\bigotimes\ket{1}$ respectively, which belong to the trivial phase. 
As a result, the parameters of the fMPS are in  $\mathbb{C}\cup\{\infty\} \cong S^2$. 
Thus, at this stage, the space $\tilde{\mathcal{M}}^{\mathrm{non}\hyphen \mathrm{trivial}}_{n=1,N=2}$ is recast as the two-dimensional sphere minus the north and south poles $S^{2}\backslash{2\mathrm{pts.}}$ (Figure \ref{fig:n1N1} [Left]). 
The remaining gauge transformation is the fermion parity symmetry $A^{i}\mapsto(-1)^{\abs{i}}A^{i}$, which leads to the identification $a^1 \sim -a^1$. 
This transformation acts on the space $S^{2}\backslash{2\mathrm{pts.}}$ by swapping the antipodal points at each circle that appears when the sphere is cut by a constant latitude plane, which results in the topologically same space $S^{2}\backslash{2\mathrm{pts.}}$. 
Now we have identified all gauge redundancy, and the space of injective fMPSs in the non-trivial phase $\mathcal{M}_{n=1,N=2}^{\mathrm{non}\hyphen \mathrm{trivial}}$ is homotopic to a sphere with two points removed $S^{2}\backslash{2\mathrm{pts}}$ (Figure~\ref{fig:n1N1} [Right]).

\begin{figure}[H]
\centering
  \includegraphics[width=80mm]{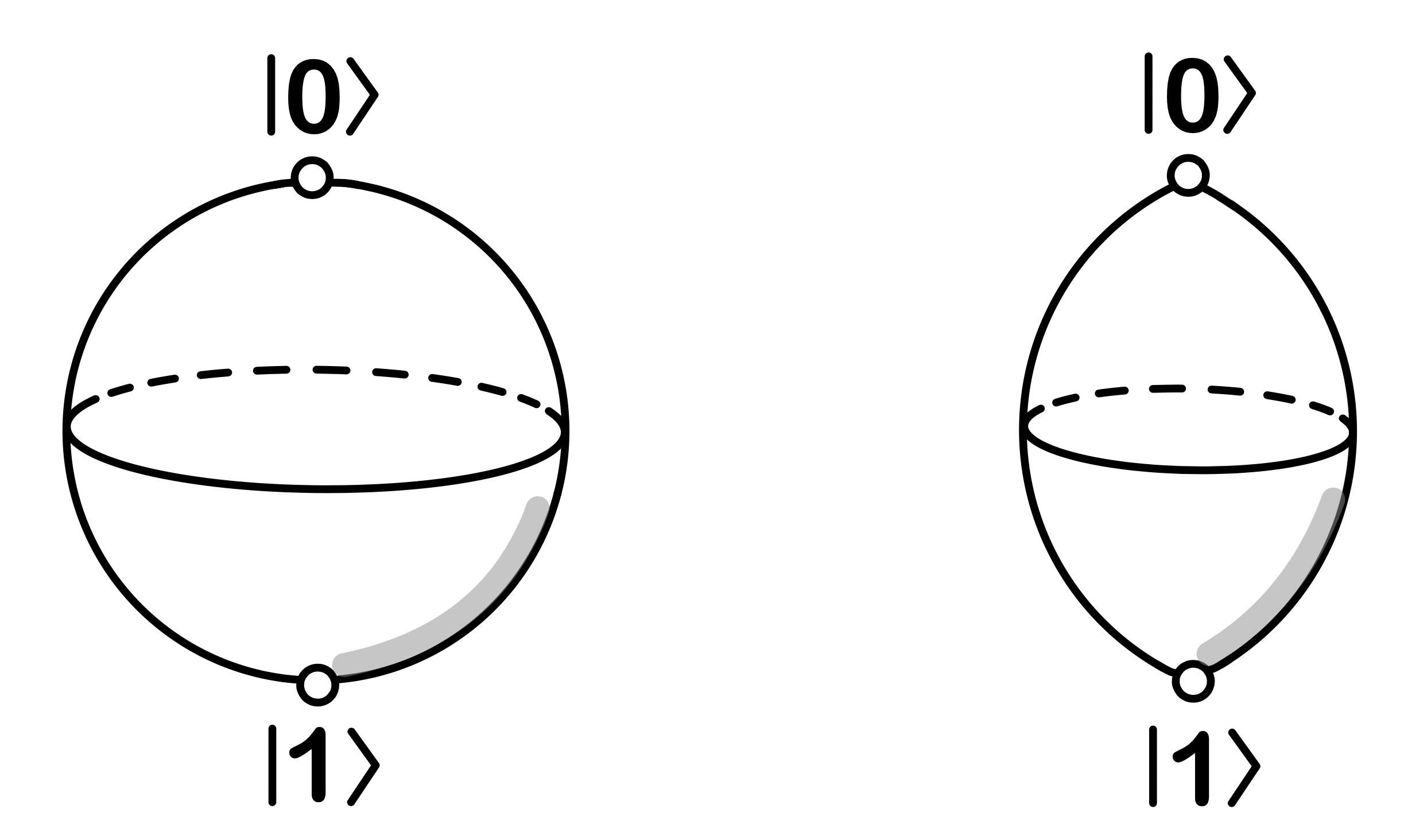}
  \caption{
  [Left] Figure of the space of injective fMPS with the Wall invariant $(-)$ when the matrix size is $2$ and the number of flavors is $1$. Since the north and south poles belong to the trivial phase, the space is homotopic to the space of $S^2$ with two points removed.
  [Right]
  Figure of the space of injective fMPSs when the matrix size is $2$ and the number of flavors is $1$.
}
 \label{fig:n1N1}
\end{figure}

Let's compute the classification of the Thouless pump in non-trivial phases. $\mathcal{M}_{n=1,N=2}^{\mathrm{non}\hyphen \mathrm{trivial}}:=S^{2}\backslash{2\mathrm{pts.}}$ corresponds to the non-trivial phase and its fundamental group is isomorphic to $\mathbb{Z}$:
\begin{eqnarray}
\pi_{1}(\mathcal{M}_{n=1,N=2}^{\mathrm{non}\hyphen \mathrm{trivial}})\simeq\mathbb{Z}
\end{eqnarray}
This result suggests the existence of a non-trivial Thouless pump classified by $\mathbb{Z}$. It can be seen that the $2\pi$ rotation of the gap function of the Kitaev chain defines a path $\abs{a^{1}}=1$ in $\mathcal{M}_{n=1,N=2}^{\mathrm{non}\hyphen \mathrm{trivial}}$ (see Section \ref{Kitaevpump}), which defines the generator of the fundamental group above. In the case of a free fermionic system, such a path of the Kitaev chain also generates a non-trivial Thouless pump \cite{TK10} and, in particular, when the flavor number is $1$, it is known that pumps are classified by $\mathbb{Z}$ generated by the loop. Therefore, this result is consistent with [\onlinecite{TK10}]. We also calculated the Berry phase of the ground state for the non-trivial path in this space, and confirmed that the ratio of the values calculated under the periodic boundary condition and the anti-periodic boundary condition converges to $-1$ in the limit of increasing the size of the system. See Appendix \ref{App:BerryPhase} for details of the calculation.

\subsection{For $2\times 2$ Matrices and generic flavors i.e. $n=1, N_F>1$}\label{subsubsec:2by2Nf}

Keep the size of the matrix $2\times2$ (i.e. $n=1$) and firstly consider the case of 2-flavors (i.e. $N_F=2$)\footnote{An example of a Hamiltonian corresponding to these fMPSs is given in Sec.\ref{homotopy}.}.
We define 
\begin{align}
    A^{00} = \begin{pmatrix}
    a\\
    &a\\
    \end{pmatrix}, \quad
    A^{11} = \begin{pmatrix}
    b\\
    &b\\
    \end{pmatrix}, \quad
    A^{10} = \begin{pmatrix}
    &c\\
    c&\\
    \end{pmatrix}, \quad
    A^{01} = \begin{pmatrix}
    &d\\
    d&\\
    \end{pmatrix}, 
\end{align}
with $a,b,c,d \in \mathbb{C}$. 
The condition for the canonical form is 
\begin{align}
|a|^2+|b|^2+|c|^2+|d|^2 = 1. 
\end{align}
The graded injectivity is met if both the conditions $(a,b) \neq (0,0)$ and $(c,d) \neq (0,0)$ are satisfied. 
Here, we suppose that $a \neq 0$. 
Then, using the residual gauge transformation by $e^{i\beta} \in {\rm U}(1)$, $a^0$ can be a real positive number $a^0>0$. 
Let us parameterize $(a,b,c,d)$ by a real parameter $t \in (0,1)$, unit $2$-sphere $(n_1,n_2,n_3) \in S^2$, and unit $3$-sphere $(n_1',n_2',n_3',n_4') \in S^3$ as in 
\begin{align}
    &(a,b) = t (n_3,n_1+in_2), \\
    &(c,d) =\sqrt{1-t^2} (n_1'+in_2',n_3'+in_4').
\end{align}
Here, $t= 0,1$ are excluded from the graded injectivity. 
Also, $a>0$ implies that $(n_1,n_2,n_3)$ runs only over the north hemisphere, namely, a disc $D^2$. 
The gauge transformation $B^i \sim (-1)^{|i|} B^i$ leads to the identification $\bm{n}' \sim -\bm{n}'$, that is, we have the real projective space $\rp^3 = S^3/(\bm{n}' \sim -\bm{n}')$. 
Therefore, $\mathcal{M}^{\rm non\hyphen triv}_{n=1,N=4} = \tilde{\cal M}^{\rm non\hyphen triv}_{n=1,N=4}/\sim$, the space of injective fMPSs with $n=1$ and $N_F=2$ divided by the gauge transformation, is found as 
\begin{eqnarray}
\mathcal{M}_{n=1,N=4}^{\mathrm{non}\hyphen \mathrm{trivial}}
\cong 
(0,1) \times D^2 \times \rp^3, 
\end{eqnarray}
and this is homotopy equivalent to $\rp^3$. 
It is easy to generalize the discussion above to generic flavor number $N_F>1$. 
The space $\mathcal{M}^{\rm non\hyphen triv}_{n=1,N=2^{N_F}>2}$ is homotopically equivalent to the real projective space $\rp^{2^{N_F}-1}$.
We get the first homotopy group~\footnote{We will show in Sec.\ref{homotopy} that a loop wrapped twice can be continuously transformed into a loop wrapped zero times, specifically in terms of Hamiltonians.} 
\begin{eqnarray}
\pi_{1}(\mathcal{M}_{n=1,N=2^{N_F}>2}^{\mathrm{non}\hyphen \mathrm{trivial}})\simeq\zmodtwo.
\end{eqnarray}
This result suggests the existence of a non-trivial Thouless pump classified by $\zmodtwo$. In the case of the free Hamiltonian, a path that turns the phase of the gap function of the Kitaev chain (see Sec.\ref{Kitaevpump}) by $2\pi$ generates a non-trivial Thouless pump  and, in particular, when the number of flavors $N_F$ is $2$ or more, it is known that pumps are classified by $\zmodtwo$ \cite{TK10}. 
It can be seen that the the $2\pi$ rotation of the gap function of the 2-flavor Kitaev chain model defines a path in $\theta \in [0,2\pi]$ with $a=1/\sqrt{2},b=0, c=e^{i\theta/2}/\sqrt{2},d=0$ in $\mathcal{M}_{n=1,N=4}^{\mathrm{non}\hyphen \mathrm{trivial}}$ (see Section \ref{exampleGTP}), which defines the generator of the fundamental group above. 
Therefore, this result is reasonable.

\subsection{For $4\times4$ Matrices and $1$-flavor i.e. $n=2, N_F=1$}\label{4by4mpstop}

Consider the above problem for $n=2,N_F=1$. It is difficult to analyze the case of $4\times4$ matrices in general. So we consider the following special case.
\begin{eqnarray}\label{eq:4by4matrix}
A^{0}=\begin{pmatrix}1_{2}&\\&1_{2}\\
\end{pmatrix},\quad 
A^{1}=\begin{pmatrix}0&\tilde{A}^{1}\\\tilde{A}^{1}&0\\
\end{pmatrix},
\end{eqnarray}
for arbitrary $2\times2$ matrix $\tilde{A}^{1}$. 
Since $A^0$ is the unit matrix, the graded injectivity is the same as that the algebra ${\cal A}$ generated by the set of matrices $\{A^0,A^1\}$ is graded central simple. 
In this case,the following theorem holds:

\begin{theorem}\label{thm1}\quad\label{eq}\\
Let $\mathcal{A}$ be the algebra generated by $A^{0}$ and $A^{1}$. When $A^{0}$ equals to the unit matrix, the following conditions are equivalent:\\
(A)\;$\mathcal{A}$ is a $\zmodtwo$-graded central simple algebra\\
(B)\;(i)$\mathrm{det}(\tilde{A}^{1})\neq 0$ and $\mathrm{tr}(\tilde{A}^{1})=0$ or (ii)$\mathrm{det}(\tilde{A}^{1})\neq 0$ and $\mathrm{tr}(\tilde{A}^{1})^{2}-4\mathrm{tr}(\tilde{A}^{1})=0$\\
where $\tilde{A}^{1}$ is a $2\times2$ matrix defined as the off-diagonal block element of $A^{1}$.
\end{theorem}
The proof of this theorem is given in Appendix \ref{prf}. Using this theorem, we can determine the structure of a part of the space of MPS as follows:
\begin{theorem}\quad\label{top}\\
Consider the same situation as in Theorem \ref{eq}, and let $\tilde{\mathcal{M}}^{n\leq 2,\mathrm{non}\hyphen \mathrm{trivial}}_{N=2}\left|_{A^{0}=1_4}\right.$ be the space of MPS in this case. Then the topology of $\tilde{\mathcal{M}}^{n\leq 2,\mathrm{non}\hyphen \mathrm{trivial}}_{N=2}\left|_{A^{0}=1_4}\right.$ is \begin{eqnarray}
\tilde{\mathcal{M}}^{n\leq 2,\mathrm{non}\hyphen \mathrm{trivial}}_{N=2}\left|_{A^{0}=1_4}\right. \sim\mathbb{R}_{>0}\times\cfrac{S^{1}\times\{S^{2}\cup\zmodtwo\}}{\zmodtwo_{\mathrm{diag}}}\subset \mathbb{R}_{>0}\times {\rm U}(2),
\end{eqnarray}
where $\zmodtwo$ is subgroup $\pm 1_{2}$ in ${\rm SU}(2)$ and $\zmodtwo_{\mathrm{diag}}$ is center of ${\rm U}(2)$. In particular, 
\begin{eqnarray}
(i)\quad\mathbb{R}_{>0}\times\cfrac{S^{1}\times\zmodtwo}{\zmodtwo_{\mathrm{diag}}}\subset\tilde{\mathcal{M}}^{n\leq 2,\mathrm{non}\hyphen \mathrm{trivial}}_{N=2}\left|_{A^{0}=1_4}\right.
\end{eqnarray}
is parameterized by 
\begin{eqnarray}
A^{0}=\frac{1}{\sqrt{1+\lambda^{2}}}\begin{pmatrix}1_{2}&\\&1_{2}\end{pmatrix},\;\;\;A^{1}=\frac{\lambda e^{i\theta}}{\sqrt{1+\lambda^{2}}}\begin{pmatrix}&&\pm1&\\&&&\pm1\\\pm1&&&\\&\pm1&&\end{pmatrix}
\end{eqnarray}
for $\lambda\in\mathbb{R}_{>0},\theta \in [0,2\pi)$, and 
\begin{eqnarray}
(ii)\quad\mathbb{R}_{>0}\times\cfrac{S^{1}\times S^{2}}{\zmodtwo_{\mathrm{diag}}}\subset\tilde{\mathcal{M}}^{n\leq 2,\mathrm{non}\hyphen \mathrm{trivial}}_{N=2}\left|_{A^{0}=1_4}\right.
\end{eqnarray}
is parameterized by 
\begin{eqnarray}
A^{0}=\frac{1}{\sqrt{1+\lambda^{2}}}\begin{pmatrix}1_{2}&\\&1_{2}\end{pmatrix},\;\;\;A^{1}=\frac{\lambda e^{i\theta}}{\sqrt{1+\lambda^{2}}}\begin{pmatrix}0&\tilde{A}^{1}\\\tilde{A}^{1}&\end{pmatrix},\;\;\;
\tilde{A}^{1}=\begin{pmatrix}i\cos{\chi}&-e^{-i\varphi}\sin{\chi}\\e^{i\varphi}\sin{\chi}&-i\cos{\chi}\end{pmatrix}
\end{eqnarray}
for $\lambda\in\mathbb{R}_{>0},\theta\in[0,\pi),\varphi\in[0,2\pi),\chi\in[0,\pi]$.
\end{theorem}
The proof of this theorem is given in Appendix \ref{prf2}.

First, consider the case $(i)$. In this case, the matrix $A^{1}$ is
\begin{eqnarray}
A^{1}=\lambda e^{i\theta}\begin{pmatrix}&&1&\\&&&1\\1&&&\\&1&&\end{pmatrix}=\lambda e^{i\theta}\begin{pmatrix}&1\\1&\end{pmatrix}\otimes 1_{2},
\end{eqnarray}
so this can be regarded as an embedding of the case $n=1$. This is natural because a $4\times4$ matrix can represent a $2\times2$ matrix  and this component is ignored in the following. Therefore, the space of MPS in the case of $n=2$ is essentially
\begin{eqnarray}
\tilde{\mathcal{M}}^{n=2,\mathrm{non}\hyphen \mathrm{trivial}}_{N=2}\left|_{A^{0}=1_4}\right.\sim\mathbb{R}_{>0}\times\cfrac{S^{1}\times S^{2}}{\zmodtwo_{\mathrm{diag}}}.
\end{eqnarray}

Next, we consider the redundancy of the space $\tilde{\mathcal{M}}^{n=2,\mathrm{non}\hyphen \mathrm{trivial}}_{N=2}\left|_{A^{0}=1_4}\right.$. As in the case of $n=1$, the action of fermion parity symmetry is 
\begin{eqnarray}
A^{I}\mapsto (-1)^{\abs{I}}A^{I},
\end{eqnarray}
and states are invariant under this transformation. Therefore, by dividing $\tilde{\mathcal{M}}^{n=2}_{N=2}\left|_{A^{0}=1_4}\right.$ by this transformation, we obtain
\begin{eqnarray}
\tilde{\mathcal{M}}^{n=2,\mathrm{non}\hyphen \mathrm{trivial}}_{N=2}\left|_{A^{0}=1_4}\right.\Big/\zmodtwo_{\mathrm{diag}}\sim\mathbb{R}_{>0}\times\cfrac{S^{1}\times S^{2}}{\zmodtwo_{\mathrm{diag}}}\Big/{\zmodtwo_{f.p.}}\sim \mathbb{R}_{>0}\times\rp^{1}\times\rp^{2}
\end{eqnarray}
Note that we used the following relation in the last equation:
\begin{eqnarray}
\begin{matrix}
&&\zmodtwo_{\mathrm{f.p.}}&&\\
&(e^{i\theta},U)&\sim&(-e^{i\theta},U)&\\
\zmodtwo_{\mathrm{diag}}&\rotatebox{90}{$\sim$}&&\rotatebox{90}{$\sim$}&\zmodtwo_{\mathrm{diag}}\\
&(-e^{i\theta},-U)&\sim&(e^{i\theta},-U)&\\
&&\zmodtwo_{\mathrm{f.p.}}&&
\end{matrix}
\end{eqnarray}
Here the $\rp^{2}$ coordinates are redundancies that can be eliminated by unitary transformations and do not change the fMPS. In fact, the fMPS is given by 
 \begin{eqnarray}
\ket{\{A^{i}(\theta),u(\theta)\}}&=&\sum_{\{i_{k}\}}\tr\left(u A^{i_1}\cdots A^{i_{L}}\right)\ket{i_1,...i_L}
=\sumodd\left(\lambda e^{i\theta}\right)^{\sum_{k}\left|i_{k}\right|}\ket{i_1,...i_L},
\end{eqnarray}
so the fMPS does not depend on $\chi$ and $\varphi$, which are coordinates of $\rp^{2}$. Therefore, the topology of the space of SRE states is
\begin{eqnarray}
\mathcal{M}^{n=2,\mathrm{non}\hyphen \mathrm{trivial}}_{N=2}\left|_{A^{0}=1_{4}}\right.\sim\mathbb{R}_{>0}\times\rp^{1}.
\end{eqnarray}
This has the same structure in the case of $n=1,N=2$, and the classification of the Thouless pump in the non-trivial phase is given by $\mathbb{Z}$.

\section{Invariants}\label{inv}

In this section, we define the topological invariants that detect the pump in trivial (in Sec.\ref{PumpInTriv}) and non-trivial (in Sec.\ref{subsubsec:topinv}) phases. Each invariant is defined heuristically based on the free Hamiltonian model introduced in Sec.\ref{nume2} and Sec.\ref{nume3}. Applications of invariants to interacting systems are given in Sec.\ref{exampleGTP}.

\subsection{Topological Invariant of Pump in the Trivial Phase}\label{PumpInTriv}

The pump invariant in the trivial phase is the same as the pump invariant for bosonic MPS with $\Z_2$ onsite symmetry constructed in Sec.~\ref{sec:bMPS_G_sym}. 

Let $\{A^i(\theta)\}_i$ for $\theta \in [0,2\pi]$ be a family of injective fMPS with the Wall invariant (+) in the canonical form (\ref{eq:+_c_form}).
Suppose that the physical state is periodic in the sense that $\{A^i(2\pi)\}_i \sim \{A^i(0)\}_i$. 
Let $u(\theta)$ for $\theta \in [0,2\pi]$ be a continuous family of grading (Wall) matrices for $\{A^i(\theta)\}_i$. 
By using Theorem \ref{fthmplus}, there exists a $\mathrm{U}(1)$ phase $e^{i\beta}$, a unitary matrix $V$, and a sign $\eta^{(+)}_{\rm top} \in \{\pm 1\}$ such that 
\begin{align}
&A^i(2\pi) = e^{i\beta} V^\dag A^i(0) V, \\
&u(2\pi) = \eta^{(+)}_{\rm top} V^\dag u(0) V.
\label{eq:inv_+_eta}
\end{align}
Since $u(\theta)$ is continuous for $\theta \in [0,2\pi]$, the gauge transformation $u(\theta)\mapsto \mu u(\theta)$ with $\mu \in \{\pm 1\}$ does not change the sign $\eta^{(+)}_{\rm top}$, meaning that the sign $\eta^{(+)}_{\rm top}$ is gauge invariant.
Thus, the sign $\eta^{(+)}_{\rm top}$ defined in  (\ref{eq:inv_+_eta}) serves as the topological invariant of pump. 

Let us compute the pump invariant $\eta^{(+)}_{\rm top}$ for the domain wall counting model (\ref{eq:DW_count_model_H}). 
As we saw in Sec.\ref{DWCM}, we have a gauge such that the fMPS is given by 
\begin{eqnarray}
A^{\circ}(\theta)=\frac{1}{2}\begin{pmatrix}
1+e^{i\frac{\theta}{2}}&\\
&1-e^{i\frac{\theta}{2}}
\end{pmatrix},\;\;\;
A^{\bullet}(\theta)=\frac{1}{2}\begin{pmatrix}
&1-e^{i\frac{\theta}{2}}\\
1+e^{i\frac{\theta}{2}}&
\end{pmatrix},\;\;\;
u(\theta)=\begin{pmatrix}
1&\\
&-1
\end{pmatrix},
\end{eqnarray}
for each $\theta$.\footnote{This fMPS is not injective at $\theta=0$. 
However, since it can be made injective by an infinitesimal perturbation, the pump invariant $\eta$ is well-defined.}
We find that $A^i(\theta + 2\pi)=\sigma_x A^i(\theta) \sigma_x$. 
Therefore, the pump invariant is computed as $\eta^{(+)}_{\rm top} = {\rm sign}[u(\theta + 2\pi) \sigma_x u(\theta) \sigma_x] = -1$, as expected.


Relaxing the condition $u(\theta)^2 = 1$ for the Wall matrix $u(\theta)$, we obtain an alternative expression of the pump invariant. 
As discussed in Sec.~\ref{sec:bMPS_G_sym}, there always exists a $2\pi$-periodic global gauge for $\{A^i(\theta)\}$, which we denote them by $\{\check A^i(\theta)\}$, such that $\check A^i(2\pi) = \check A^i(0)$ for all $i$s. 
In the global gauge, the grading matrix $u(\theta)$ is also $2\pi$-periodic up to a $\mathrm{U}(1)$ phase.
Namely, for the global gauge, there is a $2\pi$-periodic unitary matrix $\check u(\theta)$ such that 
\begin{align}
    (-1)^{|i|} \check A^i(\theta) = \check u(\theta)^\dag \check A^i(\theta) \check u(\theta),
\end{align}
and $\check u(\theta)^2 \propto {\bf 1}$ holds true.
We have an integer-valued quantity $n^{(+)}_{\rm top}$ as the $\mathrm{U}(1)$ phase winding of $\check u(\theta)^2$, 
\begin{align}
n^{(+)}_{\rm top} =\frac{1}{2\pi i} \oint d \log \tr[\check u(\theta)^2] \in \mathbb{Z}.
\label{eq:top_inv_n_+}
\end{align}
The $2\pi$ periodicity of $\check u(\theta)$, however, remains satisfied even after $\mathrm{U}(1)$ phase replacement
$\check u(\theta) \mapsto \check u(\theta) \alpha(\theta)$ with a $\mathrm{U}(1)$ valued periodic function $\alpha(\theta)$. 
Under this replacement, the winding number $n^{(+)}_{\rm top}$ changes as $n^{(+)}_{\rm top} \mapsto n^{(+)}_{\rm top}+2m$ where $m=\int d\log \alpha\in\mathbb{Z}$.
Therefore, the winding number $n^{(+)}_{\rm top}$ is defined up to $2\mathbb{Z}$ and $n^{(+)}_{\rm top}$ takes a value in $\zmodtwo$. 
It can be shown that the two pump invariants $\eta^{(+)}_{\rm top}$ and $n^{(+)}_{\rm top}$ are related to each other by $\eta^{(+)}_{\rm top}=(-1)^{n^{(+)}_{\rm top}}$. 

For the domain wall counting model  (\ref{eq:DW_count_model_H}), 
with a gauge transformation, a $2\pi$-periodic fMPS $\{\check A^i(\theta)\}_i$ and a $2\pi$-periodic Wall matrix $\check u(\theta)$ are given by 
\begin{eqnarray}
\check{A^{i}}(\theta)=e^{i\frac{\sigma^{x}}{2}\frac{\theta}{2}}A^{i}(\theta)e^{-i\frac{\sigma^{x}}{2}\frac{\theta}{2}},\quad
\check u(\theta)=e^{i\frac{\theta}{2}}e^{i\frac{\sigma^{x}}{2}\frac{\theta}{2}}Z(\theta)e^{-i\frac{\sigma^{x}}{2}\frac{\theta}{2}}. 
\end{eqnarray}
Then $\check u(\theta)^{2}=e^{i\theta}1_{2}$ and the winding number is found to be $n^{(+)}_{\rm top}=1$.

\subsection{Topological Invariant of Pump in the Non-trivial Phase}\label{subsubsec:topinv}\label{defofpumpinvinnontriv} 

The construction of the pump invariant in the non-trivial phase is parallel to that of the trivial phase. 

Let $\{A^i(\theta)\}_i$ for $\theta \in [0,2\pi]$ be a family of injective fMPS with the Wall invariant $(-)$ in the canonical form (\ref{eq:+_c_form}), and we assume the physical state is periodic $\{A^i(2\pi)\}_i \sim \{A^i(0)\}_i$. 
Let $u(\theta)$ with the condition $u(\theta)^2=1$ for $\theta \in [0,2\pi]$ be a continuous family of the Wall matrices for $\{A^i(\theta)\}_i$. 
By using Theorem \ref{fthmminus}, there exists a $\mathrm{U}(1)$ phase $e^{i\beta}$, a unitary matrix $V$, and a sign $\eta^{(-)}_{\rm top} \in \{\pm 1\}$ such that 
\begin{align}
&A^i(2\pi) = e^{i\beta} V^\dag A^i(0) V, \\
&u(2\pi) = \eta^{(-)}_{\rm top} V^\dag u(0) V. 
\end{align}
As with the trivial phase, the sign $\eta^{(-)}_{\rm top}$ is gauge invariant and serves serves as the topological invariant of pump. 

Let us compute the pump invariant $\eta^{(-)}_{\rm top}$ for a free Kitaev chain model (\ref{eq:fixed_pt_Kitaev_open}).
As we saw in Sec.\ref{Kitaevpump}, the $2\times2$ fMPS of this Hamiltonian is given by
\begin{eqnarray}
A^{0}(\theta)=\frac{1}{\sqrt{2}}\begin{pmatrix}1&\\&1\end{pmatrix},\quad A^{1}(\theta)=\frac{e^{\frac{i\theta}{2}}}{\sqrt{2}}\begin{pmatrix}&-1\\1&\end{pmatrix},\quad u(\theta)=\begin{pmatrix}&-i\\i&\end{pmatrix}
\end{eqnarray}
for each $\theta$. 
We have $A^i(2\pi) = \sigma_z A^i(0) \sigma_z$, and thus, $u(2\pi) = - \sigma_z u(0) \sigma_z$. 
The pump invariant is found as $\eta^{(-)}_{\rm top} =-1$. 

In the same way as for the trivial phase, relaxing the condition $u(\theta)^2 = 1$ for the Wall matrix $u(\theta)$ gives us an alternative expression of the pump invariant. 
Suppose that we have a $2\pi$-periodic global gauge of fMPS $\{\check A^i(\theta)\}$ satisfying $\check A^i(2\pi) = \check A^i(0)$ for all $i$s. 
In the global gauge, the Wall matrix $\check u(\theta)$ without any constraint on the $\mathrm{U}(1)$ phase can also be $2\pi$-periodic 
\begin{align}
\check u(2\pi) = \check u(0). 
\end{align}
We have an integer-valued quantity $n^{(-)}_{\rm top}$ as the $\mathrm{U}(1)$ phase winding of $\check u(\theta)^2$ as in 
\begin{align}
n^{(-)}_{\rm top} = \frac{1}{2\pi i} \oint d \log \tr[\check u(\theta)^2] \in \mathbb{Z}.
\end{align}
The replacement
$\check u(\theta) \mapsto \check u(\theta) \alpha(\theta)$ with a $\mathrm{U}(1)$ valued periodic function $\alpha(\theta)$ yields $n^{(-)}_{\rm top} \mapsto n^{(-)}_{\rm top}+2m$ with $m=\int d\log \alpha\in\mathbb{Z}$, implying that $n^{(-)}_{\rm top}$ takes a value in $\zmodtwo$. 
It is easy to see the two pump invariants $\eta^{(-)}_{\rm top}$ and $n$ are related to each other by $\eta^{(-)}_{\rm top}=(-1)^{n^{(-)}_{\rm top}}$. 

For the Kitaev chain model  (\ref{eq:fixed_pt_Kitaev_open}), a $2\pi$-periodic fMPS $\{\check A^i(\theta)\}_i$ and a $2\pi$-periodic Wall matrix $\check u(\theta)$ are given by 
\begin{eqnarray}
\check{A^{i}}(\theta)=\begin{pmatrix}1&\\&e^{i\frac{\theta}{2}}\end{pmatrix}A^{i}(\theta)\begin{pmatrix}1&\\&e^{-i\frac{\theta}{2}}\end{pmatrix},\quad
\check{u}(\theta)=e^{i\frac{\theta}{2}}\begin{pmatrix}1&\\&e^{i\frac{\theta}{2}}\end{pmatrix}u(\theta)\begin{pmatrix}1&\\&e^{-i\frac{\theta}{2}}\end{pmatrix}
\end{eqnarray}
where we put $e^{i\frac{\theta}{2}}$ on the Wall matrix $u$ so that $\check{u}(\theta)$ is $2\pi$-periodic. Then $\check u(\theta)^{2}=e^{i\theta}1_{2}$ and the winding number of the proportional constant is a nontrivial value 
\begin{eqnarray}
n^{(-)}_{\rm top}=\frac{1}{2\pi i} \int \log \tr(\tilde{u}(\theta)^{2})=1.
\end{eqnarray}

\subsection{Geometric Interpretation}\label{geometricint}

We have defined invariants heuristically in the previous sections. From a geometric point of view, this topological invariant can be regarded as a monodromy. First, let me explain this interpretation. 

The generalized Thouless pump is given by a loop $\gamma=\{\ket{A^{i}(\theta)}\}_{0\sim2\pi}$ in the set of SRE states $\mathcal{M}$. When the state goes around the loop $\gamma$, it returns to the original one, but the matrix representation of the MPS can only return to its original one up to a unitary, that is,
\begin{eqnarray}
A^{i}(\theta=2\pi) = e^{i\beta}VA^{i}(\theta=0)V^{\dagger}
\end{eqnarray}
for some $V$ and $e^{i\beta}$ of the form in Thms.\ref{fthmplus} or \ref{fthmminus}. 

The space of SRE states $\mathcal{M}$ can be constructed by dividing the space of MPS $\tilde{\mathcal{M}}$ by  redundancy. Let $\mathcal{A}(\theta)$ be the algebra generated by $\{A^{i}(\theta)\}$ and ${\rm Aut}_{\zmod{2}}(\mathcal{A})$ be the $\zmod{2}$-grading preserving automorphism group of fMPS, which is generated by a unitary matrices of Thms. \ref{fthmplus} or \ref{fthmminus}, that is, 
\begin{eqnarray}
{\rm Aut}_{\zmod{2}}(\mathcal{A})=\{V\in\mathrm{PU}(2n)\left|\right. uV=Vu\}\cup\{V\in\mathrm{PU}(2n)\left|\right. uV=-Vu\}
\end{eqnarray}
and for both $(+)$ and $(-)$ cases. 
We call an element of the first component of ${\rm Aut}_{\zmod{2}}(\mathcal{A})$ an even unitary, and an element of the second component an odd unitary for both $(+)$ and $(-)$ cases. Let $s:{\rm Aut}_{\zmod{2}}(\mathcal{A})\to\zmod{2}$ be the function that measures whether it is even unitary or odd unitary. $s$ is group homomorphism.

Then, $\tilde{\mathcal{M}}$ is the principal ${\rm Aut}_{\zmod{2}}(\mathcal{A})$-bundle over $\mathcal{M}$
\begin{eqnarray}
{\rm Aut}_{\zmod{2}}(\mathcal{A})\to\tilde{\mathcal{M}}\to\mathcal{M}
\end{eqnarray}
and $\{\tilde{A}^{i}(\theta)\}_{0\sim2\pi}$ gives the lift $\tilde{\gamma}$ of $\gamma$ in $\tilde{\mathcal{M}}$. In particular, as a general theory of fiber bundles, it is known that the fundamental group of the base space acts on the fiber, that is, there exist a group homomorphism
\begin{eqnarray}
m:\pi_{1}(\mathcal{M})\to{\rm Aut}_{\zmod{2}}(\mathcal{A})
\end{eqnarray}
 As we saw in Sec.\ref{PumpInTriv} and Sec.\ref{subsubsec:topinv}, if the topological invariant $n_{\rm top.}\equiv0$, then MPS matrices are glued with an even unitary matrix and if $n_{\rm top.}\equiv1$, then MPS matrices are glued with an odd unitary matrix. This means that $s\circ m(\gamma)$ coincide with $n_{\rm top.}$. In other words, $n_{\rm top.}$ is an invariant that measures the $\zmod{2}$ monodromy for $\gamma$.

Such quantities are mathematically described as characteristic classes of the $\zmod{2}$-graded central simple algebra bundle over a parameter space $X=S^{1}$ \cite{DK70}\footnote{We would like to thank Mayuko Yamashita for telling us about a $\zmodtwo$-graded central simple algebra  bundles and twists of $K$-theory.}. Here, a $\zmod{2}$-graded central simple algebra bundle over $X$ is a bundle over $X$ whose typical fiber is a $\zmod{2}$-graded central simple algebra. In the case of $X=S^{1}$, such bundles are classified by characteristic class $u_{1}(\mathcal{A})$ defined in \cite{Wall64}. When going around the circle $X=S^{1}$ from $\theta=0$ to $2\pi$, $u_{1}(\mathcal{A})$ is defined as $1$ if the fibers are glued together by an even unitary matrix, and $-1$ if they are glued together by an odd unitary matrix. Therefore, our topological invariant can be regarded as the characteristic class of $\zmodtwo$-graded central simple algebra  bundle $u_{1}(\mathcal{A})$ over $S^{1}$. 

Note that these are invariants for families of SRE states, but are not detectable in the higher dimensional Berry curvature, which was recently proposed in \cite{KS20-1,KS20-2}. In fact, for $1+1$ dimensional systems, the higher Berry curvature give an invariant for $3$-parameter families and is an invariant for the free part of the integer coefficient cohomology, so the torsion part cannot be detected as in this case.

\section{Examples of Thouless Pump}\label{exampleGTP}

In this section, we will calculate the invariants of pump formulated in the previous section for a concrete system. 

First, as an example of pumps in the trivial phase, we implement the Kitaev's canonical pump construction for the fermionic case and confirm the existence of the non-trivial pump (in Section \ref{Exampleofpumps}). We also show that a fermion parity pump can be obtained in the Gu-Wen model by rotating bosonic spins  and give a physical interpretation of the pump (in Section \ref{GuWenpump}). 

Next, as an example of pumps in the non-trivial phase, we consider a loop that rotates the phase of the gap function of the Kitaev chain, and construct the corresponding families of fMPSs and verify the existence of the non-trivial pump (in Section \ref{Kitaevpump}). We also compute the fMPS of the multi-flavor Kitaev chain and give the homotopy of the Hamiltonian that transforms the model with winding number $1$ into the model with winding number $-1$ (in Section \ref{homotopy}). This gives us a concrete confirmation that the classification of the pump in this model is given by $\zmod{2}$.

\subsection{Examples of the Thouless Pump in the Trivial Phase}\label{Exampleofpumps}

\subsubsection{Kitaev's Canonical Pump}\label{canonicalpump}

As an example of the calculation of the pump in the trivial phase, we derive the fMPS describing the Kitaev pump:
\begin{eqnarray}
\ket{\chi}\hspace{0.4mm}\ket{0}\hspace{0.4mm}\ket{0}\hspace{0.4mm}\ket{0}\hspace{2.5mm}&\cdots&\hspace{3mm}\ket{0}\hspace{0.4mm}\ket{0}\hspace{0.4mm}\ket{0}\hspace{0.4mm}\ket{\bar{\chi}}\hspace{0.4mm}\label{finalstate}\\
\abs{\grayrect}\quad\abs{\grayrect}\hspace{-0mm}&\cdots&\hspace{-0mm}\abs{\grayrect}\quad\abs{\grayrect}\nonumber\\
\ket{\chi}\ket{\bar{\chi}}\ket{\chi}\ket{\bar{\chi}}\hspace{2.5mm}&\cdots&\hspace{2.5mm}\ket{\chi}\ket{\bar{\chi}}\ket{\chi}\ket{\bar{\chi}}\label{mediumstate}\\
\abs{\grayrect}\hspace{3.5mm}\abs{\grayrect}\hspace{5mm}&\cdots&\hspace{4.3mm}\abs{\grayrect}\hspace{3.5mm}\abs{\grayrect}\nonumber\\
\ket{0}\hspace{0.4mm}\ket{0}\hspace{0.4mm}\ket{0}\hspace{0.4mm}\ket{0}\hspace{3.3mm}&\cdots&\hspace{2.8mm}\ket{0}\hspace{0.4mm}\ket{0}\hspace{0.4mm}\ket{0}\hspace{0.4mm}\ket{0}\hspace{0.4mm}\label{initialstate}
\end{eqnarray}
 Let's consider the BdG Hamiltonian
\begin{eqnarray}
h_{\mathrm{mat}}(\theta)=\cos(\theta)\tau_{z}-\sin(\theta)\tau_{x}\sigma_{y}
\end{eqnarray}
where $\tau$ acts on Nambu space and $\sigma$ acts on flavor space, and total Hamiltonian is defined by 
\begin{eqnarray}
H(\theta)=\sum_{j}h_{j}(\theta)
\end{eqnarray}
where 
\begin{eqnarray}
h_{j}(\theta)=\begin{pmatrix}a^{i\dagger}_{j},a^{i}_{j}\end{pmatrix}h_{\mathrm{mat}}(\theta)\begin{pmatrix}a^{i}_{j}\\a^{i\dagger}_{j}\end{pmatrix}
\end{eqnarray}d
and $i=\uparrow,\downarrow$, and we regard $\uparrow$ and $\downarrow$ as odd site and even  site, respectively. The ground state of the local Hamiltonian $h_{j}(\theta)$ is 
\begin{eqnarray}
\left(\cos(\frac{\theta}{2})+i\sin(\frac{\theta}{2})a^{\uparrow\dagger}_{j}a^{\downarrow\dagger}_{j}\right)\ket{00}_{j}
\end{eqnarray}
where $\ket{00}_{j}$ is vacuum state at site $j$, defined by $a^{\uparrow}_{j}a^{\downarrow}_{j}\ket{00}_{j}=0$. Therefore, the variation of $\theta$ from $0$ to $2\pi$ gives the variation of ground state from $\ket{00}_{j}$ to $\ket{11}_{j}:=a^{\uparrow\dagger}_{j}a^{\downarrow\dagger}_{j}\ket{00}_{j}$ and regarding $\uparrow$ and $\downarrow$ as two sites, this is the half of the Kitaev pump from ($\ref{initialstate}$) to (\ref{mediumstate}) in the case of $\ket{\chi}=\ket{1}$. For the remaining half of the process, we perform the same transformation for one shifted site.

It is easy to compute the fMPS in each process and there is a $2\times2$ representation labeled by $i,j\in\{0,1\}$ 
\begin{eqnarray}
A^{00}=\cos(\frac{\theta}{2}),\;A^{11}=i\sin(\frac{\theta}{2}),\;A^{01}=0,\;A^{10}=0,
\end{eqnarray}
for $\theta\in\left[0,\pi\right]$ and
\begin{eqnarray}
A^{00}&=&\begin{pmatrix}\cos(\frac{\theta}{2})&\\&0\end{pmatrix},\;A^{11}=\begin{pmatrix}0&\\&i\sin(\frac{\theta}{2})\end{pmatrix},\nonumber\\
A^{01}&=&\begin{pmatrix}&\sqrt{\cos(\frac{\theta}{2})\sin(\frac{\theta}{2})}\\0&\end{pmatrix},\;A^{10}=\begin{pmatrix}&0\\\sqrt{\cos(\frac{\theta}{2})\sin(\frac{\theta}{2})}&\end{pmatrix}\label{latterhalf},
\end{eqnarray}
for $\theta=\left[\pi,2\pi\right]$. For $\theta\in\left[0,\pi\right]$, this MPS is already in the canonical form. To get MPSs for $\theta\in\left(\pi,2\pi\right)$ to be in the canonical form, we only need to find a positive and invertible eigenmatrix of the CP map $\mathcal{E}(X):=\sum_{i}A^{i}XA^{i\dagger}$ \cite{P-GVWC07}. In fact, if we find a positive and invertible matrix $X$ such that $\mathcal{E}(X)=X$, then $X^{-\frac{1}{2}}A^{i}X^{-\frac{1}{2}}$ is in the canonical form. In particular, the existence of such a matrix is always guaranteed when the MPS is central simple as ungraded algebra \cite{EH-KR78,P-GVWC07}\footnote{The existence of a positive eigenmatrix $X$ of the CP map is guaranteed in Ref.\cite{EH-KR78}. In the proof of Theorem 4 in Ref.\cite{P-GVWC07}, it is shown that $\{A^{i}\}$ is reducible if $X$ is non-invertible. Therefore, if $\{A^{i}\}$ is central simple as ungraded algebra, $X$ is invertible.}. In the case of Eq.(\ref{latterhalf}), we can easily check that $X=\begin{pmatrix}
\cos(\frac{\theta}{2})&\\&\sin(\frac{\theta}{2})
\end{pmatrix}$ is a positive and invertible eigenmatrix of the CP map. There fore, the canonical form of Eq.(\ref{latterhalf}) is given by
\begin{eqnarray}\label{canonicalzerotopi}
A^{00}&=&\begin{pmatrix}\cos(\frac{\theta}{2})&\\&0\end{pmatrix},\;A^{11}=\begin{pmatrix}0&\\&i\sin(\frac{\theta}{2})\end{pmatrix},\nonumber\\
A^{01}&=&\begin{pmatrix}&\sin(\frac{\theta}{2})\\0&\end{pmatrix},\;A^{10}=\begin{pmatrix}&0\\\cos(\frac{\theta}{2})&\end{pmatrix}.
\end{eqnarray}

For any $\theta\in\left[0,2\pi\right]$, however, they are $\zmodtwo$-graded central simple with the Wall invariant $(+)$, in order to connect these matrices continuously, we embed the matrices at $\theta\in\left[0,\pi\right]$ in a $2\times2$ matrix as follows:
\begin{eqnarray}\label{canonicalpito2pi}
A^{00}=\begin{pmatrix}0&\\&\cos(\frac{\theta}{2})\end{pmatrix},\;A^{11}=\begin{pmatrix}0&\\&i\sin(\frac{\theta}{2})\end{pmatrix},\;A^{01}=\begin{pmatrix}&\sin(\frac{\theta}{2})\\0&\end{pmatrix},\;A^{10}=\begin{pmatrix}&\cos(\frac{\theta}{2})\\0&\end{pmatrix},
\end{eqnarray}
where we embed it in the $(2,2)$ component so that it is continuous at $\theta=\pi$. 
Then the Wall matrix $u$ is given by
\begin{eqnarray}
u\propto\begin{pmatrix}
1&\\
&-1
\end{pmatrix},
\label{eq:ex_kitaev_canonical_wall_matrix}
\end{eqnarray}
for $\theta\in\left[0,2\pi\right]$. Nevertheless, the size of the matrix algebra generated by the matrices still differs when $\theta\in\left[0,\pi\right]$ and $\theta\in\left[0,\pi\right]$. To avoid this difficulty, we introduce a perturbation terms
\begin{eqnarray}
\delta A^{00}(\theta)=i\epsilon\cos(\frac{\theta}{2})1_{2},\;\;\;
\delta A^{11}(\theta)=\epsilon\sin(\frac{\theta}{2})1_{2},\;\;\;
\delta A^{10}(\theta)=\delta A^{01}(\theta)=i\epsilon(\cos(\frac{\theta}{2})-\sin(\frac{\theta}{2}))\sigma_{x},\;\;\;
\end{eqnarray}
for a small number $\epsilon$ and redefine 
\begin{eqnarray}
A^{ij}(\theta)\mapsto\frac{A^{ij}(\theta)+\delta A^{ij}(\theta)}{\sqrt{1+\epsilon^{2}+\epsilon^{2}(\cos(\frac{\theta}{2})-\sin(\frac{\theta}{2}))}}
\end{eqnarray}
for $\left[0,2\pi\right]$. Here we chose these perturbation terms so that the fMPS after perturbation also represents the same state at $\theta=0$ and $\theta=2\pi$. Note that the matrices are still in the canonical form. Since the algebra generated by the matrices is isomorphic to ${\rm M}_2(\mathbb{C})$ for all $\theta\in\left[0,2\pi\right]$, this is a perturbation within the trivial phase.

Since this fMPS is not $2\pi$-periodic as matrices, we perform the unitary transformation
\begin{eqnarray}
\tilde{A}^{ij}(\theta)=
\begin{cases}
A^{ij}(\theta) &\theta\in\left[0,\pi\right]\\
e^{i(\theta-\pi)}e^{-i\frac{\sigma_{x}}{2}(\theta-\pi)}A^{ij}(\theta)e^{i\frac{\sigma_{x}}{2}(\theta-\pi)}&\theta\in\left[\pi,2\pi\right]
\end{cases}
\end{eqnarray}
to compute the topological invariant of this pump\footnote{If we choose the perturbation terms $\delta A^{10}$ and $\delta A^{01}$ to be proportional to $\sigma_{x}$, we cannot take a gauge that is exactly $2\pi$-periodic. In that case, however, the above gauge also gives a periodic MPS up to $\epsilon$.}. Then we can take a $2\pi$ periodic Wall matrix, for example
\begin{eqnarray}
\tilde{u}(\theta)=
\begin{cases}
\begin{pmatrix}1&\\&-1\end{pmatrix} &(\theta\in\left[0,\pi\right])\\
e^{i(\theta-\pi)}e^{-i\sigma_{x}(\theta-\pi)}\begin{pmatrix}1&\\&-1\end{pmatrix}&(\theta\in\left[\pi,2\pi\right])
\end{cases}
\end{eqnarray}
and the topological invariant Eq.(\ref{eq:top_inv_n_+}) is 
\begin{eqnarray}\label{canonicalpumpan}
n^{(+)}_{\rm top}\equiv\int d\log\tr(\tilde{u}^{2})\equiv1.
\end{eqnarray}
Therefore, this is a non-trivial pump in the trivial phase.

The algebraic pump invariant $\eta^{(+)}_{\rm top}$ defined in (\ref{eq:inv_+_eta}) is equivalent to the pump invariant by the winding number, but this one is easier to compute. 
For the gauge Eqs.(\ref{canonicalzerotopi}) and (\ref{canonicalpito2pi}), the transition function at $\theta=2\pi$
is $V = \sigma_x$, and the Wall matrix is given in (\ref{eq:ex_kitaev_canonical_wall_matrix}). 
Thus, the pump invariant is $\eta^{(+)}_{\rm top}=-1$, 
and this is consistent with Eq.(\ref{canonicalpumpan}).


\subsubsection{The Gu-Wen Model}\label{GuWenpump}

In Section \ref{fmps}, we derive the fMPS of the Gu-Wen model. In this section, we show that a fermion parity pump can be obtained in the Gu-Wen model by rotating bosonic spins and give a physical interpretation of the pump. 

Let's consider the Hamiltonian
\begin{eqnarray}
H(\theta)=-\sum_{j}(a^{\dagger}_{j}-a_{j})\tau^{x}_{j+\frac{1}{2}}(a_{j+1}+a_{j+1}^{\dagger})-\sum_{j}\tau^{z,\theta}_{j-\frac{1}{2}}(1-2a^{\dagger}_{j}a_{j})\tau^{z,\theta}_{j+\frac{1}{2}},
\end{eqnarray}
where $\theta\in\left[0,\pi\right]$ and $\tau^{z,\theta}_{j}$ is defined by  \begin{eqnarray}
\tau^{z,\theta}_{j}=e^{i\frac{\theta}{2}\tau_{j}^{x}}\tau^{z}_{j}e^{-i\frac{\theta}{2}\tau_{j}^{x}}=\begin{pmatrix}
\cos{(\theta)} & -i\sin{(\theta)}\\
i\sin{(\theta)} & -\cos{(\theta)}\\
\end{pmatrix}.
\end{eqnarray}
The periodicity of $\theta$ is $\pi$. 
Let $\ket{\uparrow}_{\theta}$ and $\ket{\downarrow}_{\theta}$ be the eigenvectors of $\tau^{z,\theta}$ with eigenvalues $1$ and $-1$ respectively:
\begin{eqnarray}
\tau^{z,\theta}\ket{\uparrow}_{\theta}&=&\ket{\uparrow}_{\theta},\\
\tau^{z,\theta}\ket{\downarrow}_{\theta}&=&-\ket{\downarrow}_{\theta}.
\end{eqnarray}
Since the action of $\tau^{x}$ on them is 
\begin{eqnarray}
\tau^{x,\theta}\ket{\uparrow}_{\theta}&=&\ket{\downarrow}_{\theta},\\
\tau^{x,\theta}\ket{\downarrow}_{\theta}&=&\ket{\uparrow}_{\theta},
\end{eqnarray}
the structure of the ground state is the same as in the original Gu-Wen model. Therefore, the fMPS of this model is 
\begin{eqnarray}\label{GuWenfmps}
\sum\tr(A^{i_{1}}B^{j_{1}}\cdots A^{i_{L}}B^{j_{L}})\ket{i_{1}}_{\theta}\ket{j_{1}}\cdots\ket{i_{L}}_{\theta}\ket{j_{L}},
\end{eqnarray}
where the matrices are defined by
\begin{eqnarray}
A^{\uparrow}=
\begin{pmatrix}
1&\\
&0\\
\end{pmatrix}
,\hspace{5mm}
A^{\downarrow}=
\begin{pmatrix}
0&\\
&1\\
\end{pmatrix}
,\hspace{5mm}
B^{\circ}=
\begin{pmatrix}
1&\\
&1\\
\end{pmatrix}
,\hspace{5mm}
B^{\bullet}=
\begin{pmatrix}
&1\\
1&\\
\end{pmatrix}.
\end{eqnarray}

We rewrite this fMPS in the basis of $\tau^{z}$. Since $\ket{i}_{\theta}=e^{i\frac{\theta}{2}}\ket{i}$ for $i=\uparrow,\downarrow$,
\begin{eqnarray}\label{uptheta}
\ket{\uparrow}_{\theta}&=&e^{i\frac{\theta}{2}}\ket{\uparrow}=\cos{(\frac{\theta}{2})}\ket{\uparrow}+i\sin{(\frac{\theta}{2})}\ket{\downarrow}
\end{eqnarray}
and
\begin{eqnarray}\label{downtheta}
\ket{\downarrow}_{\theta}&=&e^{i\frac{\theta}{2}}\ket{\downarrow}=i\sin{(\frac{\theta}{2})}\ket{\uparrow}+\cos{(\frac{\theta}{2})}\ket{\downarrow}.
\end{eqnarray}
By substituting eq.(\ref{uptheta}) and (\ref{downtheta}) in $\ket{i_{1}}$ of eq.(\ref{GuWenfmps}), we obtain
\begin{eqnarray}
&\;&\sum\tr(A^{i_{1}}B^{j_{1}}\cdots A^{i_{L}}B^{j_{L}})\ket{i_{1}}_{\theta}\ket{j_{1}}\cdots\ket{i_{L}}_{\theta}\ket{j_{L}}\\
&=&\sum\tr(A^{\uparrow}B^{j_{1}}\cdots A^{i_{L}}B^{j_{L}})(\cos{(\frac{\theta}{2})}\ket{\uparrow}+i\sin{(\frac{\theta}{2})}\ket{\downarrow})\ket{j_{1}}\cdots\ket{i_{L}}_{\theta}\ket{j_{L}}\\
&\;&\;\;\;+\tr(A^{\downarrow}B^{j_{1}}\cdots A^{i_{L}}B^{j_{L}})(i\sin{(\frac{\theta}{2})}\ket{\uparrow}+\cos{(\frac{\theta}{2})}\ket{\downarrow})\ket{j_{1}}\cdots\ket{i_{L}}_{\theta}\ket{j_{L}}\\
&=&\sum\tr(A^{i_{1},\theta}B^{j_{1}}A^{i_{2}}B^{j_{2}}\cdots A^{i_{L}}B^{j_{L}})\ket{i_{1}}\ket{j_{1}}\ket{i_{2}}_{\theta}\ket{j_{2}}\cdots\ket{i_{L}}_{\theta}\ket{j_{L}},
\end{eqnarray}
where $A^{\uparrow,\theta}=\cos{(\frac{\theta}{2})}A^{\uparrow}+i\sin{(\frac{\theta}{2})}A^{\downarrow}$ and $A^{\downarrow,\theta}=i\sin{(\frac{\theta}{2})}A^{\uparrow}+\cos{(\frac{\theta}{2})}A^{\downarrow}$. By applying this operation to all sites, the fMPS of $H(\theta)$ is given by
\begin{eqnarray}
A^{\uparrow,\theta}=
\begin{pmatrix}
\cos{(\frac{\theta}{2})}&\\
&i\sin{(\frac{\theta}{2})}\\
\end{pmatrix}
,\hspace{2mm}
A^{\downarrow,\theta}=
\begin{pmatrix}
i\sin{(\frac{\theta}{2})}&\\
&\cos{(\frac{\theta}{2})}\\
\end{pmatrix}
,\hspace{2mm}
B^{\circ}=
\begin{pmatrix}
1&\\
&1\\
\end{pmatrix}
,\hspace{2mm}
B^{\bullet}=
\begin{pmatrix}
&1\\
1&\\
\end{pmatrix}.
\end{eqnarray}

Let's compute the topological invariant of this fMPS. In order to have translational symmetry, we combine the sublattices into one and consider $(i,j)$ as one site. Therefore we consider the following fMPS:
\begin{eqnarray}
C^{\uparrow,\circ}(\theta)=\frac{1}{\sqrt{2}}A^{\uparrow,\theta}B^{\circ}&=&
\frac{1}{\sqrt{2}}\begin{pmatrix}
\cos{(\frac{\theta}{2})}&\\
&i\sin{(\frac{\theta}{2})}\\
\end{pmatrix}
,\hspace{2mm}
C^{\downarrow,\circ}(\theta)=\frac{1}{\sqrt{2}}A^{\downarrow,\theta}B^{\circ}=
\frac{1}{\sqrt{2}}\begin{pmatrix}
i\sin{(\frac{\theta}{2})}&\\
&\cos{(\frac{\theta}{2})}\\
\end{pmatrix}
,\hspace{2mm}\label{GuWenmodelstd1}\\
C^{\uparrow,\bullet}(\theta)=\frac{1}{\sqrt{2}}A^{\uparrow,\theta}B^{\bullet}&=&
\frac{1}{\sqrt{2}}\begin{pmatrix}
&\cos{(\frac{\theta}{2})}\\
i\sin{(\frac{\theta}{2})}&\\
\end{pmatrix}
,\hspace{2mm}
C^{\downarrow,\bullet}(\theta)=\frac{1}{\sqrt{2}}A^{\downarrow,\theta}B^{\bullet}=
\frac{1}{\sqrt{2}}\begin{pmatrix}
&i\sin{(\frac{\theta}{2})}\\
\cos{(\frac{\theta}{2})}&\\
\end{pmatrix}\label{GuWenmodelstd2}.
\end{eqnarray}
In order to make the matrices to be $\pi$-periodic, we perform a gauge transformation as follows:
\begin{eqnarray}
\tilde{C}^{i,j}(\theta)=e^{-i\frac{\theta}{2}}e^{i\frac{\theta}{2}\sigma_{x}}C^{i,j}(\theta)e^{-i\frac{\theta}{2}\sigma_{x}},
\end{eqnarray}
where $\sigma_{x}=\begin{pmatrix}
&1\\
1&\\
\end{pmatrix}$ acts on a virtual index of the fMPS. Then a $\pi$-periodic Wall matrix is given by
\begin{eqnarray}
\tilde{u}(\theta)=e^{i\theta}e^{i\frac{\theta}{2}\sigma_{x}}\begin{pmatrix}
1&\\
&-1\\
\end{pmatrix}e^{-i\frac{\theta}{2}\sigma_{x}}=e^{i\theta}e^{i\theta\sigma_{x}}\begin{pmatrix}
1&\\
&-1\\
\end{pmatrix}.
\end{eqnarray}
Since $\tilde{u}(\theta)^{2}=-e^{i2\theta}1_{2}$,
the topological invariant is $n^{(+)}_{\rm top}=1$. 
Therefore, $\{H(\theta)\}_{\theta\in \left[0,\pi\right]}$ has a non-trivial pump of the fermion parity.

The algebraic pump invariant $\eta^{(+)}_{\rm top}$ defined in (\ref{eq:inv_+_eta}) is equivalent to the pump invariant by the winding number, but this one is easier to compute. 
In the gauge Eqs.(\ref{GuWenmodelstd1}) and (\ref{GuWenmodelstd2}), the transition function at $\theta=\pi$ is given by $V=\sigma_{x}$, which anticommuets with the Wall matrix $u(\theta) = \sigma_z$. 
Therefore, the invariant is given by $\eta^{(+)}_{\rm top}=-1$.

Let's consider the physical description of this pump. When $\theta=0$, the ground state contains a following configuration
\begin{eqnarray}\label{thetazero}
\cdots\uparrow\circ\uparrow\circ\uparrow\cdots,\;\;\;\cdots\uparrow\bullet\downarrow\bullet\uparrow\cdots,
\end{eqnarray}
where $\uparrow,\downarrow$ are bosonic spin and $\circ$ and $\bullet$ is fermion unoccupied and occupied state. Performing $\pi$-rotation on the middle spin, we obtain
\begin{eqnarray}\label{thetapi}
\cdots\uparrow\circ\downarrow\circ\uparrow\cdots,\;\;\;\cdots\uparrow\bullet\uparrow\bullet\uparrow\cdots.
\end{eqnarray}
Comparing two configurations eq.(\ref{thetazero}) and  eq.(\ref{thetapi}), we can see that the fermion parity of both sides of the middle site is flipped. By applying this operation to even sites, the fermion parity of all fermions is flipped. This state corresponds to the intermediate state eq.(\ref{mediumstate}) in the Kitaev pump. Then, by applying the same operation to odd sites, the fermion parity of all fermions is flipped again and the ground state return to the original state. This corresponds to the final state eq.(\ref{finalstate}) of the Kitaev pump. Therefore, if we consider a system with edges, the fermion parity is pumped to both edges of the system.

As a concrete example, consider a system with edges. Assuming that the number of sites is $4$ and the spins at both edges are fixed to up as a boundary condition, then the ground state at $\theta=0$ is a superposition of 
\begin{eqnarray}\label{initialstate4site}
\uparrow\circ\uparrow\circ\uparrow\circ\uparrow,\;\;\;\uparrow\circ\uparrow\bullet\downarrow\bullet\uparrow,\;\;\;\uparrow\bullet\downarrow\bullet\uparrow\circ\uparrow,\;\;\;\uparrow\bullet\downarrow\circ\downarrow\bullet\uparrow,
\end{eqnarray}
and the ground state at $\theta=\pi$ is a superposition of
\begin{eqnarray}\label{finalstate4site}
\uparrow\bullet\uparrow\circ\uparrow\bullet\uparrow,\;\;\;\uparrow\bullet\uparrow\bullet\downarrow\circ\uparrow,\;\;\;\uparrow\circ\downarrow\bullet\uparrow\bullet\uparrow,\;\;\;\uparrow\circ\downarrow\circ\downarrow\circ\uparrow.
\end{eqnarray}
Comparing configuration (\ref{initialstate4site}) and (\ref{finalstate4site}), we can see that the bulk (i.e. the middle fermion and two spins) is in the same state and the fermion parity flips only at both edges. This is nothing but a pumping of fermion parity and this result does not change when the number of sites is increased.

\subsection{Examples of the Thouless Pump in a the Non-trivial Phase}\label{GTPKitaev}

\subsubsection{The Interacting Kitaev chain in the non-trivial phase}

Let's consider the Kitaev chain with interactions
\begin{eqnarray}
H=\sum_{j}-t(a_{j}^{\dagger}a_{j+1}+a^{\dagger}_{j+1}a_{j})+\abs{\Delta}\left(a_{j}a_{j+1}+a_{j}^{\dagger}a_{j+1}^{\dagger}\right)-\frac{\mu}{2}(2n_{j}-1)+U(2n_{j}-1)(2n_{j+1}-1),
\end{eqnarray}
where $n_{j}=a_{j}^{\dagger}a_{j}$ is a number operator at site $j$ and $U\geq0$ is the strength of the nearest-neighbor repulsive interaction. It is known that the model is frustration-free \cite{KST15}, at 
\begin{eqnarray}
\mu=\mu^{\ast}:=4\sqrt{U^{2}+tU+\frac{t^{2}-\abs{\Delta}^{2}}{4}},
\end{eqnarray}
and its ground state in PBC is given by
\begin{eqnarray}\label{exactgs}
\ket{\psi^{\rm odd}}=e^{\lambda a_{1}^{\dagger}}e^{\lambda a_{2}^{\dagger}}\cdots e^{\lambda a_{L}^{\dagger}}\ket{0}-e^{-\lambda a_{1}^{\dagger}}e^{-\lambda a_{2}^{\dagger}}\cdots e^{-\lambda a_{L}^{\dagger}}\ket{0},
\end{eqnarray}
where $\lambda=\sqrt{\cot(\phi^{\ast}/2)}$ and $\phi^{\ast}=\arctan(2\abs{\Delta}/\mu^{\ast})$ \cite{KKWK17,BS-B20}.

It is easy to check that the fMPS matrices of this ground state is given by
\begin{eqnarray}
A^{0}=\frac{1}{\sqrt{1+\lambda^{2}}}\begin{pmatrix}
1&\\
&1
\end{pmatrix},\quad
A^{1}=\frac{i\lambda }{\sqrt{1+\lambda^{2}}}\begin{pmatrix}
&-1\\
1&
\end{pmatrix}\quad
u=\begin{pmatrix}
&-1\\
1&
\end{pmatrix}.
\end{eqnarray}
In fact, explicitly, the fMPS is
\begin{eqnarray}
\ket{\{A^{i},u\}}&=&\sum_{\{i_{k}\}}\tr\left(u A^{i_1}\cdots A^{i_{L}}\right)\ket{i_1,...i_L}
=\sum_{\{i_{k}\},{\rm odd}}\lambda ^{\sum_{k}\left|i_{k}\right|}\ket{i_1,...i_L},
\end{eqnarray}
and this is nothing but the state Eq.(\ref{exactgs}) (up to a normalization factor). We can show that, as with the interaction-free Kitaev chain, this model gives rise to a fermion parity pump by rotating the phase of the gap function $\abs{\Delta}e^{i\theta}$. Actually, the fMPS matrices of this path is given by
\begin{eqnarray}\label{intKitaevMPS}
A^{0}=\frac{1}{\sqrt{1+\lambda^{2}}}\begin{pmatrix}
1&\\
&1
\end{pmatrix},\quad
A^{1}=\frac{i\lambda e^{i\theta}}{\sqrt{1+\lambda^{2}}}\begin{pmatrix}
&-1\\
1&
\end{pmatrix}\quad
u=\begin{pmatrix}
&-1\\
1&
\end{pmatrix}.
\end{eqnarray}
Similar to the calculation in the case of free Kitaev chain (see Sec.\ref{defofpumpinvinnontriv}), we can easily check that the pump invariants $\eta^{(+)}_{\rm top} =(-1)^{n^{(+)}_{\rm top}}$ of this path is non-trivial.

This result indicates that the fermion parity pump is stable to the interaction.

\subsubsection{The Homotopy of the Hamiltonian}\label{homotopy}

Let's show that the two laps of the path that turns the phase of the Kitaev chain is a trivial path. For this purpose, it is sufficient to construct a homotopy that connects the model with a winding number of $1$ to the one with a winding number of $-1$. 

First, consider a homotopy of fMPS. Let $t\in\left[0,\pi\right]$ be a parameter of the homotopy and at $t=0$, the fMPS is given by
\begin{eqnarray}
A^{0}(\theta)=1_{2},\;\;
A^{1,a}(\theta)=e^{i\frac{\theta}{2}}i\sigma_{2},
\end{eqnarray}
where $a$ is a label of a flavor and $\theta\in\left[0,2\pi\right]$ is a parameter of the path. To deform this path, let's introduce a second flavor $b$ and define
\begin{eqnarray}\label{fmpshomotopy}
A^{0}(\theta,t)&=&1_{2},\label{fmpshomotopya}\\
A^{1,a}(\theta,t)&=&(\cos{(\frac{\theta}{2})}+i\sin{(\frac{\theta}{2})}\cos{(t)})i\sigma_{2},\label{fmpshomotopyb}\\
A^{1,b}(\theta,t)&=&i\sin{(\frac{\theta}{2})}\sin{(t)}i\sigma_{2},\label{fmpshomotopyc}
\end{eqnarray}
which will be 
\begin{eqnarray}
A^{0}(\theta,t)&=&1_{2},\\
A^{1,a}(\theta,t)&=&(\cos{(\frac{\theta}{2})}-i\sin{(\frac{\theta}{2})})i\sigma_{2}
\end{eqnarray}
at $t=\pi$.
By multiplying a by $-1$ using fermion parity symmetry, we get
\begin{eqnarray}
A^{0}(\theta,t)&=&1_{2},\\
-A^{1,a}(\theta,t)&=&-(\cos{(\frac{\theta}{2})}-i\sin{(\frac{\theta}{2})})i\sigma_{2}
=e^{-i\frac{\theta}{2}}i\sigma_{2},
\end{eqnarray}
which gives a model of winding $-1$.
 
Next, consider a Hamiltonian which corresponds to the fMPS (\ref{fmpshomotopya})-(\ref{fmpshomotopyc}). Such a Hamiltonian is given by
\begin{eqnarray}
H(\theta,t)=\sum_{j}\tilde{a}^{\dagger,a}_{j}(\theta,t)\tilde{a}^{a}_{j}(\theta,t)+\sum_{j}a^{\dagger,b}_{j}(\theta,t)a^{b}_{j}(\theta,t),
\end{eqnarray}
where complex fermions $a^{a}_{j}(\theta,t)$ and $a^{b}_{j}(\theta,t)$ are defined by
\begin{eqnarray}\label{rotfermion}
\begin{pmatrix}
a^{a}_{j}(\theta,t)\\
a^{b}_{j}(\theta,t)
\end{pmatrix}=\begin{pmatrix}
\cos{(\frac{\theta}{2})}+i\sin{(\frac{\theta}{2})}\cos{(t)}&i\sin{(\frac{\theta}{2})}\sin{(t)}\\
-i\sin{(\frac{\theta}{2})}\sin{(t)}&-\cos{(\frac{\theta}{2})}+i\sin{(\frac{\theta}{2})}\cos{(t)}
\end{pmatrix}\begin{pmatrix}
a^{a}_{j}\\
a^{b}_{j}
\end{pmatrix}
\end{eqnarray}
and virtual complex fermion $\tilde{a}^{a}(\theta,t)$ is defined by replacing $a$ and $b$ with $\tilde{a}$ and $\tilde{b}$ in equation (\ref{rotfermion}) respectively. 

Before showing that the ground state of this Hamiltonian is given by the MPS above, we note that this Hamiltonian is a homotopy connecting the Kitaev chain with a winding number of $1$ and that with a winding number of $-1$. In fact, at $t=0$, the Hamiltonian reads to
\begin{eqnarray}
H(\theta,t=0)=-\sum_{j}\left( a^{\dagger}_{j}a_{j+1}+ a^{\dagger}_{j+1}a_{j}+e^{i\theta} a_{j}a_{j+1}+e^{-i\theta}a^{\dagger}_{j+1}a^{\dagger}_{j}\right)+\sum_{j}a^{\dagger,b}_{j}a^{b}_{j},
\end{eqnarray}
Since the second term is a Kitaev chain in the trivial phase without $\theta$ dependence, it does not affect the pump. Therefore, this is a model with winding number $1$. Similarly, at $t=\pi$, the Hamiltonian reads to
\begin{eqnarray}
H(\theta,t=0)=-\sum_{j}\left( a^{\dagger}_{j}a_{j+1}+ a^{\dagger}_{j+1}a_{j}+e^{-i\theta} a_{j}a_{j+1}+e^{i\theta}a^{\dagger}_{j+1}a^{\dagger}_{j}\right)+\sum_{j}a^{\dagger,b}_{j}a^{b}_{j},
\end{eqnarray}
and this is a model with winding number $1$.

We conclude this section by proving that the ground state of the
Hamiltonian is given by fMPS (\ref{fmpshomotopya})-(\ref{fmpshomotopyc}). Let $\ket{0}_{\theta,t}$ be a vacuum state of $a^{a}_j(\theta,t)$ and $a^{b}_{j}(\theta,t)$ i.e. $a^{a}_{j}(\theta,t)\ket{0}_{\theta,t}=0$ and $a^{b}_{j}(\theta,t)\ket{0}_{\theta,t}=0$. Then, fermionic MPS representation of the ground state $\ket{\mathrm{GS}(\theta,t)}$ is given by
\begin{eqnarray}\label{rotgroundstate}
\ket{\mathrm{GS}(\theta,t)}=\sum\tr(\Omega A^{i_1}\cdots A^{i_{L}})(a^{a,\dagger}(\theta,t)_{1})^{i_{1}}\cdots(a^{a,\dagger}(\theta,t)_{L})^{i_{L}}\ket{0}_{\theta,t},
\end{eqnarray}
where $A^{0}=1_{2}$ and $A^{1}=i\sigma_{2}$. 
Here, $\Omega$ is a boundary operator. 
Remark that since the transformation (\ref{rotfermion}) is invertible, the vacuum $\ket{0}_{\theta,t}$ for $a^{a}_{j}(\theta,t)$ and $a^{b}_{j}(\theta,t)$ is equal to the vacuum $\ket{0}$ for $a^{a}_{j}$ and $a^{a}_{j}$ i.e. $\ket{0}_{\theta,t}=\ket{0}$. 

Now, substituting $\tilde{a}(\theta,t)$ of the equation (\ref{rotfermion}) for site $1$ in equation (\ref{rotgroundstate}), we get 
\begin{eqnarray}
\ket{\mathrm{GS}(\theta,t)}&=&\sum\tr(\Omega A^{0}\cdots A^{i_{L}})(a^{a,\dagger}_{2}(\theta,t))^{i_{2}}\cdots(a^{a,\dagger}_{L}(\theta,t))^{i_{L}}\ket{0}\\
&+&\sum\tr(\Omega A^{1}\cdots A^{i_{L}})\{\cos{(\frac{\theta}{2})}+i\sin{(\frac{\theta}{2})}\cos{(t)}\}a^{a,\dagger}_{1}\cdot(a^{a,\dagger}_{2}(\theta,t))^{i_{2}}\cdots(a^{a,\dagger}_{L}(\theta,t))^{i_{L}}\ket{0}\\
&+&\sum\tr(\Omega A^{1}\cdots A^{i_{L}})i\sin{(\frac{\theta}{2})}\sin{(t)}a^{b,\dagger}_{1}\cdot(a^{a,\dagger}_{2}(\theta,t))^{i_{2}}\cdots(a^{a,\dagger}_{L}(\theta,t))^{i_{L}}\ket{0}\\
&=&\sum\tr(\Omega A^{i_{1},j_{1}}(\theta,t)A^{i_{2}}\cdots A^{i_{L}})(a^{j_{1},\dagger}_{1})^{i_{1}}(a^{a,\dagger}_{2}(\theta,t))^{i_{2}}\cdots(a^{a,\dagger}_{L}(\theta,t))^{i_{L}}\ket{0},
\end{eqnarray}
where $\{A^{i}(\theta,t)\}_{i=0,(1,a),(1,b)}$ is defined by (\ref{fmpshomotopya})-(\ref{fmpshomotopyc}). By applying this operation to all sites, we see that the ground state of $H(\theta,t)$ is given by fMPS (\ref{fmpshomotopya})-(\ref{fmpshomotopyc}).

\section{Summary and Discussion}

\subsection{Summary}
We used the fermionic MPS to study the space of $1+1$-dimensional interacting SRE states $\mathcal{M}$.  In particular, using this, we study the generalized Thouless pump in the non-trivial phase in interacting systems, which has not been investigated before. As a result, for approximations with matrix sizes up to $4\times4$, we obtained the topology of the space of SRE states $\mathcal{M}$ in several cases, and revealed the existence of non-trivial Thouless pumps classified by $\mathbb{Z}$ or $\zmodtwo$. As a special case, this result includes the pumping of the fermion parity in the Kitaev chain\cite{Kitaev01}, which was already known in the free fermion system, in a manner consistent with previous studies\cite{FK09,TK10}, and this study shows that this pump is stable in the presence of the interaction.

In addition, we used MPS to construct invariants of the pump in the trivial and non-trivial phases. This invariant also works in models with interactions, and we used this invariant to construct new models of the fermion parity pump: Kitaev's cononical pump model, the Gu-Wen model, a domain-wall counting model, the interacting Kitaev chain. We also showed that this invariant is related to the characteristic class of $\zmod{2}$-graded central simple algebra bundles that have been studied mathematically in the context of a twist of $K$-theory \cite{DK70}.

\subsection{Discussion and Future Direction} In this paper, we studied the pumping phenomenon in the trivial and non-trivial phases of the SRE state using MPS. The fundamental and important problem is to generalize our approach to a more general symmetry.

It is also an interesting problem to investigate the classification of pumps parameterized by a more general $X$. This is a difficult problem in general, but mathematically there are several known results for the classification of $\zmod{2}$-graded central simple algebra bundles\cite{DK70}: let $X$ be a parameter space. Then the graded Brauer group over $X$ is defined by 
\begin{eqnarray}\label{defofGBr}
\mathrm{GBrU}(X):=\{\text{isomorphic class of a $\zmod{2}$-graded central simple $\mathbb{C}$-algebra bundle over $X$ }\}/\sim,
\end{eqnarray} 
where $\left[\mathcal{A}\right]\sim\left[\mathcal{B}\right]$ if and only if there are vector space $E_{0},E_{1},F_{0},F_{1}$ such that $\mathcal{A}\otimes{\rm End}(E_{0}\oplus E_{1})\simeq\mathcal{B}\otimes{\rm End}(F_{0}\oplus F_{1})$. $\mathrm{GBrU}(X)$ has a group structure under the graded tensor product, and the isomorphism 
\begin{eqnarray}\label{GBrU}
\mathrm{GBr}\mathbb{C}(X)\simeq\cohoZmodtwo{0}{X}\times\cohoZmodtwo{1}{X}\times\cohoZ{3}{X}_{\mathrm{tor.}}.
\end{eqnarray}
is known as a group. Here, $\cohoZ{3}{X}_{\mathrm{tor.}}$ is the torsion subgroup of $\cohoZ{3}{X}$ and the group structure of the right-hand side is defined by
\begin{eqnarray}
(l,a,b)+(l',a',b')=(l+l',a+a',b+b'+\beta(a\cup a'))
\end{eqnarray}
using the Bockstein map $\beta:\cohoZmodtwo{2}{X}\to\cohoZ{3}{X}$ 

We do not know whether the equivalence relation in the definition Eq.(\ref{defofGBr}) is appropriate from the perspective of pump classification. However, the $E_{2}$-page of the Atiyah-Hirzebruch spectral sequence (AHSS) of the Anderson dual of spin bordism group $(\mathrm{I}\Omega_{\mathrm{Spin}})^{n+2}(X)$, a candidate for the group that gives the classification of $n+1$-dimensional SRE states (predicted by field theory \cite{FH21},\cite{HKT20}), is given by

\begin{sseqdata}[ name = grid example e2d, xscale = 2.5 , scale = 1 , title = {$E_{2}^{p,q}=\coho{p}{X}{(\mathrm{I}\Omega_{\mathrm{Spin}})^{q}}\Rightarrow(\mathrm{I}\Omega_{\mathrm{Spin}})^{p+q}(X)$}, classes = {draw = none } , y axis gap = 1.5cm , x axis extend end = 1.5cm ] 

\class["\zmodtwo"](0,3)
\class["\zmodtwo"](0,2)\class["\cohoZmodtwo{1}{X}"](1,2) 
\class["0"](0,1) \class["0"](1,1) \class["0"](2,1)  
\class["\mathbb{Z}"](0,0) \class["\cohoZ{1}{X}"](1,0)\class["\cohoZ{2}{X}"](2,0)\class["\cohoZ{3}{X}"](3,0)   

\end{sseqdata}
\vspace{3mm}
\hspace{20mm}\printpage[ name = grid example e2d, grid = chess, page=1]\\
and the right-hand side of Eq.(\ref{GBrU}) is a subset of a line corresponding to $(\mathrm{I}\Omega_{\mathrm{Spin}})^{3}(X)$ (i.e., a line with $p+q=3$). In this sense, we expect to be able to construct some non-trivial pumping models by using $\mathrm{GBr}\mathbb{C}(X)$.

Similarly, for 
\begin{eqnarray}
\mathrm{GBrO}(X):=\{\text{isomorphic class of a $\zmod{2}$-graded central simple $\mathbb{R}$-algebra bundle over $X$ }\}/\sim,
\end{eqnarray} 
it is known that there is an isomorphism of a group
\begin{eqnarray}\label{GBrO}
\mathrm{GBrO}(X)\simeq\coho{0}{X}{\zmod{8}}\times\cohoZmodtwo{1}{X}\times\cohoZmodtwo{2}{X}.
\end{eqnarray}
where the group structure of the right-hand side is defined by
\begin{eqnarray}
(l,a,b)+(l',a',b')=(l+l',a+a',b+b'+a\cup a').
\end{eqnarray} 
Physically, $\zmodtwo$-graded central simple $\mathbb{R}$-algebra describes a $1+1$-dimensional SRE state with time-reversal symmetry, so $\mathrm{GBrO}(X)$ is expected to be related to the Anderson dual of pin bordism group. In fact, the right-hand side of Eq. (\ref{GBrO}) is the same as a line corresponding to $(\mathrm{I}\Omega_{\mathrm{Pin}^{-}})^{3}(X)$ (i.e., a line with $p+q=3$) in the $E_{2}$-page of the AHSS of coefficients $(\mathrm{I}\Omega_{\mathrm{Pin}^{-}})^{3}(\mathrm{pt})$: 

\begin{sseqdata}[ name = grid example e2, xscale = 2.5 , scale = 1 , title = {$E_{2}^{p,q}=\coho{p}{X}{(\mathrm{I}\Omega_{\mathrm{Pin}^{-}})^{q}}\Rightarrow(\mathrm{I}\Omega_{\mathrm{Pin}^{-}})^{p+q}(X)$}, classes = {draw = none } , y axis gap = 1.5cm , x axis extend end = 1.5cm ] 

\class["\zmodtwo"](0,2)
\class["\zmodtwo"](0,1)\class["\cohoZmodtwo{1}{X}"](1,1) 
\class["\mathbb{Z}"](0,0) \class["\cohoZ{1}{X}"](1,0)\class["\cohoZ{2}{X}"](2,0)  

\end{sseqdata}

\vspace{3mm}
\hspace{32mm}\printpage[ name = grid example e2, grid = chess,  page=1]\\
Actually, the isomorphism 
\begin{eqnarray}
\mathrm{GBr}\mathbb{R}(X)\simeq(\mathrm{I}\Omega_{\mathrm{Pin^{-}}})^{3}(X)
\end{eqnarray}
as a group has been mathematically proven in [\onlinecite{TY21}]. Therefore, we expect that the invariant of pumps can be defined by a similar construction for the interacting fermionic SRE state with time-reversal symmetry.

\begin{acknowledgments}
S.O. acknowledge helpful discussions with R. Kobayashi, M. Yoshino and Y. Ikeda for helpful discussions in early phases of this work. S.O. would also like to thank M. Yamashita for pointing out to be a relation to twists of $K$-theory, H. Katsura for helpful comments on the interacting Kitaev chain, and Y. Terashima for inspiring conversation. Finally S.O. would also like to thank S. Terashima for encouraging me to write.
K.S. thanks H. Tasaki and K. Totsuka for helpful discussions. 
S.O. was supported by the establishment of university fellowships towards the creation of science technology innovation.
This work was supported by JST CREST Grant No. JPMJCR19T2, and Center for Gravitational Physics and Quantum Information (CGPQI) at Yukawa Institute for Theoretical Physics.
\end{acknowledgments}

\appendix
\def\thesection{\Alph{section}}

\section{Central Simple Algebras and the Brauer Group}\label{CSA}
In this section we summarize the classical results on central simple rings in mathematics, with the minimum definitions and theorems needed for this paper. 

\subsection{Central Simple Algebras over $k$}
Let $k$ be the field. First, we give the definition of a central simple algebra over $k$.
\begin{definition}\;\\
Let $\mathcal{A}$ be the algebra over $k$.\\
(1)\;$\mathcal{A}$ is simple if and only if there are not non-trivial both-side ideals.\\
(2)\;$\mathcal{A}$ is central if and only if the center $Z(\mathcal{A})$ of the algebra $\mathcal{A}$ is isomorphic to $k$.
\end{definition}
An algebra which is simple and central is called a central simple algebra. For example, the matrix algebra $M_{n}(\mathbb{R})$ is central simple as an algebra over $\mathbb{R}$, and $\mathbb{C}$ is not central simple algebra over $\mathbb{R}$ because the center of $\mathbb{C}$ is $\mathbb{C}$ itself and not isomorphic to $\mathbb{R}$. There are two different characterizations of central simple algebras.
\begin{theorem}\;(Wedderburn)\\
Let $\mathcal{A}$ be an algebra over $k$, $\mathrm{dim}(\mathcal{A})=n<\infty$ and $\mathcal{A}\neq0$. Then following conditions are equivalent:\\
(1)\;$\mathcal{A}$ is central simple.\\
(2)\;There exist the a division algebra $D$ over $k$ and $m\in\mathbb{N}$ such that 
\begin{eqnarray}
\mathcal{A}\simeq {\rm M}_{m}(D).
\end{eqnarray}
(3)\;The map 
\begin{eqnarray}
\mathcal{A}\otimes_{k}\mathcal{A}^{\mathrm{op}}\to\mathrm{End}(\mathcal{A})\simeq {\rm M}_{n^{2}}(k)
\end{eqnarray}
gives an isomorphism, where $\mathcal{A}^{\mathrm{op}}$ is the opposite algebra of $\mathcal{A}$, that  is, $\mathcal{A}^{\mathrm{op}}$ is the algebra that is the same as $\mathcal{A}$ as a set but whose product $\circ$ is defined by $a\circ b=b\cdot a$.
\end{theorem}
The tensor product of two central simple algebras is again central simple. Therefore, it is expected that the set of all central simple algebras $M_{\mathrm{CSA}}$ has a group structure with respect to the tensor product. However, this does not make it a group, so the following equivalence relation is introduced.
\begin{definition}\;\\
Let $\mathcal{A}$ and $\mathcal{B}$ are central simple algebras. They are the Brauer equivalence if and only if there exist integers $p,q\in\mathbb{N}$ such that
\begin{eqnarray}
\mathcal{A}\otimes_{k} {\rm M}_{p}(k)\simeq\mathcal{B}\otimes_{k}{\rm M}_{q}(k).
\end{eqnarray}
We denote $\mathcal{A}\sim\mathcal{B}$ when $\mathcal{A}$ and $\mathcal{B}$ are the Brauer equivalence.
\end{definition}
Then, we can define the Brauer group as follows:
\begin{definition}\;\\
Let $M_{\mathrm{CSA}}$ be the set of all central simple algebras and $\sim$ be the Brauer equivalence. The Brauer group is the abelian group defined by
\begin{eqnarray}
\mathrm{Br}(k)=M_{\mathrm{CSA}}/\sim
\end{eqnarray}
with the unit $[k]$ and inverse $[\mathcal{A}]^{-1}=[\mathcal{A}^{\mathrm{op}}]$.
\end{definition}
One way of understanding the Brauer group is to consider it as a group that classifies division algebras. In fact, there is a one-to-one correspondence between the Brauer group and the set of equivalence classes of the division algebra. 

For example, it is known that $\mathrm{Br}(\mathbb{C})\simeq 0$. Therefore the only division algebra over $\mathbb{C}$ is $\mathbb{C}$ itself, and the algebra $\mathcal{A}$ is central simple over $\mathbb{C}$ if and only if there exist $m\in\mathbb{N}$ such that $\mathcal{A}\simeq {\rm M}_m(\mathbb{C})$. \footnote{In addition, it is known that if an algebra over $\mathbb{C}$ is simple, then it is central simple.} For another example, it is known that $\mathrm{Br}(\mathbb{R})\simeq \zmodtwo\simeq \{[\mathbb{R}],[\mathbb{H}]\}$. In fact, $\mathbb{H}\otimes\mathbb{H}$ is known to be isomorphic to ${\rm M}_{4}(\mathbb{R})$, indicating that $[\mathbb{H}]^{2}=[\mathbb{R}]$. 

\subsection{$\zmodtwo$-graded Central Simple Algebras over $k$}\label{GCSA}
The theory of central simple algebras and Brauer groups is known to be extended to the $\zmodtwo$-graded algebra \cite{Wall64}. We will first define a $\zmodtwo$-graded algebra.
\begin{definition}\;\\
An algebra $\mathcal{A}$ is $\zmodtwo$-graded if and only if there exists a direct product decomposition 
\begin{eqnarray}
\mathcal{A}\simeq\mathcal{A}^{(0)}\oplus\mathcal{A}^{(1)}
\end{eqnarray}
 such that $\mathcal{A}^{(i)}\cdot\mathcal{A}^{(j)}\subset\mathcal{A}^{(i+j)}$ for $i,j\in \{0,1\}$.
\end{definition}
Simplicity and centrality for $\zmodtwo$-graded algebras are extended as follows.
\begin{definition}\;\\
Let $\mathcal{A}$ be a $\zmodtwo$-graded algebra over $k$.\\
(1)\;$\mathcal{A}$ is simple if and only if there are no  $\zmodtwo$-graded non-trivial both-side ideals.\\
(2)\;$\mathcal{A}$ is central if and only if the even part of the center $Z(\mathcal{A})\cap\mathcal{A}^{(0)}$  is isomorphic to $k$.
\end{definition}
An $\zmod{2}$-graded algebra which is simple and central is called a $\zmod{2}$-graded central simple algebra.  For example, let $\mathcal{A}$ be the $\zmodtwo$-graded algebra over $\mathbb{C}$ generated by matrices
\begin{eqnarray}
A^{0}=\begin{pmatrix}
1&0\\
0&1\\
\end{pmatrix},\quad 
A^{1}=\begin{pmatrix}
0&-1\\
1&0\\
\end{pmatrix}\quad
\end{eqnarray}
with $\zmodtwo$-grading $0$ and $1$ respectively. This algebra is neither simple nor central as an ungraded algebra. In fact, this algebra has non-trivial ideals\footnote{Here $(a)$ denotes the both-sides ideal generated by $a\in\mathcal{A}$.} 
\begin{eqnarray}\label{ideals}
I^{\pm}=\left(A^{0}\pm iA^{1}\right),
\end{eqnarray}
and since $\mathcal{A}$ is a commutative algebra, the center $Z(\mathcal{A})$ is isomorphic to $\mathcal{A}$ itself and not to $\mathbb{C}$. However, this algebra is simple and central as $\zmodtwo$-graded algebra. In fact, the only non-trivial ideals of $\mathcal{A}$ are the ideals of Eq.(\ref{ideals}), these are not $\zmodtwo$-graded, and the even part of the center $Z(\mathcal{A})\cap\mathcal{A}^{(0)}$ is isomorphic to $\mathbb{C}$. 

As we observed above, an algebra that is not central or simple as an ungraded algebra may be central or simple as a $\zmodtwo$-graded algebra. 
Given a $\zmodtwo$-graded algebra, the pattern of breaking centrality and simplicity as an ungraded algebra is classified by the following structure theorem \cite{Wall64}.
\begin{theorem}\label{structurethm}\;\\
Let $\mathcal{A}$ be the $\zmodtwo$-graded central simple algebra. Then either one of the following is satisfied.\\
(1)\;$\mathcal{A}$ is central simple as an ungraded algebra, there exists $u\in \mathcal{A}^{(0)}$ such that $u^{2}=a\cdot 1$ for some $a\in k$, and the center $Z(\mathcal{A}^{(0)})$ is isomorphic to ${\rm Span}(1,u)$.\\
(2)\;$\mathcal{A}^{(0)}$ is central simple as an ungraded algebra, there exists $u\in \mathcal{A}^{(1)}$ such that $u^{2}=a\cdot 1$ for some $a\in k$, and the center $Z(\mathcal{A})$ is isomorphic to ${\rm Span}(1,u)$.
\end{theorem}

In Wall's paper, the case (1) is called $(+)$, and the case (2) is called $(-)$. It is possible to define the analogue of the Brauer group for $\zmodtwo$-graded algebra , it is called graded Brauer group or Brauer-Wall group denoted by $\mathrm{BW}(k)$. The Brauer-Wall group can be regarded as a group that classifies three sets of data: $(+)$ or $(-)$, $a\in k$, and division algebra $D$ such that if $\mathcal{A}$ is a $(+)$-algebra, $\mathcal{A}\simeq{\rm Mat}_{n}(D)$ and if $\mathcal{A}$ is a $(-)$-algebra, $\mathcal{A}^{(0)}\simeq{\rm Mat}_{n}(D)$. For example, $\mathrm{BW}(\mathbb{C})\simeq\zmodtwo$ and $\mathrm{BW}(\mathbb{R})\simeq\zmod{8}$, which is known as the complex and real Bott periodicity.

\section{A Proof of the Theorem \ref{fthmplus} for PBC}
\label{sec:pf_thm_2}

For a Wall matrix $u$ of $\{A^i\}_i$, the set of matrices $\{u A^i\}_i$ is also injective fMPS with the Wall invariant (+) in the canonical form~\footnote{
The products generated by the matrices $\{u A^i\}_i$ are written as $u A^{i_1} \cdots u A^{i_l}
\propto (u)^{\frac{1-(-1)^l}{2}} A^{i_1}\cdots A^{i_l}$.
Thus, if the set $\{A^{i_1} \cdots A^{i_l}\}$ spans ${\rm M}_{2n}(\mathbb{C})$, so is $\{(u A^{i_1}) \cdots (u A^{i_l})\}$.}.
Let us consider the wave functions of the bosonic MPS $\ket{\{u A^i\}_i}_L$ for odd length $L$, 
\begin{align}
    \tr[(uA^{i_1}) \cdots (uA^{i_L})]
    =
    (-1)^{\sum_{k=1}^{(L-1)/2} |i_{2k}|} 
    \tr[u A^{i_1} \cdots A^{i_L}].
\end{align}
Since the overall sign $(-1)^{\sum_{k=1}^{(L-1)/2} |i_{2k}|}$ does not depend on $\{A^i\}_i$ or $u$, the gauge equivalence $\{A^i\}_i \sim_{\rm PBC} \{\tilde A^i\}_i$ implies that the equivalence of two injective bosonic MPSs $\ket{\{u A^i\}_i} = e^{i\alpha_L} \ket{\{\tilde u \tilde A^i\}_i}$ for any odd $L$. 
From Theorem~\ref{fthm_original}, there exists a unique $\mathrm{U}(1)$ phase $e^{i\beta}$ and ${\rm U}(2n)$ matrix $V$ such that
\begin{align}
\tilde u \tilde A^i = e^{i\beta} V^\dag u A^i V
\end{align}
holds. 
Here, $V$ is unique up to $\mathrm{U}(1)$ phase. 
It can also be written as 
\begin{align}
    \tilde u \tilde A^i
    =e^{i\beta} (uV\tilde u)^\dag u A^i (uV\tilde u). 
\end{align}
From the uniqueness of $V$, $uV\tilde u= e^{i\phi} V$, and $u^2=\tilde u^2=1$ gives us $\tilde u = \pm V^\dag u V$.

\section{A Proof of Lemma \ref{lem:fMPS_-}}
\label{prfoffundamental}
 Before going into the proof of Lemma \ref{lem:fMPS_-}, we prove two lemmas.
 \begin{lemma}\label{lem:injandsimp}\;\\
 Let $\{B^{i}\}_{i=1,...,N}$ be a set of $n\times n$ matrices. 
 For $L\in\mathbb{N}$, we introduce 
 \begin{eqnarray}
 \Gamma^{e}_{L}(X):=\sum_{\{i_{k}\left.\right|\sum_{k}\abs{i_{k}}\equiv0\}}\tr(XB^{i_{1}}\cdots B^{i_{L}})\ket{i_{1},...,i_{L}}
 \end{eqnarray} 
 and 
 \begin{eqnarray}
 \Gamma^{o}_{L}(X):=\sum_{\{i_{k}\left.\right|\sum_{k}\abs{i_{k}}\equiv1\}}\tr(XB^{i_{1}}\cdots B^{i_{L}})\ket{i_{1},...,i_{L}}.
 \end{eqnarray}
 Then $\{B^{i_{1}}\cdots B^{i_{L}}\left.\right|\sum_{k}\abs{i_{k}}\equiv0\}$ ($\{B^{i_{1}}\cdots B^{i_{L}}\left.\right|\sum_{k}\abs{i_{k}}\equiv1\}$, resp.) span the matrix algebra $\mathrm{M}_{n}(\mathbb{C})$ as vector space if and only if ($\Gamma^{e}_{L}$) ($\Gamma^{o}_{L}$, resp.) is injective.
 \end{lemma}
 \begin{proof}
 ($\Rightarrow$):Let $X$ be a $n\times n$ matrix such that $\Gamma^{e}_{L}(X)=0$. Then 
 \begin{eqnarray}\label{tracetrace}
 \tr(XB^{i_{1}}\cdots B^{i_{L}})=0
 \end{eqnarray} for any $(i_{1},...,i_{L})$ with $\sum_{k}\abs{i_{k}}\equiv0$. Since $\{B^{i_{1}}\cdots B^{i_{L}}\left.\right|\sum_{k}\abs{i_{k}}\equiv0\}$ span the matrix algebra $\mathrm{M}_{n}(\mathbb{C})$, there are $c_{i_{1},...,i_{L}}^{k,l}\in\mathbb{C}^{\times}$ such that $e_{k,l}=\sum_{\{i_{k}\left.\right|\sum_{k}\abs{i_{k}}\equiv0\}}c_{i_{1},...,i_{L}}^{k,l}B^{i_{1}}\cdots B^{i_{L}}$ for any $k,l\in\{1,...,n\}$. 
 Here $e_{k,l}$ is a matrix in which only the $(k,l)$ component is $1$ and the others are $0$. By taking a linear combination of Eq.(\ref{tracetrace}) with weight $c_{i_{1},...,i_{L}}^{k,l}$, we obtain $0=\tr(Xe_{k,l})=X_{l,k}$ and thus $X=0$. Therefore $\Gamma^{e}_{L}$ is injective. In the same way, we can also show the injectivity of $\Gamma^{o}_{L}$.\\
 ($\Leftarrow$) : Note that $\Gamma^{e}_{L}:\mathbb{C}^{n^{2}}\to\mathbb{C}^{NL}$. 
 Taking $\{e_{k,l}\}$ and $\{\ket{i_{1},...,i_{L}}|\sum_{k=1}^L|i_k|\equiv 0\}$ as the basis of $\mathbb{C}^{n^{2}}$ and $\mathbb{C}^{NL}$, the matrix representation of $\Gamma_{L}^{e}$ is 
 \begin{eqnarray}
 \left(\Gamma_{L}^{e}\right)_{(i_{1},...,i_{L}),(k,l)}=(B^{i_1}\cdots B^{i_L})_{l,k}.
 \end{eqnarray}
 The $(i_1,\dots,i_L)$th row vector is identified with the matrix $B^{i_1} \cdots B^{i_L}$. 
 Since the matrix rank of $\Gamma_{L}^{e}$ is $n^{2}$ from the injectivity of $\Gamma_{L}^{e}$, $\{B^{i_{1}}\cdots B^{i_{L}}\left.\right|\sum_{k}\abs{i_{k}}\equiv0\}$ is a basis of the $n\times n$ matrix algebra $\mathrm{M}_{n}(\mathbb{C})$. 
 The odd case can be proved in the same way.
 \end{proof}

 \begin{lemma}\label{lem:LtoL+1}\;\\
 For $L\in\mathbb{N}$, suppose that both  $\{B^{i_{1}}\cdots B^{i_{L}}\left.\right|\sum_{k}\abs{i_{k}}\equiv0\}$ and $\{B^{i_{1}}\cdots B^{i_{L}}\left.\right|\sum_{k}\abs{i_{k}}\equiv1\}$ span the matrix algebra $\mathrm{M}_{n}(\mathbb{C})$ as vector space. Then the same is true for $L+1$.
 \end{lemma}
 \begin{proof}
 By Lemma~\ref{lem:injandsimp}, it is sufficient to show that $\Gamma_{L+1}^{e}$ and $\Gamma_{L+1}^{o}$ are injective. 
 Let $X$ be a $n\times n$ matrix such that $\Gamma_{L+1}^{e}(X)=0$. Then
 \begin{eqnarray}
 \Gamma_{L+1}^{e}(X)=0&\Leftrightarrow&\sum_{i_{L+1}, \abs{i_{L+1}}=0}\Gamma_{L}^{e}(B^{i_{L+1}}X)\ket{i_{L+1}}+\sum_{i_{L+1}, \abs{i_{L+1}}=1}\Gamma_{L}^{o}(B^{i_{L+1}}X)\ket{i_{L+1}}=0\nonumber\\
 &\Leftrightarrow&\begin{cases}
     \Gamma_{L}^{e}(B^{i_{L+1}}X)=0,&(\abs{i_{L+1}}=0)\nonumber\\
     \Gamma_{L}^{o}(B^{i_{L+1}}X)=0,&(\abs{i_{L+1}}=1)
 \end{cases}\\
 &\Leftrightarrow&B^{i_{L+1}}X=0\;\;\;(\text{any } i_{L+1})\label{annihi}.
 \end{eqnarray}
 Since $\{B^{i_{1}}\cdots B^{i_{L}}\left.\right|\sum_{k}\abs{i_{k}}\equiv0\}$ span the matrix algebra ${\rm M}_n(\mathbb{C})$, there are $c_{i_{1},...,i_{L}}^{e}\in\mathbb{C}^{\times}$ such that $1=\sum_{\{i_{k}\left.\right|\sum_{k}\abs{i_{k}}\equiv0\}}c_{i_{1},...,i_{L}}^{e}B^{i_{1}}\cdots B^{i_{L}}$. By taking a linear combination of Eq.(\ref{annihi}) with weight $c_{i_{1},...,i_{L}}^{k,l}$, we obtain $X=\sum_{\{i_{k}\left.\right|\sum_{k}\abs{i_{k}}\equiv0\}}c_{i_{1},...,i_{L}}^{e}B^{i_{1}}\cdots B^{i_{L}}X=0$ and thus $X=0$. Therefore $\Gamma_{L+1}^{e}$ is injective. In the same way, we can also show the injectivity of $\Gamma^{o}_{L+1}$.
 \end{proof}

We also provide a type of the fundamental theorem for bosonic injective MPS matrices not in the canonical form: 
\begin{lemma}\label{fthmnoncano}\ \\
Let $\{A^{i}\}$ and $\{\tilde{A}^{i}\}$ be injective $n\times n$ MPSs. 
The followings are equivalent. 
\begin{itemize}
\item[(i)]
They give the same states for any even system sizes, i.e. for any $L \in 2 \mathbb{N}$ there exists a $\mathrm{U}(1)$ phase $e^{i\alpha_L}$ such that $\ket{\{A^{i}\}}_{L}=e^{i\alpha_L}\ket{\{\tilde{A}^{i}\}}_{L}$ holds. 
\item[(ii)]
There exist an invertible matrix $M$ and a $\mathrm{U}(1)$ phase $z\in\mathrm{U}(1)$ satisfying $A^{i}=zM^{-1}\tilde{A}^{i}M$. 
\end{itemize}
Here, $z$ is unique and $M$ is unique up to $\mathbb{C}^{\times}$. 
This statement holds when $L$ rums over odd integers $L \in 2\mathbb{N}+1$.
\end{lemma}
\begin{proof}
To make injective MPS matrices into the canonical form, we can use the following procedure:
Let $\mathcal{E}(X)=\sum_{i}A^{i}XA^{i\dagger}$ and $\tilde{\mathcal{E}}(X)=\sum_{i}\tilde{A}^{i}X\tilde{A}^{i\dagger}$ be the transfer matrices and $\rho,\tilde{\rho}\in\mathbb{R}_{>0}$ be its spectral radius. It is known that if $\Lambda$ is an eigenmatrix of $\mathcal{E}$ with eigenvalue $\rho$, $\Lambda$ is unique up to $\mathbb{C}^{\times}$ and positive definite \cite{P-GVWC07}. It is also known that if $\Lambda'$ is an eigenmatrix of $\mathcal{E}$ with eigenvalue $\lambda$ and $\abs{\lambda}=\rho$, then $\Lambda'=\Lambda$ and $\lambda=\rho$\cite{SP-GWC10}. Therefore if we define $A'^{i}=\frac{1}{\rho^{\frac{1}{2}}}\Lambda^{-\frac{1}{2}}A^{i}\Lambda^{\frac{1}{2}}$, then $\sum_{i}A'^{i}A'^{i\dagger}=\frac{1}{\rho}\sum_{i}X^{-\frac{1}{2}}A^{i}\Lambda^{\frac{1}{2}}(\Lambda^{-\frac{1}{2}}A^{i}\Lambda^{\frac{1}{2}})^{\dagger}=\Lambda^{-\frac{1}{2}}XX^{\frac{1}{2}}=1$, so the canonical form. 
Let $\Lambda$ and $\tilde{\Lambda}$ be the positive definite eigenmatrix of $\mathcal{E}$ and $\tilde{\mathcal{E}}$.

Assume that $\rho=\tilde{\rho}$. Then, since $\ket{\{A^{i}\}}_{L}=e^{i\alpha_L}\ket{\{\tilde{A}^{i}\}}_{L}$,
\begin{eqnarray}
\tr(A^{i_{1}}\cdots A^{i_{L}})=e^{i\alpha_L}\tr(\tilde{A}^{i_{1}}\cdots \tilde{A}^{i_{L}})\Leftrightarrow\tr(A'^{i_{1}}\cdots A'^{i_{L}})=e^{i\alpha_L}\tr(\tilde{A}'^{i_{1}}\cdots \tilde{A}'^{i_{L}}).
\end{eqnarray}
By using the fundamental theorem for bosonic MPS in the canonical form (Theorem~\ref{fthm_original}), there is a unitary matrix $U$ and a unique $\mathrm{U}(1)$ phase $z$ such that
\begin{eqnarray}
A'^{i}=zU\tilde{A}'^{i}U^{\dagger}\Leftrightarrow A^{i}=z(\Lambda^{\frac{1}{2}}U\tilde{\Lambda}^{-\frac{1}{2}})\tilde{A}^{i}(\Lambda^{\frac{1}{2}}U\tilde{\Lambda}^{-\frac{1}{2}})^{-1},
\end{eqnarray}
and $\Lambda^{\frac{1}{2}}U\tilde{\Lambda}^{-\frac{1}{2}}$ is unique up to $\mathbb{C}^{\times}$.

Finally we show that $\rho=\tilde{\rho}$. The norm of an MPS is given by
\begin{eqnarray}
\braket{\{A^{i}\}|\{A^{i}\}}&=&\sum_{\{i_{k}\}\}}\tr(A^{i_{1}}\cdots A^{i_{L}})\tr(A^{i_{1}}\cdots A^{i_{L}})^{\ast}\nonumber\\
&=&\sum_{\{i_{k}\}}\sum_{k,l}\bra{k}A^{i_{1}}\cdots A^{i_{L}}\ket{k}\bra{l}A^{i_{1}\dagger}\cdots A^{i_{L}\dagger}\ket{l}\nonumber\\
&=&\sum_{k,l}\bra{k}\mathcal{E}^{L}(\bra{k}\ket{l})\ket{l}\nonumber\\
&=&\tr(\mathcal{E}^{L}).
\end{eqnarray}
Since $\ket{\{A^{i}\}}_{L}=e^{i\alpha_L}\ket{\{\tilde{A}^{i}\}}_{L}$, for sufficiently large even integer $L$,
\begin{eqnarray}
\braket{\{A^{i}\}|\{A^{i}\}}=\braket{\{\tilde{A}^{i}\}|\{\tilde{A}^{i}\}}
&\Leftrightarrow&\tr(\mathcal{E}^{L})=\tr(\tilde{\mathcal{E}}^{L})\\
&\Leftrightarrow&\rho^{L}(1+\sum_{i}(\frac{\rho_{i}}{\rho})^{L})=\tilde{\rho}^{L}(1+\sum_{i}(\frac{\tilde{\rho}_{i}}{\tilde{\rho}})^{L}),
\end{eqnarray}
where $\rho>\rho_{1}\geq\rho_{2}\geq\cdots$ and $\tilde{\rho}>\tilde{\rho}_{1}\geq\tilde{\rho}_{2}\geq\cdots$ are eigenvalues of $\mathcal{E}$ and $\tilde{\mathcal{E}}$.
Therefore, if we take the limit of large $L$, we obtain $\rho=\tilde{\rho}$. It is obvious that this is true if instead of the condition that $L$ is even, we change it to odd.

\end{proof}

  

\begin{proof}[Proof of Lemma \ref{lem:fMPS_-}] 
First, introduce some notations. 
We denote products of matrices $B^i$s by 
\begin{align}
    B^{I^{(L)}} := B^{i_1}\cdots B^{i_L}
    \mbox{ for }
    I^{(L)} = (i_1,\dots,i_L). 
\end{align}
We denote the even and odd sets of $L$ labels of $i$s as
\begin{align}
&{\cal I}_L^{\rm e} = \{(i_1,\dots,i_L)|\sum_{k=1}^L |i_k|\equiv 0\}, \\
&{\cal I}_L^{\rm o} = \{(i_1,\dots,i_L)|\sum_{k=1}^L |i_k|\equiv 1\}, 
\end{align}
respectively.
By using these notations, the injectivity of $\{\sigma_x^{|i|} \otimes B^i\}$ is that both the sets of matrices $\{B^{I^{(l)}}| I^{(l)} \in {\cal I}^{\rm e}_l\}$ and $\{B^{I^{(l)}}|I^{(l)} \in {\cal I}^{\rm o}_l\}$ span ${\rm M}_n(\mathbb{C})$ for some $l \in \mathbb{N}$. 
Both the sets of matrices $\{B^{I^{(L)}}\}_{I^{(L)} \in {\cal I}^{\rm e}_L}$ and $\{B^{I^{(L)}}\}_{I^{(L)} \in {\cal I}^{\rm o}_L}$ can also regarded as injective MPS with the length $l=1$, but not in the canonical form because $\sum_{I^{(L)} \in {\cal I}^{\rm e/o}_L} B^{I^{(L)}} (B^{I^{(L)}})^\dag=1_{n}$ does not holds in general.
The same is true for the set of matrices $\{\tilde B^i\}_i$. 

In the following, we denote $l$ as the smallest integer such that 
all the sets of matrices $\{B^{I^{(l)}}| I^{(l)} \in {\cal I}^{\rm e}_l\}$, $\{B^{I^{(l)}}|I^{(l)} \in {\cal I}^{\rm o}_l\}$, $\{\tilde B^{I^{(l)}}| I^{(l)} \in {\cal I}^{\rm e}_l\}$, and $\{\tilde B^{I^{(l)}}|I^{(l)} \in {\cal I}^{\rm o}_l\}$ span ${\rm M}_n(\mathbb{C})$.
We introduce weight vectors $c^{\rm e}_{I^{(l)}}, c^{\rm o}_{I^{(l)}}, \tilde c^{\rm e}_{I^{(l)}}$ and $\tilde c^{\rm o}_{I^{(l)}}$ so that 
\begin{align}
\sum_{I^{(l)} \in {\cal I}^{\rm e}_l} c^{\rm e}_{I^{(l)}} B^{I^{(l)}} 
=
\sum_{I^{(l)} \in {\cal I}^{\rm o}_l} c^{\rm o}_{I^{(l)}} B^{I^{(l)}} 
=
\sum_{I^{(l)} \in {\cal I}^{\rm e}_l} \tilde c^{\rm e}_{I^{(l)}} \tilde B^{I^{(l)}} 
=
\sum_{I^{(l)} \in {\cal I}^{\rm o}_l} \tilde c^{\rm o}_{I^{(l)}} \tilde B^{I^{(l)}} 
= 1_n
\end{align}
is satisfied. 

We give a proof of Lemma \ref{lem:fMPS_-} for APBC and PBC separately. 

\medskip 

\noindent
{\it The case of APBC}---
Suppose $\{\sigma_x^{|i|} \otimes B^i\}_i \sim \{\sigma_x^{|i|} \otimes B^i\}_i$. 
This is equivalent to say that for any $L \in \mathbb{N}$ there exists $e^{\alpha_L} \in {\rm U}(1)$ such that 
\begin{align}
\tr[B^{i_1}\cdots B^{i_L}]
=e^{i\alpha_L} \tr[\tilde B^{i_1} \cdots \tilde B^{i_L}]
\mbox{ for all } (i_1,\dots,i_L) \in {\cal I}^{\rm e}_L. 
\label{eq:wave_func_B_eq_even}
\end{align}
For the integer $l$, we formally think of $I^{(l)} \in {\cal I}^{\rm e/o}_l$ as a basis at one site and consider the bosonic MPS of the set of matrices $\{B^{I^{(l)}}\}_{I^{(l)} \in {\cal I}^{e/o}_l}$. 
From (\ref{eq:wave_func_B_eq_even}), for any even integers $M \in 2 \mathbb{N}$, we have the wave function equalities 
\begin{align}
\tr[B^{I^{(l)}_1}\cdots B^{I^{(l)}_M}]
=e^{i\alpha_{lM}} \tr[\tilde B^{I^{(l)}_1}\cdots \tilde B^{I^{(l)}_M}]
\mbox{ for } (I_1^{(l)},\dots,I_M^{(l)}) \in ({\cal I}^{\rm e/o}_L)^{\times M}.
\end{align}
Applying Lemma \ref{fthmnoncano} to the bosonic MPSs of $\{B^{I^{(l)}}\}_{I^{(l)} \in {\cal I}^{\rm e}_l}$ and $\{B^{I^{(l)}}\}_{I^{(l)} \in {\cal I}^{\rm o}_l}$, there exist unique $\mathrm{U}(1)$ phases  $z^{(l)}_{\rm e}, z^{(l)}_{\rm o}$ and invertible matrices $x^{(l)}_{\rm e}, x^{(l)}_{\rm o} \in {\rm GL}_n(\mathbb{C})$ such that 
\begin{align}\label{evensim}
\tilde B^{I^{(l)}}= z^{(l)}_{\rm e} (x^{(l)}_{\rm e})^{-1} B^{I^{(l)}} x^{(l)}_{\rm e} \mbox{ for } I^{(l)} \in {\cal I}^{\rm e}_l, 
\end{align}
and 
\begin{align}\label{oddsim}
\tilde B^{I^{(l)}}= z^{(l)}_{\rm o} (x^{(l)}_{\rm o})^{-1} B^{I^{(l)}} x^{(l)}_{\rm o} \mbox{ for } I^{(l)} \in {\cal I}^{\rm o}_l
\end{align}
hold. 
Here, $x^{(l)}_{\rm e}$ and $x^{(l)}_{\rm o}$ are unique up to $\mathbb{C}^\times$ numbers. 

By Lemma~\ref{lem:LtoL+1}, we can apply the same argument to the even MPS $\{B^{I^{(l+1)}}\}_{I^{(l+1)} \in {\cal I}^{\rm e}_{l+1}}$ of length $l+1$ and obtain a unique $\mathrm{U}(1)$ phase $z^{(l+1)}_{\rm e}$ and an invertible matrix $x^{(l+1)}_{\rm e}$ such that
\begin{align}\label{Lplus1ver}
\tilde B^{I^{(l+1)}}
= z^{(l+1)}_{\rm e} (x^{(l+1)}_{\rm e})^{-1} B^{I^{(l+1)}} x^{(l+1)}_{\rm e} 
\mbox{ for } I^{(l+1)} \in {\cal I}^{\rm e}_{l+1}, 
\end{align}
where $x^{(l+1)}_{\rm e}$ is unique up to $\mathbb{C}^\times$.
Substituting Eqs.(\ref{evensim}) and (\ref{oddsim}) into Eq.(\ref{Lplus1ver}), we get
\begin{eqnarray}
\tilde B^{i_{0}} \tilde B^{I^{(l)}}
&=&z^{(l+1)}_{\rm e} (x^{(l+1)}_{\rm e})^{-1} B^{i_0} B^{I^{(l)}} x^{(l+1)}_{\rm e}\\
&=&\begin{cases}
z^{(l+1)}_{\rm e}(z^{(l)}_{\rm e})^{-1}(x^{(l+1)}_{\rm e})^{-1} B^{i_0}x^{(l)}_{\rm e}\tilde B^{I^{(l)}}(x^{(l)}_{\rm e})^{-1}x^{(l+1)}_{\rm e}&(\mbox{for }\abs{i_{0}}=0),\\
z^{(l+1)}_{\rm e}(z^{(l)}_{\rm o})^{-1}(x^{(l+1)}_{\rm e})^{-1} B^{i_0}x^{(l)}_{\rm o}\tilde B^{I^{(l)}}(x^{(l)}_{\rm o})^{-1}x^{(l+1)}_{\rm e}&(\mbox{for }\abs{i_{0}}=1).\\
\end{cases}
\end{eqnarray}
Taking linear sum $\sum_{I^{(l)} \in {\cal I}^{\rm e}_l} \tilde c^{\rm e}_{I^{(l)}}$ for $|i_0|=0$ and $\sum_{I^{(l)} \in {\cal I}^{\rm o}_l} \tilde c^{\rm o}_{I^{(l)}}$ for $|i_0|=1$, we obtain 
\begin{eqnarray}
\tilde B^{i_{0}}&=&\begin{cases}
      z^{(l+1)}_{\rm e}(z^{(l)}_{\rm e})^{-1}(x^{(l+1)}_{\rm e})^{-1}B^{i_{0}}x^{(l+1)}_{\rm e}&(\mbox{for }\abs{i_{0}}=0)\\
      z^{(l+1)}_{\rm o}(z^{(l)}_{\rm e})^{-1}(x^{(l+1)}_{\rm e})^{-1}B^{i_{0}}x^{(l+1)}_{\rm e}&(\mbox{for }\abs{i_{0}}=1)\\
  \end{cases}\nonumber\\
  &=&e^{i\beta}\eta^{\abs{i_{0}}}(x^{(l+1)}_{\rm e})^{-1}B^{i_{0}}x^{(l+1)}_{\rm e},
  \label{eq:app_gauge_B_noncan}
  \end{eqnarray}
where we have put $e^{i\beta}=z^{(l+1)}_{\rm e}(z^{(l)}_{\rm e})^{-1}$ and $\eta=z^{(l)}_{\rm e}(z^{(l)}_{\rm o})^{-1}$.
  
Next, we show $\eta=\pm1$. 
Applying the above argument to the MPS $\{B^{I^{(2l)}}\}_{I^{(2l)} \in {\cal I}^{\rm e}_{2l}}$ of length $2l$, there are a unique $z^{(2l)}_{\rm e}\in {\rm U}(1)$ and a matrix $x^{(2l)}_{\rm e}\in\mathrm{GL}_{n}(\mathbb{C})$ such that
\begin{eqnarray}\label{2Lplus1ver}
\tilde B^{I^{(l)}}\tilde B^{J^{(l)}} =z^{(2l)}_{\rm e} (x^{(2l)}_{\rm e})^{-1}B^{I^{(l)}}B^{J^{(l)}}x^{(2l)}_{\rm e}
\end{eqnarray}
for $|I^{(l)}|+|J^{(l)}|\equiv 0$. 
Substituting (\ref{evensim}) and (\ref{oddsim}) in the left hand side of (\ref{2Lplus1ver}) yields the equation
\begin{eqnarray}\label{etaeta}
z^{(2l)}_{\rm e}(x^{(2l)}_{\rm e})^{-1}B^{I^{(l)}}B^{J^{(l)}}x^{(2l)}_{\rm e}
=\begin{cases}
(z^{(l)}_{\rm e})^{2}(x^{(l)}_{\rm e})^{-1} B^{I^{(l)}} B^{J^{(l)}}x^{(l)}_{\rm e}&(\mbox{for }|I^{(l)}|\equiv|J^{(l)}|\equiv 0),\\
(z^{(l)}_{\rm o})^{2}(x^{(l)}_{\rm o})^{-1} B^{I^{(l)}} B^{J^{(l)}}x^{(l)}_{\rm o}&(\mbox{for }|I^{(l)}|\equiv|J^{(l)}|\equiv 1).\\
\end{cases}
\end{eqnarray}
Taking the linear sum $\sum_{I^{(l)},J^{(l)} \in {\cal I}^{\rm e}_l} c^{\rm e}_{I^{(l)}} c^{\rm e}_{J^{(l)}}$ for $|I^{(l)}|\equiv|J^{(l)}|\equiv 0$ and $\sum_{I^{(l)},J^{(l)} \in {\cal I}^{\rm o}_l} c^{\rm o}_{I^{(l)}} c^{\rm o}_{J^{(l)}}$ for $|I^{(l)}|\equiv|J^{(l)}|\equiv 1$, we obtain $z'=z_{\rm e}^{2}=z_{\rm o}^2$. 
Therefore $\eta=\pm 1$.

Finlay, we show that $x^{(l+1)}_{\rm e}$ can be unitary.
From (\ref{eq:app_gauge_B_noncan}), we have 
\begin{align}
    \sum_i B^i (x^{(l+1)}_{\rm e}(x^{(l+1)}_{\rm e})^\dag) B^{i\dag} = x^{(l+1)}_{\rm e}(x^{(l+1)}_{\rm e})^\dag. 
\end{align}
Since the set of matrices $\{B^i\}_i$ is injective in the bosonic sense and in the canonical form, $1_n$ is the only eigenmatrix of the eigenvalue $1$ of the transfer matrix ${\cal E}_B(X)=\sum_i B^i X B^{i\dag}$. 
Therefore, $x^{(l+1)}_{\rm e}(x^{(l+1)}_{\rm e})^\dag = \lambda 1_n$ with $\lambda \in \mathbb{C}^\times$.
Normalizing $x^{(l+1)}_{\rm e}$ to $x^{(l+1)}_{\rm e}(x^{(l+1)}_{\rm e})^\dag =1_n$, we conclude that $x^{(l+1)}_{\rm e}$ is unitary and unique up to $\mathrm{U}(1)$ phase. 

\medskip 

\noindent
{\it The case of PBC}---
Suppose $\{\sigma_x^{|i|} \otimes B^i\}_i \sim_{\rm PBC} \{\sigma_x^{|i|} \otimes B^i\}_i$. 
This is equivalent to say that for any $L \in \mathbb{N}$ there exists $e^{\alpha_L} \in {\rm U}(1)$ such that 
\begin{align}
\tr[B^{i_1}\cdots B^{i_L}]
=e^{i\alpha_L} \tr[\tilde B^{i_1} \cdots \tilde B^{i_L}]
\mbox{ for all } (i_1,\dots,i_L) \in {\cal I}^{\rm o}_L. 
\label{eq:wave_func_B_eq_odd}
\end{align}
From (\ref{eq:wave_func_B_eq_odd}), for any odd integers $M \in 2 \mathbb{N}+1$, we have the wave function equalities 
\begin{align}
\tr[B^{I^{(l)}_1}\cdots B^{I^{(l)}_M}]
=e^{i\alpha_{lM}} \tr[\tilde B^{I^{(l)}_1}\cdots \tilde B^{I^{(l)}_M}]
\mbox{ for } (I_1^{(l)},\dots,I_M^{(l)}) \in ({\cal I}^{\rm o}_L)^{\times M}.
\end{align}
Applying Lemma \ref{fthmnoncano} to the bosonic MPS of $\{B^{I^{(l)}}\}_{I^{(l)} \in {\cal I}^{\rm o}_l}$, there exist a unique $\mathrm{U}(1)$ phase $z^{(l)}_{\rm o}$ and an invertible matrix $x^{(l)}_{\rm o} \in {\rm GL}_n(\mathbb{C})$ such that 
\begin{align}\label{oddsim_2}
\tilde B^{I^{(l)}}= z^{(l)}_{\rm o} (x^{(l)}_{\rm o})^{-1} B^{I^{(l)}} x^{(l)}_{\rm o} \mbox{ for } I^{(l)} \in {\cal I}^{\rm o}_l
\end{align}
holds. 
Here, $x^{(l)}_{\rm o}$ is unique up to $\mathbb{C}^\times$ numbers. 
Taking the linear sum $\sum_{I^{(l) \in {\cal I}^{\rm o}_l}}\tilde c^{\rm o}_{I^{l}}$, we get 
\begin{align}\label{eq:sum_tilde_c_o_B}
\sum_{I^{(l)} \in {\cal I}^{\rm o}_l} \tilde c^{\rm o}_{I^{(l)}} B^{I^{(l)}}
= (z^{(l)}_{\rm o})^{-1}.
\end{align}

In the same way, for the length $2l$, there are a unique $z^{(2l)}_{\rm o}\in {\rm U}(1)$ and a matrix $x^{(2l)}_{\rm o}\in\mathrm{GL}_{n}(\mathbb{C})$ such that
\begin{eqnarray}
\tilde B^{I^{(l)}}\tilde B^{J^{(l)}} =z^{(2l)}_{\rm o}(x^{(2l)}_{\rm o})^{-1}B^{I^{(l)}}B^{J^{(l)}}x^{(2l)}_{\rm o}
\end{eqnarray}
for any $I^{(l)}$ and $J^{(l)}$ satisfying  $|I^{(l)}|+|J^{(l)}|\equiv 1$.
Here, $x$ is unique up to $\mathbb{C}^\times$. 
Taking the linear sum $\sum_{I^{(l)} \in {\cal I}^{\rm o}_l} \tilde c^{\rm o}_{I^{l}}$, we get 
\begin{eqnarray}\label{evensim_2}
\tilde B^{J^{(l)}} =z^{(2l)}_{\rm o}(z^{(l)}_{\rm o})^{-1}(x^{(2l)}_{\rm o})^{-1}B^{J^{(l)}}x^{(l)}_{\rm o} \;\mbox{ for }\; J^{(l)} \in {\cal I}^{\rm e}_l, 
\end{eqnarray}
and the linear sum $\sum_{J^{(l)} \in {\cal I}^{\rm e}_l}\tilde c^{\rm e}_{J^{(l)}}$ gives us 
\begin{eqnarray}\label{eq:sum_tilde_c_e_B}
\sum_{J^{(l)} \in {\cal I}^{\rm e}_l}\tilde c^{\rm e}_{J^{(l)}}B^{J^{(l)}} = z^{(2l)}_{\rm o}(z^{(l)}_{\rm o})^{-1}.
\end{eqnarray}

By Lemma~\ref{lem:LtoL+1}, we can apply the same argument to the odd MPS $\{B^{I^{(l+1)}}\}_{I^{(l+1)} \in {\cal I}^{\rm o}_{l+1}}$ of length $l+1$ and obtain a unique $\mathrm{U}(1)$ phase $z^{(l+1)}_{\rm o}$ and an invertible matrix $x^{(l+1)}_{\rm o}$ such that
\begin{align}
\tilde B^{I^{(l+1)}}= z^{(l+1)}_{\rm o} (x^{(l+1)}_{\rm o})^{-1} B^{I^{(l+1)}} x^{(l+1)}_{\rm o} \mbox{ for } I^{(l+1)} \in {\cal I}^{\rm o}_{l+1}, 
\end{align}
where $x^{(l+1)}_{\rm o}$ is unique up to $\mathbb{C}^\times$.
Using (\ref{oddsim_2}) and (\ref{evensim_2}), we have 
\begin{eqnarray}
\tilde B^{i_0} \tilde B^{I^{(l)}}
&=&z^{(l+1)}_{\rm o} (x^{(l+1)}_{\rm o})^{-1} B^{i_0} B^{I^{(l)}} x^{(l+1)}_{\rm o}\\
&=&\begin{cases}
z^{(l+1)}_{\rm o}(z^{(l)}_{\rm o})^{-1} (x^{(l+1)}_{\rm o})^{-1} B^{i_0} x^{(l)}_{\rm o} \tilde B^{I^{(l)}} (x^{(l)}_{\rm o})^{-1} x^{(l+1)}_{\rm o} &(\mbox{for }\abs{i_{0}}=0),\\
z^{(l+1)}_{\rm o}(z^{(2l)}_{\rm o})^{-1}z^{(l)}_{\rm o} (x^{(l+1)}_{\rm o})^{-1} B^{i_0} x^{(2l)}_{\rm o} \tilde B^{I^{(l)}} (x^{(2l)}_{\rm o})^{-1} x^{(l+1)}_{\rm o} &(\mbox{for }\abs{i_{0}}=1).\\
\end{cases}
\end{eqnarray}
The linear sum $\sum_{I^{(l)} \in {\cal I}^{\rm o}_l} \tilde c^{\rm o}_{I^{(l)}}$ for $|i_0|=0$ and $\sum_{I^{(l)} \in {\cal I}^{\rm e}_l} \tilde c^{\rm e}_{I^{(l)}}$ for $|i_0|=1$ leads to 
\begin{align}
\tilde B^{i_0}
&=\begin{cases}
z^{(l+1)}_{\rm o}(z^{(l)}_{\rm o})^{-1} (x^{(l+1)}_{\rm o})^{-1} B^{i_0} x^{(l+1)}_{\rm o} &(\mbox{for }\abs{i_{0}}=0),\\
z^{(l+1)}_{\rm o}(z^{(2l)}_{\rm o})^{-1}z^{(l)}_{\rm o} (x^{(l+1)}_{\rm o})^{-1} B^{i_0} x^{(l+1)}_{\rm o} &(\mbox{for }\abs{i_{0}}=1),\\
\end{cases} \nonumber\\
&=
e^{i\beta} \eta^{|i_0|} (x^{(l+1)}_{\rm o})^{-1} B^{i_0} x^{(l+1)}_{\rm o}, 
\end{align}
where we have put $e^{i\beta}=z^{(l+1)}_{\rm o}(z^{(l)}_{\rm o})^{-1}$ and $\eta=(z^{(2l)}_{\rm o})^{-1} (z^{(l)}_{\rm o})^2$.

We show $\eta=\pm1$. 
Applying the above argument to the MPS $\{B^{I^{(3l)}}\}_{I^{(3l)} \in {\cal I}^{\rm o}_{3l}}$ of length $3l$, there are a unique $z^{(3l)}_{\rm o}\in {\rm U}(1)$ and a matrix $x^{(3l)}_{\rm o}\in\mathrm{GL}_{n}(\mathbb{C})$ such that
\begin{eqnarray}\label{2Lplus1ver_2}
\tilde B^{I^{(l)}}\tilde B^{J^{(l)}}\tilde B^{K^{(l)}} 
=z^{(3l)}_{\rm o}(x^{(3l)}_{\rm o})^{-1}B^{I^{(l)}}B^{J^{(l)}}B^{K^{(l)}}x^{(3l)}_{\rm o}
\end{eqnarray}
for $|I^{(l)}|+|J^{(l)}|+|K^{(l)}|\equiv 1$. 
By using (\ref{eq:sum_tilde_c_o_B}) and (\ref{eq:sum_tilde_c_e_B}), the linear sums $\sum_{I^{(l)},J^{(l)},K^{(l)} \in {\cal I}^{\rm o}_l} \tilde c^{\rm o}_{I^{(l)}} \tilde c^{\rm o}_{J^{(l)}} \tilde c^{\rm o}_{K^{(l)}}$ and $\sum_{I^{(l)},J^{(l)} \in {\cal I}^{\rm e}_l, K^{(l)} \in {\cal I}^{\rm o}_l} \tilde c^{\rm e}_{I^{(l)}} \tilde c^{\rm e}_{J^{(l)}} \tilde c^{\rm o}_{K^{(l)}}$ gives $z^{(3l)}_{\rm o} (z^{(l)}_{\rm o})^3 = z^{(3l)}_{\rm o} ((z^{(2l)}_{\rm o})^{-1}z^{(l)}_{\rm o})^2z^{(l)}_{\rm o}$ which leads to $\eta^2=1$.

In the same way as in the case of APBC, $x^{(l+1)}_{\rm o}$ can be unitary and is unique up to $\mathrm{U}(1)$ phases.
\end{proof}

\section{A Proof of Theorem \ref{eq}}\label{prf}

In this section, we prove the theorem \ref{eq} and determine the necessary and sufficient conditions for the algebra $\mathcal{A}$ to be $\zmodtwo$-graded central simple algebra in the $4\times4$ MPS when $A^{0}=1$. First of all, we prove the following lemma.
\begin{lemma}\quad\label{minimal}\\
Let $A^{0}=1_{4}$, and let $\mathcal{A}$ be the $\zmodtwo$-graded algebra generated by $A^{0}$ and $A^{1}$. In this case, the structure of $\mathcal{A}$ is given by
\begin{eqnarray}
\mathcal{A}\simeq \mathbb{C}\left[t\right]/\left(f\right)
\end{eqnarray}
as a $\zmodtwo$-graded algebra. Here, the right-hand side is regarded as a $\zmodtwo$-graded algebra with the degree of $t$ being $1$, and $\simeq$ means an isomorphism, and $f$ be the minimal  polynomial of $A^{1}$, and $\left(f\right)$ denote the two-sided ideals generated by $f$.
\end{lemma}
\begin{proof}
Let $p:\mathbb{C}\left[t\right]\to \mathcal{A}$ be the map $t\mapsto A^{1}$. All we have to do is show $ker\left(p\right)\simeq\left(f\right)$.
\begin{itemize}
  \item $\left(f\right)\subset ker\left(p\right)$ : This is obvious.
  \item $ker\left(p\right) \subset \left(f\right)$ : For any $g\in\ker(p)$, $g(A^{1})=0$ by definition. Since $f$ is the minimal polynomial, $f$ divides $g$, so $g\in\left(f\right)$, i.e. $\ker(g)\subset\left(f\right)$.
\end{itemize}
Therefore, $ker\left(p\right)\simeq\left(f\right)$ and from the fundamental theorem on homomorphisms theorem, 
\begin{eqnarray}
\mathcal{A}\simeq\mathbb{C}\left[t\right]\big/\left(f\right)
\end{eqnarray}
as a $\zmodtwo$-graded algebra.
\end{proof}
From the lemma \ref{minimal}, we need to find the minimal polynomial $f$ of $A^{1}$. Since the degree of $A^{1}$ is odd, $f$ has degree 2 at least. Since
\begin{eqnarray}
(A^{1})^{2}=
\begin{pmatrix}
(\tilde{A}^{1})^{2}&0\\
0&(\tilde{A}^{1})^{2}\\
\end{pmatrix},
\end{eqnarray}
denoting eigenvalues of $\left(A^{1}\right)^{2}$ by $\alpha$ and $\beta$, then
\begin{eqnarray}
\alpha=\beta&\Rightarrow&f(t)=(t^{2}-\alpha)(t^{2}-\beta)\\
\alpha\neq\beta&\Rightarrow& f(t)=(t^{2}-\alpha).
\end{eqnarray}
Therefore, the structure of the algebra $\mathcal{A}$ is determined as follows.
\begin{proposition}\hspace{1mm}\\
Let $\alpha$ and $\beta$ be the eigenvalues of $-(\tilde{A}^{1})^{2}$. Then the following holds. 
\begin{eqnarray}
\mathcal{A}\simeq
\begin{cases}
    \mathbb{C}\left[t\right]\big/(t^{2}-\alpha)(t^{2}-\beta) & (\alpha=\beta) \\
    \mathbb{C}\left[t\right]\big/(t^{2}-\alpha) & (\alpha\neq\beta)
  \end{cases}
\end{eqnarray}
where $t$ has a degree $1$ and $\simeq$ means isomorphism as $\zmodtwo$-graded algebra.
\end{proposition}
Therefore, the structure of the ideals in $\mathcal{A}$ can be classified as follows:
\begin{itemize}
 \item The case of $\alpha\neq\beta$.\\
 Using Chinese remainder theorem, we can decompose $\mathcal{A}$ into 
 \begin{eqnarray}
 \mathbb{C}\left[t\right]\big/(t^{2}-\alpha)(t^{2}-\beta)\simeq\mathbb{C}\left[t\right]\big/(t^{2}-\alpha)\times\mathbb{C}\left[t\right]\big/(t^{2}-\beta).
 \end{eqnarray}
 Thus $\mathcal{A}$ is not simple since each component is a subalgebra of $\mathcal{A}$.
 \item The case of $\alpha=\beta$.\\
 When $\alpha=\beta=0$, since $\mathcal{A}$ has a non trivial ideal $\left(t\right)$, $\mathcal{A}$ is not simple. When $\alpha=\beta\neq0$, $\left(t-\alpha\right)$ is a ideal of $\mathcal{A}$ but not $\zmodtwo$-graded algebra. Therefore, in this case, $\mathcal{A}$ is central simple as a $\zmodtwo$-graded algebra.
\end{itemize}
The above shows that when $A^{0}=1_{4}$, the necessary and sufficient condition for $\mathcal{A}$ to be a $\zmodtwo$-graded central simple algebra is
\begin{eqnarray}\label{simple}
\alpha=\beta\neq0
\end{eqnarray}
where $\alpha$ and $\beta$ are eigenvalues of $(\tilde{A}^{1})^{2}$.

Next, we denote
\begin{eqnarray}
\tilde{A}^{1}=
\begin{pmatrix}
a&b\\
c&d\\
\end{pmatrix}
\end{eqnarray}
and rewrite this condition for components $a,b,c$ and $d$. Since the square of $\tilde{A}^{1}$ is 
\begin{eqnarray}
(\tilde{A}^{1})^{2}=
\begin{pmatrix}
a^{2}+bc&ab+bd\\
ac+dc&bc+d^{2}\\
\end{pmatrix},
\end{eqnarray}
the eigenvalue of $-(\tilde{A}^{1})^{2}$ is 
\begin{eqnarray}
\det\left(\lambda-(\tilde{A}^{1})^{2}\right)=0\Leftrightarrow\lambda=\frac{(a^{2}+d^{2}+2bc)\pm\sqrt{(a+d)^{2}\{(a-d)^{2}+4bc\}}}{2}.
\end{eqnarray}
Therefore, eq.(\ref{simple}) can be rewritten as 
\begin{eqnarray}
\mathrm{eq.}(\ref{simple})\Leftrightarrow
\begin{cases}
    (a+d)^{2}\{(a-d)^{2}+4bc\}=0 \\
    \hspace{5mm}\mathrm{and}\\
    a^{2}+d^{2}+2bc\neq0
  \end{cases}
\end{eqnarray}
in terms of components. In particular, we rewrite the first condition as follows.
\begin{itemize}
  \item In the case of $a+d=0\;(\Leftrightarrow\tr\tilde{A}^{1}=0)$:
  \begin{eqnarray}
a^{2}+d^{2}+2bc\neq0\Leftrightarrow (a+d)^{2}-2(ad-bc)\neq0\Leftrightarrow\det(\tilde{A}^{1})\neq0
\end{eqnarray}
  \item In the case of $(a-d)^{2}+4bc=0\;(\Leftrightarrow(a+d)^{2}-4(ad-bc)=0\Leftrightarrow\tr(\tilde{A}^{1})^{2}-4\det(\tilde{A}^{1})=0)$:
   \begin{eqnarray}
a^{2}+d^{2}+2bc\neq0\Leftrightarrow\tr(\tilde{A}^{1})^{2}-2\det(\tilde{A}^{1})\neq0\Leftrightarrow\det(\tilde{A}^{1})\neq0
\end{eqnarray}
\end{itemize}

Finally, we obtained the following formula.
\begin{eqnarray}
\mathrm{eq.}(\ref{simple})\Leftrightarrow\det(\tilde{A^{1}})\neq0\;\;\mathrm{and}\;\;
\begin{cases}
    \tr(\tilde{A^{1}})=0 \\
    \hspace{5mm}\mathrm{or}\\
    \tr(\tilde{A^{1}})^{2}-4\det(\tilde{A^{1}})=0
  \end{cases}
\end{eqnarray}
This is the claim of Theorem \ref{eq}.

\section{A Proof of Theorem \ref{top}}\label{prf2}

In this section, we prove Theorem \ref{top} and determine the topology of the space of MPS when $A^{0}=1_{4}$. As we saw in Appendix \ref{prf}, the necessary and sufficient conditions for $\mathcal{A}$ to be a $\zmodtwo$-graded central simple algebra were given by (i) $\det(\tilde{A}^{1})\neq0$, and $\tr(\tilde{A}^{1})=0$ or (ii) $\det(\tilde{A}^{1})\neq0$, and $\tr(\tilde{A}^{1})^{2}-4\det(\tilde{A}^{1})=0$. In the following, we determine the topology of the space represented by each of the cases (i) and (ii).

\begin{itemize}
 \item (i)$\;\det(\tilde{A}^{1})\neq0,\;\tr(\tilde{A}^{1})=0$\\
 $\quad$ First, recall the polar decomposition theorem for complex matrices.
 \begin{theorem}\;\\
 Let $A\in {\rm M}_n(\mathbb{C})$ be a $n\times n$ matrix, then there exist a unitary matrix $U\in {\rm U}(n)$ and positive-semidefinite Hermitian matrix $P$ such that
 \begin{eqnarray}
 A=U\cdot P.
 \end{eqnarray}
 In particular, if $A\in\mathrm{GL}(n;\mathbb{C})$, $U$ and $P$ are unique.
 \end{theorem}
 Applying the polar decomposition theorem to $A^{1}$, there is a unique unitary matrix $U$ and positive-semidefinite Hermitian matrix $P$ such that $A^{1}=U\cdot P$, since the determinant of $A^{1}$ is not zero. In addition, since $P$ is Hermitian, it can be diagonalized using the unitary matrix $V$:
\begin{eqnarray}
 P=V\cdot\Lambda\cdot V^{\dagger},\hspace{3mm}
 \Lambda=\begin{pmatrix}
\lambda_{1}&0\\
0&\lambda_{2}\\
\end{pmatrix}
 \end{eqnarray}
 Here, from the semi-positivity of $P$, $\lambda_{1},\lambda_{2}\in\mathbb{R}_{>0}$. Therefore, $\tilde{A}^{1}=U\cdot V\cdot\Lambda\cdot V^{\dagger}$.
 
  Next, we calculate $\tr(\tilde{A}^{1})$. First, from the cyclicity of  trace,
   \begin{eqnarray}
 \tr(\tilde{A}^{1})=\tr(U\cdot V\cdot\Lambda\cdot V^{\dagger})=\tr(V^{\dagger}U\cdot V\cdot\Lambda).
 \end{eqnarray}
 Any $2\times 2$ unitary matrix $W\in U(2)$ can be written by
  \begin{eqnarray}
 W=e^{i\theta}\begin{pmatrix}
a&-b^{\ast}\\
b&a^{\ast}\\
\end{pmatrix}
 \end{eqnarray}
 where $a,b\in\mathbb{C}$ satisfying $\left|a\right|^{2}+\left|b\right|^{2}=1$ and $\theta\in[0,2\pi)$. Therefore, we denote $V^{\dagger}UV\in {\rm U}(2)$ as 
 \begin{eqnarray}
 V^{\dagger}UV=e^{i\theta}\begin{pmatrix}
a&-b^{\ast}\\
b&a^{\ast}\\
\end{pmatrix}
 \end{eqnarray}
 for some $a,b$ and $\theta$ as described above, we can rewrite $\tilde{A}^{1}$ as
  \begin{eqnarray}
 \tr(\tilde{A}^{1})=e^{i\theta}\tr\left(
 \begin{pmatrix}
a&-b^{\ast}\\
b&a^{\ast}\\
\end{pmatrix}
 \begin{pmatrix}
\lambda_{1}&0\\
0&\lambda_{2}\\
\end{pmatrix}
 \right)=e^{i\theta}(a\lambda_{1}+a^{\ast}\lambda_{2}).
 \end{eqnarray}
 Accordingly, we obtained the following necessary and sufficient conditions:
 \begin{eqnarray}
 \tr(\tilde{A}^{1})=0 \Leftrightarrow a\lambda_{1}+a^{\ast}\lambda_{2}=0
 \end{eqnarray}
 
  \begin{enumerate}
 \item In the case of $a\neq0$.\\
 In this case, $\lambda_{1}=\lambda_{2}$ since $\lambda_{1}=-\frac{a^{\ast}}{a}\lambda_{2}$ and $\lambda_{1}$ and $\lambda_{2}$ are positive real number. Therefore, 
 \begin{eqnarray}
 \tilde{A}^{1}=U\cdot V\cdot\Lambda\cdot V^{\dagger}=\lambda U,\hspace{3mm}(\lambda\in\mathbb{R}_{>0})
 \end{eqnarray}
 \item In the case of $a=0$.\\
 In this case, it is obvious that $\tr(\tilde{A}^{1})=0$ for any $\lambda_{1}$ and $\lambda_{2}$. By substituting
 \begin{eqnarray}
 V^{\dagger}UV\Lambda=e^{i\theta}\begin{pmatrix}
0&-b^{\ast}\\
b&0\\
\end{pmatrix}
 \begin{pmatrix}
\lambda_{1}&0\\
0&\lambda_{2}\\
\end{pmatrix}
 \end{eqnarray}
 into $\tilde{A}^{1}=UV\Lambda V^{\dagger}=V\left(V^{\dagger}UV\Lambda\right)V^{\dagger}$, we get 
 \begin{eqnarray}
 \tilde{A}^{1}=\lambda_{1} e^{i\theta}V\begin{pmatrix}
0&-b^{\ast}\\
b&0\\
\end{pmatrix}
 \begin{pmatrix}
1&0\\
0&\lambda_{2}/\lambda_{1}\\
\end{pmatrix}
V^{\dagger}.
 \end{eqnarray}
 We can use continuous deformation to make $\lambda_{2}/\lambda_{1}=1$ since $\lambda_{2}/\lambda_{1}>0$. Accordingly,  we obtain
 \begin{eqnarray}
 \tilde{A}^{1}=V^{\dagger}UV\Lambda=\lambda_{1} e^{i\theta}V\begin{pmatrix}
0&-b^{\ast}\\
b&0\\
\end{pmatrix}
V^{\dagger}
 \end{eqnarray}
 and this comes down to the case of $a\neq0$.
 \end{enumerate}
 As a result, it was found that the range of $\tilde{A}^{1}$ is
  \begin{eqnarray}
 \{\lambda e^{i\theta}U\left|\lambda>0,U\in\mathrm{U}\left(2\right),\tr(U)=0\right.\},
 \end{eqnarray}
 We need, therefore, to find the topology of this space. Let $U\in {\rm U}(2)$ be
 \begin{eqnarray}
 U=\begin{pmatrix}
a&-b^{\ast}\\
b&a^{\ast}\\
\end{pmatrix}.
 \end{eqnarray}
Since $\tr{(U)}=0$, $a=-a^{\ast}$ and therefore $a$ is pure imaginary. We define $x=-ia\in\mathbb{R}$ and by using $x^{2}+\abs{b}^{2}=1$, we get
 \begin{eqnarray}
 U&=&\begin{pmatrix}
ix&-\sqrt{1-\left|x\right|^{2}}e^{-i\varphi}\\
\sqrt{1-\left|x\right|^{2}}e^{i\varphi}&-ix\\
\end{pmatrix}\\
&=&\begin{pmatrix}
i\cos(\chi)&-e^{-i\varphi}\sin(\chi)\\
e^{i\varphi}\sin(\chi)&-i\cos(\chi)\\
\end{pmatrix}
 \end{eqnarray}
 where $\varphi$ is phase of $b$ and $x=\cos{(\chi)}$. We can see that $\chi,\phi$ are the coordinates of $S^{2}$ which is the equator of $S^{3}\sim\mathrm{SU}\left(2\right)$. Note that $\lambda$ does not contribute to the homotopy of $\mathcal{M}$, so the topology is $\cfrac{S^{1}\times S^{2}}{\zmodtwo_{diag}}$. Also, the case $a=0$ can be deformed smoothly to the equator $S^{2}$ by continuous deformation. The above shows that condition (i)$\;\det(\tilde{A}^{1})\neq0,\;\tr(\tilde{A}^{1})=0$ is homotopic to $\mathbb{R}_{>0}\times \cfrac{S^{1}\times S^{2}}{\zmodtwo_{diag}}$.
 
 \item (ii)$\;\det(\tilde{A}^{1})\neq0,\;\tr(\tilde{A}^{1})^{2}-4\det(\tilde{A}^{1})=0$\\
 By using polar decomposition theorem again, we get
  \begin{eqnarray}
 \tr(\tilde{A}^{1})^{2}-4\det(\tilde{A}^{1})=0&\Leftrightarrow&\tr(UV\Lambda V^{\dagger})^{2}-4\det(UV\Lambda V^{\dagger})=0\\
 &\Leftrightarrow&\tr(V^{\dagger}UV\Lambda)^{2}-4\det(V^{\dagger}UV)\det(\Lambda )=0.\label{decomp}
 \end{eqnarray}
 Since $V^{\dagger}UV\in {\rm U}(2)$, this can be denoted
  \begin{eqnarray}
 V^{\dagger}UV=e^{i\theta}\begin{pmatrix}
a&-b^{\ast}\\
b&a^{\ast}\\
\end{pmatrix},\hspace{3mm}\Lambda=\begin{pmatrix}
\lambda_{1}&\\
&\lambda_{2}\\
\end{pmatrix}\label{uniuni}
 \end{eqnarray}
 using $\theta\in[0,2\pi)$ and $a,b\in\mathbb{C}$ such that $\abs{a}^{2}+\abs{b}^{2}=1$. By substituting this into eq.(\ref{decomp}), we get 
 \begin{eqnarray}
  &\quad&e^{2i\theta}\tr\left(\begin{pmatrix}
a&-b^{\ast}\\
b&a^{\ast}\\
\end{pmatrix}
\begin{pmatrix}
\lambda_{1}&\\
&\lambda_{2}\\
\end{pmatrix}
  \right)^{2}-4e^{2i\theta}\lambda_{1}\lambda_{2}=0\\&\Leftrightarrow&(a\lambda_{1}+a^{\ast}\lambda_{2})^{2}-4\lambda_{1}\lambda_{2}=0\\
  &\Leftrightarrow&a^{2}\lambda_{2}^{2}+2\left|a\right|^{2}\lambda_{1}\lambda_{2}+a^{\ast 2}\lambda_{2}^{2}-4\lambda_{1}\lambda_{2}=0\\
  &\Leftrightarrow&a^{2}\left(\frac{\lambda_{2}}{\lambda_{1}}\right)^{2}+2(\left|a\right|^{2}-2)\frac{\lambda_{2}}{\lambda_{1}}+a^{\ast 2}=0.\label{quad}
 \end{eqnarray}
 Therefore,
  \begin{eqnarray}
 \frac{\lambda_{1}}{\lambda_{2}}=\frac{-\left|a\right|^{2}+2\pm\sqrt{1-\left|a\right|^{2}}}{a^{2}}=\frac{1+\left|b\right|^{2}\pm2\left|b\right|}{a^{2}}=\frac{(1\pm\left|b\right|)^{2}}{a^{2}}.
 \end{eqnarray}
  In particular, since $a\in\mathbb{R}$ (due to $\frac{\lambda{1}}{\lambda_{2}}\in\mathbb{R}_{>0}$), we get 
  \begin{eqnarray}
 \frac{(1\pm\left|b\right|)^{2}}{a^{2}}=\frac{(1\pm\left|b\right|)^{2}}{(1+\left|b\right|)(1-\left|b\right|)}=\frac{1\pm\left|b\right|}{1\mp\left|b\right|}.
 \end{eqnarray}
 Suppose $\lambda_{1}\geq\lambda_{2}$ without loss of generality, 
  \begin{eqnarray}
 \frac{\lambda_{1}}{\lambda_{2}}=\frac{1+\left|b\right|}{1-\left|b\right|}=:r\hspace{3mm}(1\leq r<\infty).
 \end{eqnarray}
 Substituting this into eq.(\ref{quad}),
 \begin{eqnarray}
 \lambda^{2}(ra+a^{\ast})^{2}-4\lambda^{2}r=0
 \end{eqnarray}
 and simplify the above equation, paying attention to $a=a^{\ast}$, we get 
  \begin{eqnarray}
 a=\frac{\pm2\sqrt{r}}{1+r},\hspace{3mm}\left|b\right|=\sqrt{1-\left(\frac{2\sqrt{r}}{1+r}\right)}=\frac{r-1}{1+r}.
 \end{eqnarray}
 Substituting this result into eq.(\ref{uniuni}), 
  \begin{eqnarray}
 V^{\dagger}UV=\begin{pmatrix}
\cfrac{\pm2\sqrt{r}}{1+r}&-\cfrac{r-1}{1+r}e^{-i\varphi}\\
\cfrac{r-1}{1+r}e^{i\varphi}&\cfrac{\pm2\sqrt{r}}{1+r}\\
\end{pmatrix}
 \end{eqnarray}
 and finally we get
 \begin{eqnarray}
 \tilde{A}^{1}=UV\Lambda V^{\dagger}=\lambda e^{i\theta}V\begin{pmatrix}
\cfrac{\pm2\sqrt{r}}{1+r}&-\cfrac{r-1}{1+r}e^{-i\varphi}\\
\cfrac{r-1}{1+r}e^{i\varphi}&\cfrac{\pm2\sqrt{r}}{1+r}\\
\end{pmatrix}\begin{pmatrix}
r&\\
&1\\
\end{pmatrix}V^{\dagger}.
 \end{eqnarray}
 The matrix part is homotopic to two $2$-dimensional open disks $D^{2}_{open}$, and each of which contains $\pm1\simeq\zmod{2}\in\mathrm{SU}\left(2\right)$.
The above shows that condition (ii)$\;\det(\tilde{A}^{1})\neq0,\;\tr(\tilde{A}^{1})^{2}-4\det(\tilde{A}^{1})=0$ is homotopic to $\mathbb{R}_{>0}\times \cfrac{S^{1}\times \zmod{2}}{\zmodtwo_{diag}}$.
\end{itemize}

\section{The Berry Phase}\label{App:BerryPhase}

In Sec.\ref{subsubsec:2by21f}, we construct a non-trivial path $\{a^{1}=\lambda e^{i\theta}\left|\right.\theta\in\left[0,\pi\right]\}$ in $\mathcal{M}_{n=1,N=2}^{\mathrm{non}\hyphen \mathrm{trivial}}$. In this section, let's compute the Berry phase of the ground state along this path. Although, in general, the Berry phase is not quantized, we claim that the ratio of them in periodic and anti-periodic systems is quantized for large system size limit and it is a candidate for invariant of the pump. 

The fermionic MPSs with the Wall matrix $u$ are given by 
\begin{eqnarray}
\ket{\{A^{i}(\theta),u(\theta)\}}&=&\sum_{\{i_{k}\}}\tr\left(u A^{i_1}\cdots A^{i_{L}}\right)\ket{i_1,...i_L}\\
&=&\sum_{\{i_{k}\},{\rm odd}}\left(\lambda e^{i\theta}\right)^{\sum_{k}\left|i_{k}\right|}\left(-1\right )^{\frac{\sum_{\{i_{k}\}}\left|i_{k}\right|+1}{2}}\ket{i_1,...i_L}\\
&=&\sum_{\{i_{k}\},{\rm odd}}i\left(i\lambda e^{i\theta}\right)^{\sum_{k}\left|i_{k}\right|}\ket{i_1,...i_L},
\end{eqnarray}
where $\sum_{\{i_{k}\}}$ means summing over all combinations of $\{i_{1},.... ,i_{L}\}$, and $\sum_{\{i_{k}\},{\rm odd}}$ means summing over all combinations whose sum is odd. Since the normalized ground state is given by
\begin{eqnarray}
\ket{\Phi(\theta)}&:=&\cfrac{1}{\sqrt{\braket{\MPS|\MPS}}}\ket{\MPS}\\
&=&\frac{1}{\sqrt{\sum_{\{i_{k}\},{\rm odd}}\left(\lambda^{2}\right)^{\sum_{k}\left|i_{k}\right|}}}\sum_{\{i_{k}\},{\rm odd}}i\left(i\lambda e^{i\theta}\right)^{\sum_{k}\left|i_{k}\right|}\ket{i_1,...i_L},
\end{eqnarray}
the Berry connection $\mathscr{A}(\theta)$ is
\begin{eqnarray}
\mathscr{A}(\theta)&=&\bra{\Phi(\theta)}\frac{\partial}{\partial\theta}\ket{\Phi(\theta)}\\
&=&\cfrac{1}{\sum_{\{i_{k}\},{\rm odd}}\left(\lambda^{2}\right)^{\sum_{k}\left|i_{k}\right|}}\sum_{\{i_{k}\},{\rm odd}}\left(i\sum_{k}\left|i_{k}\right|\right)\left(\lambda^{2}\right)^{\sum_{k}\left|i_{k}\right|}\\
&=&\cfrac{\lambda^{2}}{\sum_{\{i_{k}\},{\rm odd}}\left(\lambda^{2}\right)^{\sum_{k}\left|i_{k}\right|}}\left(\cfrac{d}{d\lambda^{2}}\sum_{\{i_{k}\},{\rm odd}}(\lambda^{2})^{\sum_{k}\left|i_{k}\right|}\right)\\
&=&L\cdot\left(1-\cfrac{\left(1+\lambda^{2}\right)^{L-1}-\left(1-\lambda^{2}\right)^{L-1}}{\left(1+\lambda^{2}\right)^{L}-\left(1-\lambda^{2}\right)^{L}}\right).
\end{eqnarray}
Note that we used the identity
\begin{eqnarray}
\sumodd\left(\lambda^{2}\right)^{\sum_{k}\left|i_{k}\right|}=\cfrac{\left(1+\lambda^{2}\right)^{L}-\left(1-\lambda^{2}\right)^{L}}{2}
\end{eqnarray}
in the last line. Since the phase of the state differs by a factor of $-1$ between $\theta = 0$ and $\theta = \pi$, the Berry phase in periodic systems $i\eta_{P}$, including this contribution, is
\begin{eqnarray}\label{berryphase}
e^{i\eta_{P}}&=&\exp(iL\pi\cdot\left(1-\cfrac{\left(1+\lambda^{2}\right)^{L-1}-\left(1-\lambda^{2}\right)^{L-1}}{\left(1+\lambda^{2}\right)^{L}-\left(1-\lambda^{2}\right)^{L}}\right)+i\pi),\\
&=&\exp(-iL\pi\cdot\cfrac{\left(1+\lambda^{2}\right)^{L-1}-\left(1-\lambda^{2}\right)^{L-1}}{\left(1+\lambda^{2}\right)^{L}-\left(1-\lambda^{2}\right)^{L}}+i\left(L+1\right)\pi).
\end{eqnarray}

Similarly, we can compute the Berry phase in anti-periodic system
\begin{eqnarray}\label{antiberryphase}
e^{i\eta_{AP}}&=&\exp(iL\pi\cdot\left(1-\cfrac{\left(1+\lambda^{2}\right)^{L-1}+\left(1-\lambda^{2}\right)^{L-1}}{\left(1+\lambda^{2}\right)^{L}+\left(1-\lambda^{2}\right)^{L}}\right)),\\
&=&\exp(-iL\pi\cdot\cfrac{\left(1+\lambda^{2}\right)^{L-1}+\left(1-\lambda^{2}\right)^{L-1}}{\left(1+\lambda^{2}\right)^{L}+\left(1-\lambda^{2}\right)^{L}}+iL\pi).
\end{eqnarray}
Note that since $\ket{\{A^{i}(\theta),1_{2}\}}$ is $2\pi$-periodic, there is not additional phase $i\pi$.

The difference of $\eta$s is
\begin{eqnarray}
i\eta_{P}-i\eta_{AP}&=&i\pi+i\cfrac{4\pi L\lambda^{2}(1+\lambda^{2})(1-\lambda^{2})}{(1+\lambda^{2})^{2L}-(1-\lambda^{2})^{2L}}
\end{eqnarray}
and since the second term converges to zero as $L$ to $\infty$, the ratio 
\begin{eqnarray}
\cfrac{e^{i\eta_{P}}}{e^{i\eta_{AP}}} 
\end{eqnarray}
exponentially converges to $-1$.

\bibliography{ref.bib}

\end{document}